\documentclass[11pt]{article}
\usepackage[square,authoryear]{natbib}
\usepackage{marsden_article}
\usepackage[all]{xy}
\usepackage{amscd}
\usepackage{amssymb}           
\usepackage{mathtools,mathrsfs}
\usepackage{framed}
\newtheorem{remark}[theorem]{Remark}
\usepackage{eucal}
\usepackage{url}
\usepackage{ulem}
\usepackage{framed}
\usepackage{graphicx}

\newcommand{\deltabar}{{\mkern0.75mu\mathchar '26\mkern -9.75mu\delta }}
\newcommand{\partialnew}{\textsf{\reflectbox{6}}}

\begin{document}

\title{General relativistic Lagrangian continuum theories\\
\medskip 
\Large Part II: electromagnetic fluids and solids with junction conditions}
\author{Fran\c{c}ois Gay-Balmaz\footnote{Division of Mathematical Sciences, Nanyang Technological University, 21 Nanyang Link, Singapore 637371; email: \url{francois.gb@ntu.edu.sg}.\;\;Data sharing not applicable to this article as no datasets were generated or analysed during the current study. The author is partially supported by a startup grant from Nanyang Technological University and by the Ministry of Education, Singapore, under Academic Research Fund (AcRF)Tier 1 Grant RG99/24.}}
\date{}
\maketitle

\begin{abstract}
We develop a covariant variational framework for relativistic electromagnetic continua (fluids and solid) based on Hamilton's principle formulated directly in the material description. The approach extends the geometric theory of relativistic continua introduced in Part I to include polarization, magnetization, and general elastic-electromagnetic coupling through a unified energy functional. By exploiting spacetime and material covariance, the framework yields the corresponding spacetime and convective variational principles, together with transparent expressions for the stress-energy-momentum tensor and the covariant Euler-type balance equations governing nonlinear electromagnetic continua.
Coupling to general relativity is naturally incorporated, and when the action is augmented with Gibbons-Hawking-York boundary terms, the gravitational and electromagnetic junction conditions follow directly from the variational formulation. 
The results provide a coherent foundation for modeling nonlinear electromagnetic continua in relativistic regimes, with relevance to astrophysical systems where relativistic continuum dynamics and electromagnetic fields are known to be strongly coupled, such as neutron-star crusts, magnetar flares, relativistic jets, and accretion disks. We also offer systematic connections with several formulations appearing in the existing literature.
\end{abstract}

\tableofcontents



\section{Introduction}

Several astrophysical systems of significant interest, such as neutron stars, accretion disks around black holes, and magnetized plasma involve a complex interplay between relativistic continuum mechanics and electromagnetic effects, see \cite{FrKiRa2002}, \cite{ShTe2008}, \cite{LeVP2012}, \cite{SB2020}, \cite{ShSu2022} for both introductory and advanced treatments on these topics. These systems frequently exhibit fluid-like behavior and, in some cases, elastic characteristics, particularly evident in the crustal dynamics of neutron stars, see \cite{HaPoYa2007}, \cite{ChHa2008} and references therein.

Motivated by the need to model such systems at the intersection of general relativity, continuum mechanics, and electromagnetism, this paper develops a Lagrangian and variational framework for relativistic electromagnetic continua in material, spacetime, and convective descriptions. Building on the formulation of relativistic continua in Part I (\cite{GB2024}), the framework is extended to include electromagnetic interactions, such as polarization and magnetization effects.

While several variational approaches for relativistic fluid and elasticity have been developed over the years from different perspectives and for diverse purposes, our approach stands out by its simplicity and strong physical justification, namely
\begin{itemize}
\item[\rm (i)] It originates from a continuum version of the Hamilton principle for the relativistic charged particle interacting with an electromagnetic field: see \eqref{material_L};
\item[\rm (ii)] It naturally incorporates the covariance properties of the theory, namely, spacetime covariance (a core axiom) and material covariance (related to system-specific properties like isotropy): see Figure \ref{figure_2}.
\end{itemize}
In particular, Hamilton’s principle is applied in a classical manner by varying the continuum’s spacetime configuration (world-tube) and the electromagnetic potential. This formulation prioritizes the material (trajectory) description of the continuum as the primary variational framework. 
From this, the Eulerian and convective variational descriptions of the continua are deduced by exploiting the material and spacetime covariance of the theory, as illustrated schematically in Figure \ref{figure_2}.
We believe this novel approach provides valuable guidance for modeling electromagnetic relativistic continua underlying the astrophysical phenomena discussed earlier. Key highlights of our approach are summarized below.

\begin{itemize}
\item[\rm (i)] \textbf{Stress-energy-momentum tensor:} It provides a direct and systematic derivation of the total stress-energy-momentum tensor for electromagnetic media. The literature often derives it through indirect or model-dependent arguments, leading to ambiguities in how the total tensor is partitioned between field and matter. Our variational approach constructs the full tensor cleanly and transparently, without relying on ad-hoc splittings.



\item[\rm (ii)] \textbf{Freedom from constraints:} It avoids a priori imposed Lagrange multiplier constraints and non-physical variables. This eliminates the often laborious task of determining multipliers, which lack direct physical relevance, reflecting constraints on chosen variables rather than the system itself.

\item[\rm (iii)] \textbf{Parallel to Newtonian mechanics:} It systematically parallels the Lagrangian variational framework of Newtonian continuum mechanics. Starting from Hamilton's principle applied to the fluid or solid configuration map, it deduces Eulerian or convective variational formulations by exploiting the symmetries of the continuum. In relativistic settings, these symmetries are either associated to spacetime covariance (a core axiom) or material covariance (related to system-specific properties like isotropy). In the electromagnetic case, another important symmetry emerges, the gauge invariance.

\item[\rm (iv)] \textbf{Geometric and electromagnetic junction conditions:} It facilitates the derivation of junction conditions between internal solutions for relativistic electromagnetic continua and external Einstein-Maxwell solutions, including Israel-Darmois and electromagnetic junction conditions. This is achieved by incorporating Gibbons-Hawking-York (GHY) boundary terms into the action functional and applying Hamilton’s principle to variations in the world-tube, metric, and electromagnetic potential.
\end{itemize}

Since their initial development for perfect fluids by \cite{Ta1954}, with subsequent development in special relativistic elasticity (\cite{Gr1971}, \cite{MaEr1972a,Ma1972a,MaEr1972b}), general relativistic elasticity (\cite{Ma1971}, \cite{Ca1973}), and relativistic fluids (\cite{CaKh1992}, \cite{CaLa1995,CaLa1998}), variational approaches have played a central role in theoretical modeling of relativistic continua. These variational formulations have since become indispensable tools in general relativistic continuum theories, as demonstrated in recent works by \cite{GaAnHa2020}, \cite{BaWa2020}, \cite{FeCa2020}, \cite{AdCo2021}, \cite{Ad2021}, \cite{CoAnCeHa2023} among others.

Regarding the inclusion of electromagnetic effects in relativistic continua, a foundational theory based on the principles of relativistic continuum mechanics, was established by \cite{GrEr1966b}. Subsequent developments and extensions were explored by \cite{Gr1971}, \cite{Ma1971,MaEr1972a,MaEr1972b,Ma1972a,Ma1978a,Ma1978b,Ma1978c,Ma1978d,Ca1980,ErMa1990}. We refer to \cite{CaChCh2006} and \cite{CaSa2006} for development in the context of neutron stars.

The equations we derive allow for a general covariant coupling of elastic deformations and electromagnetic fields through a unified energy functional depending on the Cauchy deformation tensor and the electromagnetic fields. Such a coupling between relativistic continuum mechanics and electromagnetic phenomena is relevant in various high-energy astrophysical contexts. For instance, magneto-elastic stresses have been proposed as a mechanism contributing to neutron-star crust failures and magnetar flares, as illustrated by \cite{LevinLyutikov2012}, \cite{PernaPons2011}, and numerical simulations by \cite{Gabler2012}.
Another important context is the launching of relativistic jets from compact objects, which is often modeled using force-free or general relativistic MHD frameworks (e.g.,  \cite{BlandfordZnajek1977}). More recent general relativistic MHD studies (e.g., \cite{Komissarov2004}, \cite{McKinney2006}, \cite{Tchekhovskoy_etal_2011}) highlight that jet formation relies on the coupled dynamics of matter and electromagnetic fields in curved spacetime. Our approach provides a complementary first-principles framework for examining how non-fluid matter, such as elastic or stratified media, responds to strong electromagnetic stresses in these extreme environments.
Relativistic accretion disks, where electromagnetic fields interact with fluid matter (e.g., \cite{AbFr2013}), offer another setting in which this approach could be applied.

Moreover, as we shall briefly comment, our framework naturally extends to the anisotropic case, allowing one to capture directional dependencies in the elastic response, which can be relevant for stratified neutron-star crusts, magnetically aligned matter, or other astrophysical media with directional microstructure.




\paragraph{Results of the paper.} Our contributions to the relativistic theories of electromagnetic fluids and solids can be summarized as follows:
\begin{itemize}
\item[(i)] \textsf{A natural first principle in the material (Lagrangian) description.} We develop a variational framework based directly on the most fundamental principle, Hamilton's principle, applied to the continuum analogue of the gauge invariant action functional of a charged particle interacting with an electromagnetic field:
\begin{equation}\label{material_L}
 \int_{ \lambda _0}^{ \lambda _1}    - mc \sqrt{- \mathsf{g}( \dot  x, \dot  x)} \,d \lambda - q A (x) \cdot  \dot  x\,  d \lambda - \frac{1}{2 c \mu _0} \int_ \mathcal{M} F \wedge \star F .
\end{equation}
This formulation is free of Lagrange multipliers or auxiliary variables. It involves only the primary fields of the world tube (the material configuration), the electromagnetic potential, and their unrestricted variations (and the metric when coupling to gravity).

\item[(ii)] \textsf{Covariance leads to spacetime and convective Lagrangians and principles.} We show that imposing material and spacetime covariance in the material picture naturally produces the corresponding spacetime-covariant and convective-covariant reduced Lagrangians and their associated variational principles, with constrained variations induced from the free variations of the primary fields.

\item[(iii)]  \textsf{Treatment of general covariant coupling of elastic deformations and electromagnetic fields.} We show that our continuum version of the action functional  \eqref{material_L} can accommodate arbitrary constitutive behavior for electromagnetic fluids and solids, as long as it is encoded in a covariant energy\footnote{Its exact relation with the physical energy density will be explained.} density $\epsilon(\rho,s, E, B, \mathsf{c}, \mathsf{g})$, where $\rho$ and $s$ are the proper mass and entropy densities, $E$ and $B$ the electric and magnetic fields, $\mathsf{c}$ the relativistic Cauchy deformation tensor, and $\mathsf{g}$ the spacetime metric. This shows that the framework naturally incorporates any covariant elastic-electromagnetic coupling without introducing additional variables or constraints.

\item[(iv)]  \textsf{General form of the stress-energy-momentum tensor.} We show that out variational setting efficiently yields the expression for the stress-energy-momentum tensor for general covariant coupling of elastic deformations and electromagnetic fields. The resulting tensor is
\begin{align*} 
\mathfrak{T}&= \Big[\Big(\epsilon   - \frac{\partial \epsilon   }{\partial E} \cdot E \Big) \frac{1}{c^2} u \otimes u ^\flat  + \frac{1}{c^2} \Big( u \otimes S_ \epsilon   + S^\sharp_ \epsilon    \otimes u ^\flat \Big)\\
& \qquad   +  \frac{\partial \epsilon    }{\partial E  } \otimes E - B ^\sharp  \stackrel{\rm tr}{ \otimes }  \frac{\partial \epsilon    }{\partial B  }^\flat    + \Big(   \frac{\partial \epsilon   }{\partial \rho  } \rho  + \frac{\partial \epsilon   }{\partial s} s + \frac{\partial \epsilon    }{\partial B  } : B- \epsilon \Big) \mathsf{P}+ 2  \frac{\partial \epsilon   }{\partial \mathsf{c}  } \cdot\mathsf{c}\Big]\mu ( \mathsf{g} ),
\end{align*} 
with $u$ the world-velocity, $S_ \epsilon = (-1)^n \mathbf{i} _{E^\sharp} \mathbf{i} _u \star\frac{\partial \epsilon }{\partial B}^\flat$ the Poynting one-form and $\mathsf{P}$ the projection tensor. 
We discuss how this general expression compares with the classical stress-energy-momentum tensors for electromagnetic media, which are most often derived from a fundamentally different line of reasoning, typically through microscopic or mesoscopic interaction models, such as those developed by \cite{dGSu1972}. 

\item[(v)] \textsf{Electromagnetic junction conditions via Gibbons-Hawking-York terms.} For electromagnetic fluids and solids coupled to general relativity, we show that the electromagnetic and gravitational junction conditions between interior and exterior fields emerge directly from the variational formulation when the action is augmented with the Gibbons-Hawking-York boundary term arising from the trace of the extrinsic curvature. Along the boundary of the spacetime region occupied by the continuum, we obtain the following boundary and jump conditions:
\begin{equation}\label{BC_EMF_2} 
\begin{aligned} 
&\mathsf{g} (u , n) =0, \qquad & &  [ \mathfrak{t}] ( \cdot , n ^\flat )+ [p] n ^\flat =0,& \\
&\mathbf{i} _n \mathbf{i} _u (\star [E])=0, \qquad & & \mathbf{i} _n [B]=0 ,&\\
&\mathbf{i} _n [D]=0, \qquad & & \mathbf{i} _n \mathbf{i} _u (\star [H])=0,&
\end{aligned}
\end{equation} 
with $[\cdot ]$ denoting the jump across the boundary, and where $[\mathfrak{t}]$ and $[p]$ are the jumps in the stress tensor and the isotropic pressure parts of the stress-energy-momentum tensor, respectively.
\end{itemize}


\paragraph{Comparison with the literature.} Comparison with the existing literature on relativistic electromagnetic continua can be challenging due to the variety of approaches used to derive the stress-energy-momentum tensor and the different notations employed. For this reason, in  \S\ref{EMFS} we devote considerable attention to presenting several equivalent forms of the stress-energy-momentum tensor for electromagnetic media, as well as several possible splittings of this tensor that have been considered in the literature (see especially \S\ref{SEM_fluid} and \S\ref{SEM_el}). We also provide multiple forms of the Euler-type balance equations for relativistic electromagnetic continua (see  \S\ref{subsubsec_equ} and \S\ref{equ_el}), which can be difficult to recognize as equivalent at first sight. This allows the reader to compare our framework directly with different approaches previously proposed in the field.

\paragraph{Plan of the paper.} In  \S\ref{sec_2}, we develop the geometric framework for formulating Hamilton’s principle in the material picture (\S\ref{Geom_setting}), discuss the two covariance properties, and introduce the associated reduced Lagrangian densities - both convective and spacetime (\S\ref{cov_prop}). We then present the corresponding variational principles and their reduced Euler-Lagrange equations (\S\ref{spacetime_reduction} and \S\ref{convective_reduction}). In \S\ref{material_spacetime_A}, we show that, for the class of Lagrangian densities relevant to electromagnetic continua, one may take either the spacetime or the material electromagnetic potential as the primary field.

In \S\ref{coupling} we show how the variational formulation couples to gravitation, in particular to the Einstein-Maxwell equations outside the continuum. We derive the field equations for the gravitational and electromagnetic fields inside and outside the continuum, the equations of motion of the continuum, and the junction conditions linking the interior solution to the exterior gravitational and electromagnetic fields.

In \S\ref{EMFS}, we apply the variational framework to electromagnetic fluids (\S\ref{GR_em_fluid}) and solids (\S\ref{sec_GREE}). We emphasize both the form of the stress-energy-momentum tensor (providing three equivalent descriptions) and the resulting relativistic Euler-type balance equations for general covariant coupling of the elastic deformation and electromagnetic fields, as well as for simpler illustrative examples. We also discuss relation to the existing literature and the explicit form of the junction conditions. The papers ends with a conclusion in \S\ref{conclusion}.

\section{Covariance properties and variational principles for electromagnetic continua}\label{sec_2}

The variational framework we develop makes explicit the dependence of the Lagrangian on the material and spacetime tensor fields, all transported along the world-tube of the continuum via push-forward and pull-back operations. This structure allows the covariance properties to be expressed transparently through the actions of spacetime and material diffeomorphisms on the material Lagrangian density. Our formulation extends to the electromagnetic setting the approach introduced in  \cite{GB2024}. In addition, the electromagnetic case naturally incorporates gauge invariance as an extra symmetry.

\subsection{Geometric setting}\label{Geom_setting}


The key objects required to describe the motion of a relativistic continuum are the world-tube, which characterizes the spacetime locations of the medium's particles, and a set of material and spacetime tensors that describe the various fields involved, see Figure \ref{fig_word_tube}.

\paragraph{World-tube and velocities.} The relativistic motion of a continuous media in a $(n+1)$-dimensional spacetime $\mathcal{M} $ with Lorentzian metric $\mathsf{g}$ is described by an embedding
\[
\Phi : \mathcal{D} = [ \lambda _0, \lambda _1] \times \mathcal{B} \rightarrow \mathcal{M},
\]
the \textit{world-tube}, where $ \mathcal{B} $ a $n$-dimensional compact orientable manifold with smooth boundary and where $\mathsf{g}( \partial _ \lambda \Phi , \partial _ \lambda \Phi )<0$, for all $ \lambda \in  [ \lambda _0, \lambda _1]$.
The \textit{generalized velocity} of the media is the vector field on $ \mathcal{N} = \Phi( \mathcal{D} )$ defined as
\begin{equation}\label{gen_vel_1}
w=  \partial_ \lambda  \Phi  \circ \Phi ^{-1} \in \mathfrak{X} (\mathcal{N} ).
\end{equation} 
Its normalized version is the \textit{world-velocity} defined by
\[
u= \frac{c}{\sqrt{-\mathsf{g}(w,w)}} w\in \mathfrak{X} ( \mathcal{N} ), \qquad \mathsf{g}(u,u)=- c ^2 ,
\]
with $c$ the speed of light. Note that we can write the generalized velocity as $w= \Phi _* \partial _ \lambda $, where $\Phi_*$ denotes the push-forward of the vector field $ \partial _ \lambda $  by $\Phi$.

\paragraph{Reference material and spacetime tensor fields.} The vector field $ \partial _ \lambda \in \mathfrak{X} ( \mathcal{D} )$ is the first example of a \textit{reference material tensor field}.
Other classical examples of reference material tensor fields are given by volume forms $R, S \in \Omega ^{n+1}( \mathcal{D} )$ satisfying
\begin{equation}\label{condition_R_S} 
\pounds _{ \partial _ \lambda } R = \pounds _{\textcolor{black}{\partial_ \lambda }} S =0.
\end{equation}
These two $(n+1)$-forms correspond to the reference mass density and entropy density and are chosen as
\begin{equation}\label{RR_0SS_0} 
R= d \lambda  \wedge \pi _ \mathcal{B} ^* R_0 \quad\text{and}\quad S= d \lambda  \wedge \pi _ \mathcal{B} ^* S_0,
\end{equation}
where $R_0, S_0 \in \Omega ^n( \mathcal{B} )$ are volume forms on $ \mathcal{B} $, called the mass form and entropy form, and $ \pi _ \mathcal{B} :[ \lambda _0, \lambda _1] \times \mathcal{B} \rightarrow \mathcal{B} $ is the projection. One can check that properties \eqref{condition_R_S} hold for $R,S$ given in \eqref{RR_0SS_0}. The corresponding Eulerian quantities are  the generalized mass density $ \varrho $, and generalized entropy density $ \varsigma $ obtained by applying the push-forward operation by the world-tube:
\[
\varrho = \Phi _* R, \qquad \varsigma= \Phi_*S.
\]
They should not be confused with the proper mass and entropy densities which will be defined later. Another reference material tensor field is given the $2$-covariant symmetric positive tensor $G= \pi _ \mathcal{B} ^* G_0$ with $G_0$ a Riemannian metric on $ \mathcal{B} $, the corresponding Eulerian quantity being the relativistic Cauchy deformation tensor
\[
\mathsf{c}= \Phi _*G.
\]

The list of \textit{spacetime tensor fields} necessarily includes the Lorentzian metric $\mathsf{g}$, which may either be considered as fixed, e.g., in the special relativistic case, or subject to its associated Euler-Lagrange equations in the self-gravitating case.

In this section, in order to stay as general as possible, we shall consider reference tensor fields given by a material $(r,s)$-tensor field $K \in \mathcal{T} ^ r_s( \mathcal{D} )$ and a spacetime $(p,q)$-tensor field by $ \gamma \in \mathcal{T} ^p_q( \mathcal{M} )$, where it is assumed that $ \pounds _{ \partial _ \lambda } K=0$. The extension to a collection of such tensor fields is straightforward. The corresponding push-forward and pull-back tensor fields are denoted
\[
\kappa = \Phi _* K \in \mathcal{T} ^r_s( \mathcal{N} ) , \qquad \Gamma = \Phi ^* \gamma \in \mathcal{T} ^p_q( \mathcal{D}).
\] 
In this general setting $ \gamma $ may correspond to $\mathsf{g}$ with $ \Gamma $ given by its pull-back on $ \mathcal{D} $, while $K$ may correspond to $R$, $S$, or $G$, with $ \kappa $ given by $ \varrho $, $ \varsigma $, or $\mathsf{c}$. Since the velocity plays a special role, we shall explicitly include the corresponding reference vector field, which we denote $W= \partial _ \lambda $. Hence we have $w= \Phi _*W$ and we have the assumption $\pounds _WK=0$.

\paragraph{Material and spatial electromagnetic potentials.} Besides these tensors, electromagnetic continua involve the electromagnetic potential given by a one-form $A \in \Omega ^1 ( \mathcal{N} )$ on spacetime, whose exterior derivative $F= {\rm d} A \in \Omega ^2( \mathcal{N} )$ is the Faraday $2$-form. The corresponding material tensors are obtained by taking the pull-back with respect to the world-tube, written as $ \mathcal{A} = \Phi ^* A\in \Omega ^1 ( \mathcal{D} )$ and $ \mathcal{F} = \Phi ^* F = {\rm d} \mathcal{A} \in \Omega ^2 ( \mathcal{D} )$. As opposed to the reference material tensors $W$, $R$, $S$, and $G$ discussed earlier, $ \mathcal{A} $ and $ \mathcal{F} = {\rm d} \mathcal{A} $ are not fixed since $\mathcal{A} $ is subject to its associated Euler-Lagrange equations. As we shall show later on, we can consider either $ \mathcal{A} $ or $A$ as the primitive object in the Hamilton principle in the material description, while obtaining equivalent critical point conditions, i.e., the same equations, although in a quite different form, see \S\ref{material_spacetime_A}. We shall put the emphasis on $ \mathcal{A} $, since with this choice the resulting variational setting does reduce quite naturally to the one corresponding to the magnetohydrodynamic approximation. When coupling the electromagnetic continua with gravity and electromagnetism at the outer spacetime via the junction conditions, it will be more practical to use $A$, see \S\ref{coupling}.

\begin{figure}[t]
\centering
\includegraphics[scale=0.4]{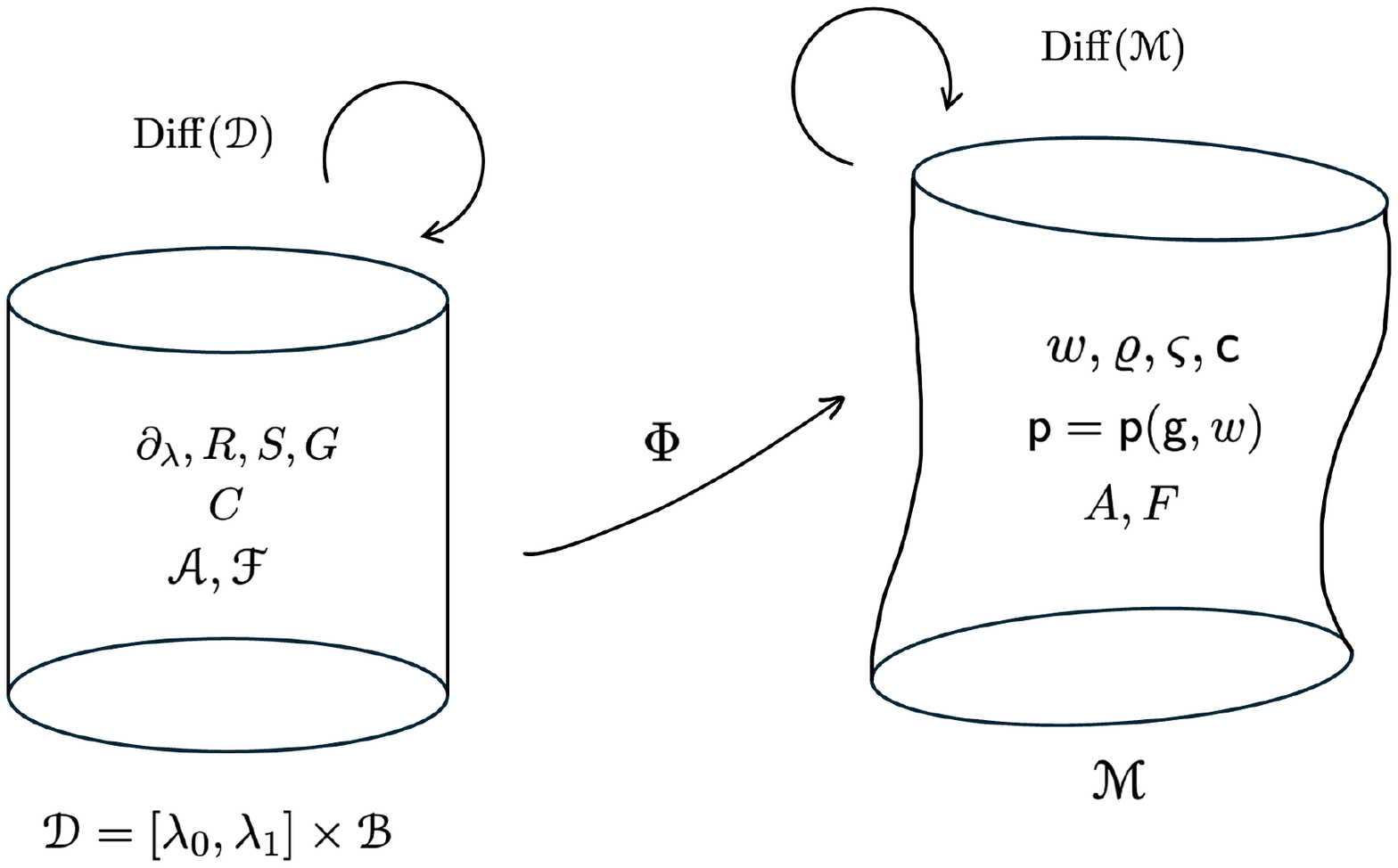}
\caption{Illustration of the word-tube $\Phi:[\lambda_0,\lambda_1]\times\mathcal{B}\rightarrow\mathcal{M}$ and the tensor fields involved in the description of electromagnetic continua. The situation is summarized as follows:\\
-- The fixed reference material tensor fields ($K$ in the general theory) include the vector field $\partial_\lambda$, the $(n+1)$-forms $R$ and $S$, and the $2$-covariant symmetric positive tensor $G$. Their push-forward by $\Phi$ are $w$, $\varrho$, $\varsigma$, and $\mathsf{c}$, respectively.\\
-- The spacetime tensor field ($\gamma$ in the general theory) is $\mathsf{g}$. This tensor becomes dynamic only when the theory is coupled with Einstein's equations (\S\ref{coupling}). From $\mathsf{g}$ and $w$, the $2$-covariant symmetric positive tensor $\mathsf{p}$ is constructed. Its pull-back by $\Phi$ yields the right Cauchy-Green tensor $C$.\\
-- Besides $\Phi$, the other dynamic field can be chosen as either the material electromagnetic potential $\mathcal{A}$ or its spacetime version $A$, obtained by pushing $\mathcal{A}$ forward with $\Phi$. The corresponding electromagnetic fields are $\mathcal{F}={\rm d}\mathcal{A}$ and $F={\rm d}A$.}
\label{fig_word_tube}
\end{figure}

\subsection{Covariance properties and reduced Lagrangians}\label{cov_prop}

We develop here the general setting for the variational principles and symmetries underlying electromagnetic continua. In particular, we present the class of Lagrangian densities in the \textit{material description} relevant to these models as well as the notions of material and spacetime covariance in this general setting. These symmetries allow the definition of the associated Lagrangian densities in the \textit{spacetime (Eulerian) descriptions} or in the \textit{convective description}. We regard the material description, i.e., the description of the continuum in terms of the world-tube, as the fundamental description, and the other ones as being deduced from it. This is consistent with the fact that in the material description the variational principle is the classical Hamilton principle, naturally arising as the continuum version of the Hamilton principle for relativistic charged particles.

\paragraph{Lagrangian density and Hamilton principle in the material description.} The Lagrangian densities describing electromagnetic continuous media in the material description are bundle maps of the form
\begin{equation}\label{material_Lagrangian} 
\mathscr{L} : J^1( \mathcal{D}  \times \mathcal{M} ) \times \wedge  ^1T^* \mathcal{D} \times \wedge ^2 T^* \mathcal{D}  \times T \mathcal{D} \times T^p_q\mathcal{D}  \times T^r_s \mathcal{M}  \rightarrow \wedge ^{n+1}T^* \mathcal{D} \footnote{Note the change of notation from $ \wedge ^k \mathcal{D} \rightarrow \mathcal{D} $ to $ \wedge ^k T^* \mathcal{D}\rightarrow \mathcal{D} $ for the bundle of $k$-forms, compared to \cite{GB2024}, to avoid confusion with the bundle $\wedge ^k T \mathcal{D}  \rightarrow \mathcal{D} $ of $k$-multivectors also occurring in the present paper.}
\end{equation}
covering the identity on $ \mathcal{D} $. Here $J^1( \mathcal{D}  \times \mathcal{M} ) \rightarrow \mathcal{D}  \times \mathcal{M} $ denotes the first jet bundle of the trivial fiber bundle $\mathcal{D}  \times \mathcal{M}  \rightarrow \mathcal{D} $. The vector fiber of $J^1( \mathcal{D}  \times \mathcal{M} )$ at $(X,x) \in \mathcal{D} \times M$ is $L(T_X \mathcal{D} , T_x \mathcal{M} )$, the space of linear maps $T_X \mathcal{D} \rightarrow T_x \mathcal{M} $.
Also, $T^p_q \mathcal{D}  \rightarrow \mathcal{D} $ denotes the vector bundle of $(p,q)$-tensors on $\mathcal{D} $, similarly for $T^r_s \mathcal{M}  \rightarrow \mathcal{M} $. 
In local coordinates $X^a$, $a=0,...,n$ and $x ^ \mu $, $ \mu =0,...,n$, we shall write the Lagrangian density $\mathscr{L}$ as
\[
\mathscr{L} =\bar{ \mathscr{L} }\big(  X^a, x^\mu, v^\mu_a, \mathcal{A} _a, \mathcal{F} _{ab},  W^a, K^{a_1...a_p}_{b_1...b_q}, \gamma^{\mu_1...\mu_r}_{\nu_1...\nu_s} \big) {\rm d} ^{n+1} X.
\]
When evaluating the Lagrangian density along the world tube $ \Phi  $ and the fields $ \mathcal{A}$, $K$, $ \gamma $, we shall use the notation 
\begin{equation}\label{Lagrangian_general} 
\mathscr{L}(j^1\Phi, \mathcal{A} , {\rm d}\mathcal{A} , W, K, \gamma \circ \Phi  ),
\end{equation}
where $j^1\Phi: \mathcal{D}  \rightarrow J^1( \mathcal{D}  \times \mathcal{M} )$ denotes the first jet extension of the world-tube $\Phi$, given locally as $X^a \mapsto (X ^a , \Phi ^\mu(X^a), \Phi ^\mu_{,b}(X^a))$, with $ \Phi ^\mu_{,b}:= \partial _{X^b}\Phi ^ \mu$.

For each $K$ and $ \gamma $ fixed, Hamilton's principle reads
\[
\left. \frac{d}{d\varepsilon}\right|_{\varepsilon=0} \int_ \mathcal{D} \mathscr{L}(j^1 \Phi _ \varepsilon ,  \mathcal{A}_ \varepsilon , {\rm d} \mathcal{A} _ \varepsilon , W, K , \gamma  \circ \Phi _ \varepsilon )=0,
\]
for arbitrary variations $ \Phi  _ \varepsilon $ and $ \mathcal{A} _ \varepsilon $ with fixed endpoints at $ \lambda =  \lambda _0, \lambda _1$.

This principle yields the Euler-Lagrange equations and boundary conditions written in local coordinates as
\begin{equation}\label{EL_material}
\begin{aligned} 
&\partial _a \frac{\partial \bar{\mathscr{L}}}{\partial \Phi^\mu_{,a}}  - \frac{\partial \bar{\mathscr{L}}}{\partial \Phi^\mu} = \frac{\partial \bar{\mathscr{L}}}{\partial \gamma } \partial_\mu \gamma \quad\text{on}\quad \mathcal{D} , \quad \frac{\partial \bar{ \mathscr{L} }}{\partial \Phi^\mu_{,a}}d^nX_a=0\quad \text{on} \quad  \textcolor{black}{[\lambda _0, \lambda _1] \times \partial \mathcal{B}}\\
& \partial _a\frac{\partial \bar{\mathscr{L}}}{\partial \mathcal{F} _{ab}} - \frac{\partial \bar{\mathscr{L}}}{\partial \mathcal{A}  _b}=0\quad\text{on}\quad \mathcal{D} , \qquad \quad  \quad  \frac{\partial  \bar{\mathscr{L}} }{\partial \mathcal{F} _{ab}}   d^nX_ a=0\quad \text{on} \quad  \textcolor{black}{[\lambda _0, \lambda _1] \times \partial \mathcal{B}},
\end{aligned}
\end{equation}
where $ {\rm d} ^nX_a= \mathbf{i} _{ \partial _a}{\rm d} ^{n+1}X$ and $\mathcal{F} = \frac{1}{2} \mathcal{F} _{ab} {\rm d} X^a \wedge {\rm d} X^b$ with $\mathcal{F} _{ab}= \partial _a \mathcal{A} _{b}- \partial _b \mathcal{A} _{a}$.

\paragraph{Partial derivatives.} For our next development, it is useful to give two interpretations of the partial derivatives of Lagrangian densities with respect to $ \mathcal{A} $ and $ \mathcal{F} $. We shall first regard them as $k$-vector field densities, $k=1,2$, denoted
\begin{equation}\label{first_incarnation} 
\frac{\partial  \mathscr{L}}{\partial \mathcal{A}} \in \mathfrak{X} ^1_d( \mathcal{D} ) \quad\text{and}\quad \frac{\partial \mathscr{L}}{\partial \mathcal{F}} \in  \mathfrak{X} ^2_d( \mathcal{D} )
\end{equation} 
with $ \mathfrak{X} ^k_d( \mathcal{D} ):= \Gamma ( \wedge ^kT  \mathcal{D} \otimes \wedge ^{n+1} T^* \mathcal{D}  )$, and defined as
\begin{equation}\label{partial2}
\begin{aligned} 
\left. \frac{d}{d\varepsilon}\right|_{\varepsilon=0} \mathscr{L}( ..., \mathcal{A}+ \varepsilon \delta \mathcal{A}, F) &= \delta \mathcal{A} \cdot  \frac{\partial \mathscr{L}}{\partial \mathcal{A}} \\
\left. \frac{d}{d\varepsilon}\right|_{\varepsilon=0} \mathscr{L}( ..., \mathcal{A}, \mathcal{F}+ \varepsilon \delta \mathcal{F})& = \delta \mathcal{F} : \frac{\partial \mathscr{L}}{\partial \mathcal{F}}.
\end{aligned}.   
\end{equation} 
In \eqref{partial2}, the symbols $`` \cdot "$ and $``:"$ denote the contraction, locally given by $ \mathcal{A} \cdot \mathcal{X} = \mathcal{A} _a \mathcal{X} ^a {\rm d} X^{n+1}$ and $\mathcal{F}: \mathcal{Y} = \frac{1}{2}  \mathcal{F}_{ab} \mathcal{Y} ^{ab} {\rm d} X^{n+1}$, for $\mathcal{X}= \mathcal{X}^a \partial _a \otimes {\rm d} ^{n+1}X$, $\mathcal{Y}= \frac{1}{2} \mathcal{Y}^{ab} \partial _a \wedge \partial _b \otimes {\rm d} ^{n+1}X$. The local expressions are found as
\[
\frac{\partial \mathscr{L}}{\partial  \mathcal{A}} = \frac{\partial \bar{\mathscr{L}}}{\partial \mathcal{A}_{a}} \partial _a \otimes  {\rm d} ^{n+1}X  \quad\text{and}\quad  \frac{\partial \mathscr{L}}{\partial \mathcal{F}} = \frac{1}{2} \frac{\partial \bar{\mathscr{L}}}{\partial \mathcal{F}_{ab}} \partial _a \wedge \partial _b \otimes {\rm d} ^{n+1}X.
\]
 
In order to write the Euler-Lagrange equations and boundary conditions for the electromagnetic potential $\mathcal{A}$ in intrinsic form, it is also useful to interpret the partial derivatives as $(n+1-k)$-forms, $k=1,2$, denoted
\begin{equation}\label{secondincarnation}
\frac{\partialnew \mathscr{L}}{\partialnew \mathcal{A}} \in \Omega ^n( \mathcal{D} ) \quad\text{and}\quad \frac{\partialnew \mathscr{L}}{\partialnew \mathcal{F}} \in \Omega  ^{n-1}( \mathcal{D} )\footnote{Note the different font used for the symbol $\partialnew$, compared to $ \partial $ in \eqref{first_incarnation} and \eqref{partial2}.}
\end{equation}
and defined as
\begin{equation}\label{partial1} 
\begin{aligned}
\left. \frac{d}{d\varepsilon}\right|_{\varepsilon=0} \mathscr{L}( ..., \mathcal{A}+ \varepsilon \delta \mathcal{A}, \mathcal{F} , ...) &= \delta \mathcal{A} \wedge \frac{\partialnew \mathscr{L}}{\partialnew \mathcal{A}} \\
\left. \frac{d}{d\varepsilon}\right|_{\varepsilon=0} \mathscr{L}( ..., \mathcal{A}, \mathcal{F}+ \varepsilon \delta \mathcal{F}, ...) &= \delta \mathcal{F} \wedge \frac{\partialnew \mathscr{L}}{\partialnew \mathcal{F}},
\end{aligned}
\end{equation} 
for all $ \delta \mathcal{A}$ and $ \delta \mathcal{F}$. The local expressions are now
\begin{equation}\label{local_partial1} 
\frac{\partialnew \mathscr{L}}{\partialnew \mathcal{A}} = \frac{\partial \bar{\mathscr{L}}}{\partial \mathcal{A}_{a}} {\rm d} ^nX_a  \quad\text{and}\quad  \frac{\partialnew \mathscr{L}}{\partialnew \mathcal{F}} = \frac{1}{2} \frac{\partial \bar{\mathscr{L}}}{\partial \mathcal{F}_{ab}}{\rm d} ^{n-1}X_{ab},
\end{equation} 
where $d^nX_a= \mathbf{i}  _{ \partial _a} d^{n+1}X$ and $d^{n-1}X_{ab}= \mathbf{i} _{ \partial _b} \mathbf{i} _{ \partial _a} d^{n+1}X$.
With these definitions, the Euler-Lagrange equations and boundary conditions for $\mathcal{A}$ are intrinsically written as
\[
\frac{\partialnew \mathscr{L}}{\partialnew \mathcal{A}} + {\rm d}  \frac{\partialnew \mathscr{L}}{\partialnew \mathcal{F}}=0, \qquad i^*_{ \partial  \mathcal{D} } \frac{\partialnew \mathscr{L}}{\partialnew \mathcal{F}}=0 \quad \text{on} \quad  \textcolor{black}{[a,b] \times   \partial \mathcal{B}}.
\]

One can switch from one to the other description by using the canonical vector bundle isomorphisms
\begin{equation}\label{isomorphism}
\mathcal{X} \in \wedge ^k T  \mathcal{D} \otimes \wedge ^{n+1} T^* \mathcal{D}  \rightarrow \omega  \in \wedge ^\ell T ^* \mathcal{D}, \quad \ell= n+1-k,
\end{equation}
given in coordinates by
\begin{align*} 
&\frac{1}{k!} \mathcal{X}  ^{a_1...a_k}  \partial _ {a_1} \wedge ... \wedge \partial _{a_k} \!\otimes {\rm d} ^{n+1} X \mapsto  \frac{1}{k!} \mathcal{X}  ^{a_1...a_k}  {\rm d} ^{\ell}X_{a_1...a_k}= \frac{1}{\ell!}\omega _{ b _1... b _\ell}{\rm d} X^{ b _1} \wedge ... \wedge {\rm d} X^{ b _\ell}
\end{align*} 
with $\omega _{ b _1... b _\ell}= \frac{1}{k!}\mathcal{X}  ^{a_1...a_k} \epsilon _{a_1...a_kb _1... b _\ell} $, $\mathcal{X}  ^{a_1...a_k} = \frac{1}{\ell!} \omega _{ b _1... b _\ell} \epsilon ^{a_1...a_kb _1... b _\ell}$, and we used the notation ${\rm d} ^{\ell}X_{a_1...a_k}= \mathbf{i} _{ \partial _{a_1} \wedge ... \wedge \partial _{a_k}} {\rm d} ^{n+1}X$.
The isomorphism \eqref{isomorphism} can be written in terms of a Lorentzian metric $\mathsf{g}$ as
\[
\mathcal{X}  = -(-1)^{k\ell} ( \star \omega  )^\sharp \otimes \mu (\mathsf{g}),
\]
see \cite{AbGBYo2024} for more details. Here $ \mu ( \mathsf{g})$ and $\star$ are the volume form and Hodge star operator associated to $\mathsf{g}$ and a given orientation of $ \mathcal{M} $, and $\sharp$ is the sharp operator associated to $\mathsf{g}$.

\paragraph{Spacetime covariance.} We say that the material Lagrangian density $\mathscr{L}$ in \eqref{material_Lagrangian}  is \textit{spacetime covariant} if it satisfies
\begin{equation}\label{spacetime_cov}
\begin{aligned} 
&\mathscr{L}(j^1( \psi \circ  \Phi ), \mathcal{A} , {\rm d} \mathcal{A} , W, K,\psi _* \gamma  \circ \psi \circ \Phi )\\
& \qquad    = \mathscr{L}(j^1 \Phi , \mathcal{A} , {\rm d} \mathcal{A} ,W, K, \gamma  \circ \Phi ), \quad \forall \psi \in \operatorname{Diff}( \mathcal{M} ).
\end{aligned}
\end{equation}  
In this case $\mathscr{L}$ induces a unique \textit{reduced convective Lagrangian density}
\[
\mathcal{L}: \wedge ^1 T^* \mathcal{D} \times \wedge ^2 T^* \mathcal{D}  \times T \mathcal{D} \times T^p_q\mathcal{D}  \times T^r_s \mathcal{D}  \rightarrow \wedge ^{n+1}T^* \mathcal{D} 
\]
defined by the condition
\begin{equation}\label{def_calL}
\mathscr{L}(j^1 \Phi , \mathcal{A} ,{\rm d}\mathcal{A} ,  W, K, \gamma  \circ \Phi )= \mathcal{L} \big(\mathcal{A} , {\rm d} \mathcal{A},  W, K, \Gamma )
\end{equation} 
for all $ \Phi $, $ \mathcal{A} $, $W$, $K$, $ \gamma $, where
\[
\Gamma := \Phi ^* \gamma \in \mathcal{T} ^r_s( \mathcal{D} ).
\]

\paragraph{Material covariance.} We say that the material Lagrangian density $\mathscr{L}$ in \eqref{material_Lagrangian}  is \textit{materially covariant} if it satisfies
\begin{equation}\label{material_cov}
\begin{aligned} 
&\mathscr{L}(j^1( \Phi   \circ  \varphi  ), \varphi ^* \mathcal{A} , \varphi ^* {\rm d} \mathcal{A} ,  \varphi ^*W, \varphi ^* K,\gamma   \circ \Phi \circ \varphi )\\
&\qquad   = \varphi ^* \big[\mathscr{L}(j^1 \Phi , \mathcal{A} , {\rm d} \mathcal{A},W, K, \gamma  \circ \Phi )\big], \quad \forall \varphi  \in \operatorname{Diff}( \mathcal{D} ).
\end{aligned}
\end{equation} 
In this case  $\mathscr{L}$ induces a unique \textit{reduced spacetime Lagrangian density}
\[
\ell: \wedge ^1 T^* \mathcal{M} \times \wedge ^2 T^* \mathcal{M} \times  T \mathcal{M} \times T ^p _q  \mathcal{M}  \times  T^r _s  \mathcal{M} \rightarrow \wedge  ^{n+1} T^*\mathcal{M}
\]
such that
\begin{equation}\label{def_ell} 
\mathscr{L}(j^1 \Phi ,\mathcal{A} ,{\rm d} \mathcal{A} ,  W, K, \gamma  \circ \Phi )= \Phi ^* \big[\ell( A, {\rm d} A, w,\kappa , \gamma )\big],
\end{equation} 
for all $ \Phi $, $ \mathcal{A} $, $W$, $K$, $ \gamma $, where
\[
A= \Phi _*\mathcal{A} , \quad F= \Phi _*\mathcal{F} , \quad w= \Phi _*W, \quad \kappa = \Phi _*K.
\]
The existence of the reduced Lagrangian densities $ \mathcal{L} $ and $\ell$ under the corresponding covariance can be shown as in \cite{GB2024}.

\paragraph{Gauge invariance.} Besides the spacetime and material covariance properties of the Lagrangian densities under the diffeomorphism groups $ \operatorname{Diff}( \mathcal{D} )$ and $ \operatorname{Diff}( \mathcal{M} )$, the theory (not necessarily the Lagrangian density itself) should also exhibit invariance under gauge transformations. In the material description, as defined by the Lagrangian density \eqref{Lagrangian_general}, these gauge transformations are given by the action of the Abelian group $C^\infty( \mathcal{D} )$ on the bundle $ \wedge ^1T^* \mathcal{D} \rightarrow \mathcal{D} $ of electromagnetic potentials, given by $ \mathcal{A}_X \mapsto \mathcal{A}_X + {\rm d} f(X)$, $ f \in C^\infty( \mathcal{D} )$. Let us consider the action functional
\begin{equation}\label{action_functional}
\mathscr{S}( \Phi ,\mathcal{A})= \int_{ \lambda _0}^ { \lambda _1}\int_ \mathcal{B} \mathscr{L}(j^1\Phi, \mathcal{A} , {\rm d} \mathcal{A} , W, K, \gamma \circ \Phi  ),
\end{equation}
for some fixed $W$, $K$, $\gamma$.
Gauge invariance is satisfied if
\begin{equation}\label{gauge_invariance} 
\mathscr{S}( \Phi ,\mathcal{A}+ {\rm d} f)= \mathscr{S}( \Phi ,\mathcal{A}) + C,
\end{equation}
for all $f \in C^\infty( \mathcal{D} )$, where $C$ may depend on quantities that are not varied when computing the critical conditions, such as $f$, $W$, $K$, or the endpoints values $ \Phi ( \lambda _0, X)$ and $ \Phi ( \lambda _1, X)$. Consequently, $\mathcal{A}$ is a critical field if and only if  $\mathcal{A}+{\rm d}f$ is also critical field. Recall that the critical field conditions include both the Euler-Lagrange equations and the boundary conditions. We shall see in \S\ref{gauge_continua} that for the Lagrangian densities of electromagnetic fluids and solids, \eqref{gauge_invariance} is satisfied. 

When combining these gauge transformations with the covariance properties under material and spacetime diffeomorphisms, the total invariance of the theory can be described by the group $\left( \operatorname{Diff}( \mathcal{D} ) \,\circledS\, C^\infty( \mathcal{D} ) \right)  \times \operatorname{Diff}( \mathcal{M} )\ni (\varphi, f, \psi)$, where  $\circledS$ denotes the semidirect product with group multiplication $( \varphi _1,f_1)(\varphi _2,f_2)= (\varphi _1 \circ \varphi _2, f_1 \circ \psi_2 + f_2)$, see also Remark \ref{remark_GI}. For instance, it acts on the electromagnetic potential as $\mathcal{A}\mapsto \varphi^*\mathcal{A}+ {\rm d} f$ and on the world-tube as $\Phi\mapsto \psi^{-1}\circ\Phi\circ\varphi$. While gauge invariance is most naturally expressed in the material description, it can also be formulated in the spacetime or convective descriptions. Details of these formulations are left to the interested reader.

\paragraph{Technical results on Lie derivatives and divergences.}
We recall below a technical result that we will crucially use for the derivation of the Euler-Lagrange equations, see  \cite[Lemma 4.1]{GB2024}. The local expression of the Lie derivative of a $(p,q)$-tensor field $ \kappa $ along a vector field $ \zeta $ can be written as
\begin{equation}\label{Lie_der}
\begin{aligned} 
&(\pounds _ \zeta \kappa  )^{\alpha_1 ... \alpha_p}_{\beta _1 ... \beta _q} 
= \zeta ^ \gamma\partial _ \gamma \kappa  ^{\alpha_1 ... \alpha_p}_{\beta _1 ... \beta _q}  + \widehat{ \kappa  }\;^{\alpha_1 ... \alpha_p\mu}_{\beta _1 ... \beta _q\nu} \partial _ \mu \zeta ^ \nu,
\end{aligned}
\end{equation}  
where $\widehat{ \kappa  }$ is the $(p+1, q+1)$ tensor field defined by
\begin{equation}\label{hat_kappa} 
\widehat{ \kappa  }\;^{ \alpha _1 ... \alpha _p\nu} _{ \beta _1... \beta _q\mu}= \sum_r \big( \kappa  ^{ \alpha _1 ... \alpha _p} _{ \beta _1... \beta _{r-1} \mu \beta _{r+1} ...\beta _q} \delta ^ \nu_{ \beta _r}- \kappa  ^{ \alpha _1 ... \alpha  _{r-1} \nu \alpha  _{r+1} ...\alpha _p} _{ \beta _1... \beta _q} \delta ^ {\alpha  _r} _\mu\big).
\end{equation} 
Formulas \eqref{Lie_der} can also be written in terms of a given torsion free covariant derivative $ \nabla $ on $ \mathcal{M} $ by replacing $ \partial _ \gamma $ with $ \nabla _ \gamma $. With such choice, we have the global formula
\begin{equation}\label{global_formula} 
\pounds _ \zeta  \kappa = \nabla _ \zeta \kappa  + \widehat{ \kappa  }: \nabla \zeta .
\end{equation}

\begin{lemma}\label{technical_lemma} Let $ \kappa  $ be a $(p,q)$-tensor field and $\pi$ a $(q,p)$-tensor field density.
Then, we have locally
\begin{equation}\label{formula_Lemma_local}
\begin{aligned} 
(\pounds _ \zeta  \kappa ) ^{ \alpha _1 ... \alpha _p} _{ \beta _1... \beta _q}  \pi  _{ \alpha _1 ... \alpha _p} ^{ \beta _1... \beta _q}&=  \zeta ^ \mu \left( \partial _ \mu \kappa  ^{ \alpha _1 ... \alpha _p} _{ \beta _1... \beta _q} \pi _{ \alpha _1 ... \alpha _p} ^{ \beta _1... \beta _q} - \partial _ \nu( \widehat{ \kappa  }\;^{ \alpha _1 ... \alpha _p\nu} _{ \beta _1... \beta _q\mu}\pi  _{ \alpha _1 ... \alpha _p} ^{ \beta _1... \beta _q})\right) \\
& \qquad + \partial _ \nu \left( \widehat{ \kappa  }\;^{ \alpha _1 ... \alpha _p\nu} _{ \beta _1... \beta _q\mu}\pi  _{ \alpha _1 ... \alpha _p} ^{ \beta _1... \beta _q} \zeta ^ \mu \right).
\end{aligned}
\end{equation}  
The same formula holds with $\partial _ \gamma $ replaced by $ \nabla _ \gamma $, with $ \nabla $ a torsion free covariant derivative. In this case, we have the global formula
\begin{equation}\label{formula_Lemma_global} 
\textcolor{black}{ \pounds _ \zeta  \kappa  : \pi = \nabla _ \zeta \kappa  : \pi - \operatorname{div}^ \nabla (  \pi \!\therefore\! \widehat{ \kappa  }) \cdot \zeta+ \operatorname{div}( ( \pi\!\therefore\! \widehat{ \kappa  } ) \cdot \zeta  ) },
\end{equation}
where the colon ``$\,:\,$" in $\pounds _ \zeta  \kappa  : \pi $ and $\nabla _ \zeta \kappa  : \pi$ denotes the full contraction and where $ \pi \!\therefore\! \widehat{\kappa }$ is the $(1,1)$ tensor field density obtaining by contracting all the respective indices of $ \widehat{ \kappa  } $ and $ \pi $ except the last covariant and contravariant indices of $\widehat{ \kappa  }$.
\end{lemma}

Note that in \eqref{formula_Lemma_global}, the first divergence operator is associated to the torsion free covariant derivative $ \nabla $, while the last one is canonically defined since it acts on a vector field density, see \cite[Remark 4.2]{GB2024}.

We now specify this result to the case where $ \kappa $ is the electromagnetic potential $A$ or the Faraday two-form $F$.
We note that 
\begin{align*}
(\pounds _\zeta A)_ \beta &=\zeta ^ \gamma \partial _ \gamma A_ \beta + \widehat{A}_ \beta {}^ \mu _ \nu \partial _ \mu \zeta  ^\nu , \qquad  \widehat{A}_ \beta {}^ \mu _ \nu = A_ \nu \delta ^ \mu _ \beta \\
(\pounds _ \zeta F)_{ \beta _1 \beta _2}&= \zeta ^ \gamma  \partial _ \gamma F_{ \beta _1 \beta _2} + \widehat{F}_{ \beta _1 \beta_2} {}^ \mu _ \nu \partial _ \mu \zeta  ^\nu, \qquad \widehat{F}_{ \beta _1 \beta _2}{} ^\mu  _ \nu = F_{ \nu  \beta _2} \delta  _{ \beta _1} ^ \mu  + F_{  \beta _1 \nu }  \delta  ^ \mu _{ \beta _2}
\end{align*} 
hence, we get $(\pi \therefore \widehat{A} ) ^ \mu _ \nu  = \pi^ \mu A_ \nu $ and $(\pi\therefore\widehat{F})^ \mu _ \nu = \frac{1}{2} \pi ^{ \alpha \beta }\widehat{F}_{ \alpha \beta }{}^ \mu _ \nu = \pi^ { \mu \gamma }F_{ \nu \gamma }$, which we write as
\begin{equation}\label{notation_therefore_AF} 
\pi\therefore \widehat{A} =\pi \otimes A \quad\text{and}\quad \pi\therefore\widehat{F} =\pi \stackrel{\rm tr}{ \otimes }F,
\end{equation}
where $ \pi $ is the $k$-multivector field density associated with $A$ ($k=1$) or $F$ ($k=2$).
The first operation is the tensor product $ \otimes $, while the second one involves taking a tensor product as well as a trace, hence the notation $\stackrel{\rm tr}{ \otimes }$.

\subsection{Reduced Euler-Lagrange equations on spacetime}\label{spacetime_reduction}

We shall now derive the variational principle and equations on spacetime, under the covariance assumption \eqref{material_cov}  of the material Lagrangian density $\mathscr{L}$. While the variational principle using $\mathscr{L}$ is the usual Hamilton principle, the variational principle induced on spacetime involves constrained Eulerian variations induced by the arbitrary Lagrangian variations.  This is stated in the following theorem in which we denote by
\[
\mathcal{N}= \Phi ( \mathcal{D}) \quad\text{and}\quad  \partial _{\rm cont} \mathcal{N} = \Phi ([\lambda _0, \lambda _1] \times \partial \mathcal{B} )
\]
the spacetime region occupied by the continuum and its boundary, respectively. We note that when $W=\partial_\lambda$, then the generalized velocity $w=\Phi_*\partial_\lambda$ is parallel to the boundary $\partial_{\rm cont}\mathcal{N}$, a boundary condition that we shall not always explicitly write down in the equations.

\begin{theorem}[Covariant Eulerian reduction for electromagnetic continua]\label{spacetime_reduced_EL} Let $ \mathscr{L} : J^1( \mathcal{D}  \times \mathcal{M} ) \times \wedge ^1 T ^*\mathcal{D} \times \wedge ^2T^* \mathcal{D} \times  T \mathcal{D} \times T^p_q\mathcal{D}  \times T^r_s \mathcal{M}  \rightarrow \wedge ^{n+1}T^* \mathcal{D} $ be a Lagrangian density with the material covariance property \eqref{material_cov} and consider the associated reduced spacetime Lagrangian density $\ell : T^*\mathcal{M} \times \wedge ^2 T^* \mathcal{M}  \times  T \mathcal{M} \times T ^p _q  \mathcal{M}  \times  T^r _s  \mathcal{M} \rightarrow \wedge  ^{n+1} T^*\mathcal{M} $, defined in \eqref{def_ell}.
Fix
\begin{itemize}
\item[-] the reference tensor fields $W \in \mathfrak{X} ( \mathcal{D} )$, $K \in \mathcal{T} _q^p( \mathcal{D} )$;
\item[-]  the spacetime tensor field $ \gamma \in \mathcal{T} ^r_s( \mathcal{M} )$.
\end{itemize}
For each world-tube $ \Phi : \mathcal{D} \rightarrow \mathcal{M} $ and electromagnetic potential $\mathcal{A}\in \Omega^1(\mathcal{D})$, define
\begin{itemize}
\item[-] $\mathcal{N}= \Phi ( \mathcal{D})$ and $ \partial _{\rm cont} \mathcal{N} = \Phi ([\lambda _0, \lambda _1] \times \partial \mathcal{B} )$;
\item[-] $A= \Phi _* \mathcal{A} $, $w=\Phi_* W$, and $ \kappa = \Phi _ * K$.
\end{itemize}
Then, the following statements are equivalent:
\begin{itemize}
\item[\bf (i)] $\Phi: \mathcal{D} \rightarrow \mathcal{M} $ and $\mathcal{A} \in \Omega ^1 ( \mathcal{D} ) $ are critical points of the {\bfi Hamilton principle}
\[
\left. \frac{d}{d\varepsilon}\right|_{\varepsilon=0} \int_\mathcal{D}  \mathscr{L} ( j^1\Phi_ \varepsilon , \mathcal{A} _ \varepsilon , {\rm d} \mathcal{A} _ \varepsilon , W, K , \gamma \circ \Phi _ \varepsilon  )=0
\]
for arbitrary variations $\Phi _ \varepsilon  $ and $\mathcal{A} _ \varepsilon $ \textcolor{black}{with fixed endpoints at $ \lambda = \lambda _0, \lambda _1$}.

\item[\bf (ii)] $\Phi: \mathcal{D} \rightarrow \mathcal{M} $ and $\mathcal{A} \in \Omega ^1 ( \mathcal{D} ) $ are solutions of the {\bfi Euler-Lagrange equations} \eqref{EL_material}.
\item[\bf (iii)] $A \in \Omega ^1( \mathcal{N} )$, $w \in \mathfrak{X} (\mathcal{N} )$, and $ \kappa \in \mathcal{T} ^p_q( \mathcal{N} )$ are critical points of the {\bfi Eulerian variational principle} \color{black} 
\begin{equation}\label{Eulerian_VP}
\begin{aligned} 
&\!\!\left. \frac{d}{d\varepsilon}\right|_{\varepsilon=0}\int_{ \mathcal{N}_ \varepsilon } \ell\big( A_ \varepsilon , {\rm d} A_ \varepsilon  , w_ \varepsilon , \kappa _ \varepsilon , \gamma \big)=0 \quad \text{for variations}\\
& \delta \mathcal{N} = \zeta |_{ \partial \mathcal{N} } \big/T \partial \mathcal{N} , \;\; \delta A = - \pounds _ \zeta A + 
\deltabar A,\;\; \delta w = - \pounds _ \zeta w ,\;\;   \delta \kappa = - \pounds _ \zeta \kappa,\phantom{\int_A^B}
\end{aligned} 
\end{equation}
where $ \zeta$ is an arbitrary vector field on $ \mathcal{N} $ such that $ \zeta |_{ \Phi ( \lambda _0, \mathcal{B} )}= \zeta |_{ \Phi ( \lambda _1, \mathcal{B} )}=0$ and $\deltabar A$ is an arbitrary one-form on $ \mathcal{N} $ such that $\deltabar A|_{ \Phi ( \lambda _0, \mathcal{B} )}= \deltabar A |_{ \Phi ( \lambda _1, \mathcal{B} )}=0$.
\item[\bf (iv)] $A \in \Omega ^1 ( \mathcal{N} )$, $w \in \mathfrak{X} (\mathcal{N} )$, and $ \kappa \in \mathcal{T} ^p_q( \mathcal{N} )$ are solution of the {\bfi reduced Euler-Lagrange equations on spacetime}
\begin{equation}\label{spacetime_EL} 
\hspace{-1cm}\left\{
\begin{array}{l}
\displaystyle\vspace{0.2cm}\operatorname{div}^ \nabla \!\Big(  \ell \delta  + w \otimes \frac{\partial \ell}{\partial w} -  \textcolor{black}{  \frac{\partial \ell}{\partial \kappa }\!\therefore\!  \widehat{ \kappa  }} -  \frac{\partial \ell}{\partial A } \otimes  A  -  \frac{\partial \ell}{\partial F }\stackrel{\rm tr}{ \otimes }F\Big)  =  \frac{\partial ^\nabla\!\ell}{\partial x}+\frac{\partial \ell}{\partial \gamma } : \nabla  \gamma\\
\displaystyle\vspace{0.2cm} \frac{\partialnew \ell}{\partialnew A} + {\rm d}  \frac{\partialnew \ell}{\partialnew F}=0, \qquad \pounds _w \kappa =0, \qquad i^*_{ \partial _{\rm cont}\mathcal{N} } \frac{\partialnew \ell}{\partialnew F}=0,\\
\displaystyle i_{ \partial _{\rm cont} \mathcal{N} }^*\Big(\operatorname{tr}\Big(  \ell \delta  + w \otimes \frac{\partial \ell}{\partial w} -   \textcolor{black}{ \frac{\partial \ell}{\partial \kappa }\!\therefore\!  \widehat{ \kappa  } }-  \frac{\partial \ell}{\partial A }\otimes A - \frac{\partial \ell}{\partial F }\stackrel{\rm tr}{ \otimes }F\Big) \cdot \zeta\Big)=0,\;\;\forall \zeta,
\end{array}
\right.
\end{equation} 
written with the help of a torsion free covariant derivative $ \nabla $ \textcolor{black}{and with $ \delta $ the Kronecker delta\footnote{\textcolor{black}{No confusion should arise with the same symbol $ \delta $ also used for variations}.}}. In local coordinates, writing $\ell=\bar\ell d^{n+1}x$, the boundary conditions read
\[
\Big( \bar\ell  \delta  + w \otimes \frac{\partial \bar\ell}{\partial w}-   \textcolor{black}{\frac{\partial \bar{\ell}}{\partial \kappa }\!\therefore\! \widehat{ \kappa  } } -  \frac{\partial \bar\ell}{\partial A } \otimes  A  -  \frac{\partial \bar\ell}{\partial F }\stackrel{\rm tr}{ \otimes }F\Big) ^\mu_\nu d^nx_\mu=0, \quad \text{on} \quad   \textcolor{black}{\partial_{\rm cont} \mathcal{N}}
\]
\[
\text{and} \quad \frac{\partial \bar \ell}{\partial F_{ \mu \nu  }}d^nx_ \mu =0 , \quad \text{on} \quad   \textcolor{black}{\partial_{\rm cont} \mathcal{N}}.
\]
\end{itemize}
\end{theorem}
\begin{proof} \textcolor{black}{First, let us note that the constrained variations in {\bf (iii)} are induced by the variations $\Phi _ \varepsilon  $ and $\mathcal{A} _ \varepsilon $ of the world-tube and electromagnetic potential by using the relations $ \mathcal{N} _ \varepsilon = \Phi _ \varepsilon  ( \mathcal{D} )$, $A_ \varepsilon =(\Phi _ \varepsilon )_* \mathcal{A} _ \varepsilon  $, $w_ \varepsilon = (\Phi_ \varepsilon  )_*W$, $ \kappa _ \varepsilon = (\Phi_ \varepsilon  )_*K$ and defining $ \zeta= \delta \Phi \circ \Phi ^{-1}$}. For simplicity, we include only the material tensor $K$, since the treatment of the vector field $W$ is a particular case of it. Using Lemma \ref{technical_lemma} as well as the definition of the partial derivatives in \eqref{partial2} and \eqref{partial1}, we have\color{black} 
\begin{equation}\label{big_computation}
\begin{aligned} 
&\left. \frac{d}{d\varepsilon}\right|_{\varepsilon=0}\int_{ \Phi _ \varepsilon  ( \mathcal{D} )} \ell\big( (\Phi_ \varepsilon  )_* \mathcal{A} _ \varepsilon , (\Phi_ \varepsilon  )_* {\rm d} \mathcal{A} _ \varepsilon ,(\Phi_ \varepsilon  )_*K, \gamma \big)\\
&\quad = \int_\mathcal{N}  ( - \pounds _ \zeta A + \deltabar A) \wedge \frac{  \partialnew\ell}{\partialnew A}+  {\rm d} ( - \pounds _ \zeta A + \deltabar A) \wedge \frac{  \partialnew\ell}{\partialnew A}+ \frac{\partial \ell}{\partial \kappa  } : \delta \kappa+ \int_{ \partial \mathcal{N} } \mathbf{i} _{ \delta \Phi \circ \Phi ^{-1} }\ell \\
& \quad = - \int_\mathcal{N}  \frac{\partial \ell}{\partial A  } : \pounds _ \zeta A+  \frac{\partial \ell}{\partial F  } : \pounds _ \zeta F+  \frac{\partial \ell}{\partial \kappa  } : \pounds _ \zeta \kappa + \deltabar A \wedge \left(  \frac{  \partialnew\ell}{\partialnew A} + {\rm d}  \frac{  \partialnew\ell}{\partialnew F}\right) \\
& \qquad \qquad + \int_{ \partial \mathcal{N} } \mathbf{i} _{ \zeta }\ell + \deltabar A \wedge  \frac{  \partialnew\ell}{\partialnew A}\\
&\quad = - \int_\mathcal{N}  \Big(    \frac{\partial \ell}{\partial A } :  \nabla _ \zeta A +  \frac{\partial \ell}{\partial F } :  \nabla _ \zeta F +  \frac{\partial \ell}{\partial \kappa } :  \nabla _ \zeta \kappa - \operatorname{div}^ \nabla \!\Big( \frac{\partial \ell}{\partial \kappa } \!\therefore\!  \widehat{ \kappa  }+\frac{\partial \ell}{\partial A } \!\therefore\!  \widehat{ A  }+\frac{\partial \ell}{\partial F } \!\therefore\!  \widehat{ F  }\Big)  \cdot \zeta \Big) \\
& \qquad \qquad +  \int_\mathcal{N}  \deltabar A \wedge \left(  \frac{  \partialnew\ell}{\partialnew A} + {\rm d}  \frac{  \partialnew\ell}{\partialnew F}\right) \\
&\qquad \qquad + \int_{ \partial \mathcal{N} }\operatorname{tr}  \Big( \ell \delta -\frac{\partial \ell}{\partial A } \!\therefore\!  \widehat{ A  }-\frac{\partial \ell}{\partial F } \!\therefore\!  \widehat{ F  }-  \frac{\partial \ell}{\partial \kappa }\!\therefore\! \widehat{ \kappa  } \Big) \cdot \zeta  +\int_{ \partial \mathcal{N} }  \deltabar A \wedge  \frac{  \partialnew\ell}{\partialnew A} \\
&\quad = - \int_ \mathcal{N}  \Big(    \operatorname{div}^ \nabla \!\Big( \ell \delta   -  \frac{\partial \ell}{\partial A } \otimes  A  -  \frac{\partial \ell}{\partial F }\stackrel{\rm tr}{ \otimes }F -   \frac{\partial \ell}{\partial \kappa }\!\therefore\!  \widehat{ \kappa  }\Big) \cdot \zeta - \frac{\partial ^ \nabla \!\ell}{\partial x} \cdot \zeta - \frac{\partial \ell}{\partial \gamma } : \nabla _ \zeta \gamma  \Big)  \\
& \qquad \qquad +  \int_\mathcal{N}  \deltabar A \wedge \left(  \frac{  \partialnew\ell}{\partialnew A} + {\rm d}  \frac{  \partialnew\ell}{\partialnew F}\right) \\
& \qquad \qquad +\int_{ \partial \mathcal{N} }\operatorname{tr}  \Big( \ell \delta -\frac{\partial \ell}{\partial A } \!\therefore\!  \widehat{ A  }-\frac{\partial \ell}{\partial F } \!\therefore\!  \widehat{ F  }-  \frac{\partial \ell}{\partial \kappa }\!\therefore\! \widehat{ \kappa  } \Big) \cdot \zeta  +\int_{ \partial \mathcal{N} }  \deltabar A \wedge  \frac{  \partialnew\ell}{\partialnew A}.
\end{aligned}
\end{equation} 
In the third equality we used $\mathbf{i} _ \zeta \ell =\operatorname{tr}(\ell \delta ) \cdot  \zeta $. In the fourth equality we used the following formula
\begin{equation}\label{nabla_derivative_ell} \color{black} 
\nabla_ \zeta  [ \ell( A, F, \kappa , \gamma )]=  \frac{\partial ^\nabla\!\ell}{\partial x} \cdot \zeta +\frac{\partial \ell}{\partial A } : \nabla _ \zeta A+\frac{\partial \ell}{\partial F } : \nabla _ \zeta F+\frac{\partial \ell}{\partial \kappa } : \nabla _ \zeta \kappa + \frac{\partial \ell}{\partial \gamma } : \nabla _ \zeta \gamma 
\end{equation} 
for the derivative of the bundle map $\ell$ with respect to the given connection $ \nabla $, with $\frac{\partial ^ \nabla \ell}{\partial x}$ the derivative of $\ell$ with respect to the base point defined with the help of $ \nabla $. We also used the relations \eqref{notation_therefore_AF}.

To include the reference vector field $W$ we note that for $w$ a vector field, we have
\[
\textcolor{black}{ \frac{\partial \ell}{\partial w}\!\therefore\!  \widehat{w}}= - w \otimes \frac{\partial \ell}{\partial w} .
\]
Since the vector field $ \zeta $ and the one-form $ \deltabar A$ are arbitrary, while vanishing on $ \Phi ( \lambda _0, \mathcal{B} )$ and $ \Phi ( \lambda _1, \mathcal{B} )$, we get the equations \eqref{spacetime_EL}.
\end{proof}


\paragraph{Spacetime covariance.} We now examine the case in which the material Lagrangian density $\mathscr{L}$ in \eqref{Lagrangian_general}  is also spacetime covariant with respect to $ \operatorname{Diff}( \mathcal{M} )$, see \eqref{spacetime_cov}, in addition to the material covariance \eqref{material_cov}.

\begin{lemma}\label{spacetime_mat_ell} Assume that the material Lagrangian density $ \mathscr{L}$ in \eqref{material_Lagrangian} is materially covariant and consider the associated spacetime Lagrangian $\ell $, see \eqref{def_ell}. Then if $ \mathscr{L} $ is also spacetime covariant, we have
\begin{equation}\label{double_covariance}
\psi ^\ast\left[ \ell (  \psi _* A, \psi _* F, \psi_\ast w, \psi_\ast \kappa , \psi _\ast \gamma ) \right] = \ell (A, F,  w, \kappa , \gamma ), \quad \forall\psi \in \operatorname{Diff}(\mathcal{M} ).
\end{equation}
Therefore, $\ell$ satisfies the following equalities:
\[
\frac{\partial ^ \nabla \!\ell}{\partial x}=0 \quad\text{and}\quad \textcolor{black}{ \frac{\partial \ell}{\partial A} \otimes A +  \frac{\partial \ell}{\partial F} \stackrel{\rm tr}{\otimes} F - w \otimes \frac{\partial \ell}{\partial w} +   \frac{\partial \ell}{\partial \kappa }\!\therefore\!  \widehat{ \kappa  } +   \frac{\partial \ell}{\partial \gamma  }\!\therefore\!   \widehat{ \gamma   } =\ell \delta }.
\]
\end{lemma} 
\begin{proof} The first assertion follows from the relation \eqref{def_ell} combined with the spacetime covariance condition \eqref{spacetime_cov}. Taking a path of diffeomorphisms $ \psi _ \varepsilon  \in \operatorname{Diff}( \mathcal{M} )$ passing through the identity at $\varepsilon =0$ and taking the $ \varepsilon $-derivative of the condition $\ell (  \psi_ \varepsilon ^\ast A, \psi_ \varepsilon ^\ast F, \psi_ \varepsilon ^\ast w, \psi_ \varepsilon ^\ast \kappa , \psi _ \varepsilon ^\ast \gamma )= \psi _ \varepsilon ^*[\ell(A, F, w, \kappa , \gamma )]$, we get $ \frac{\partial \ell}{\partial A} \cdot \pounds _ \zeta A+  \frac{\partial \ell}{\partial F}:\pounds _ \zeta F+  \frac{\partial \ell}{\partial w} \cdot \pounds _ \zeta w + \textcolor{black}{ \frac{\partial \ell}{\partial  \kappa } : \pounds _ \zeta  \kappa + \frac{\partial \ell}{\partial  \gamma } : \pounds _ \zeta \gamma }= \operatorname{div}(\ell \zeta )$, for the vector field $ \zeta = \left. \frac{d}{d\varepsilon}\right|_{\varepsilon=0} \psi _ \varepsilon $, where $\operatorname{div}(\ell \zeta )$ is the divergence of the vector field density $\ell \zeta $. Choosing a torsion free covariant derivative $ \nabla $ and using the formulas \eqref{global_formula}, $ \operatorname{div}(\ell \zeta ) = \nabla _ \zeta [\ell( \cdot )] + \ell  \operatorname{div}^ \nabla\!\zeta=\nabla _ \zeta [\ell( \cdot )] + \ell \delta :  \nabla \zeta  $, and \eqref{nabla_derivative_ell},  the requested identity follows since the vector field $ \zeta $ is arbitrary.
\end{proof}

In the next corollary, we obtain that the reduced spacetime Euler-Lagrange equations can be written exclusively in terms of the partial derivative with respect to $ \gamma $, when $\mathscr{L}$ satisfies both covariance properties. We recall that $ \gamma $ is not a variable in the spacetime description, but a given fixed parameter, see Remark \ref{strangeness}.

\begin{corollary}\label{corol_spacetime} Assume that the material Lagrangian density $\mathscr{L}$ in \eqref{material_Lagrangian} is material and spacetime covariant. Then the reduced Euler-Lagrange equations \eqref{spacetime_EL} are equivalently written as
\begin{equation}\label{spacetime_EL_sc_gamma} \color{black} 
\left\{
\begin{array}{l}
\displaystyle\vspace{0.2cm}\operatorname{div}^ \nabla \!\Big(\frac{\partial \ell}{\partial \gamma  }\!\therefore\!  \widehat{ \gamma   }\Big) =\frac{\partial \ell}{\partial \gamma } : \nabla  \gamma\\
\displaystyle\vspace{0.2cm} \frac{\partialnew \ell}{\partialnew A} + {\rm d}  \frac{\partialnew \ell}{\partialnew F}=0, \qquad \pounds _w \kappa =0\\
\displaystyle\vspace{0.2cm}i^*_{ \partial _{\rm cont}\mathcal{N} } \frac{\partialnew \ell}{\partialnew \mathsf{F}}=0, \qquad  i_{ \partial _{\rm cont} \mathcal{N} }^*\Big(\operatorname{tr}\Big(\frac{\partial \ell}{\partial \gamma  }\!\therefore\!  \widehat{ \gamma   } \Big) \cdot \zeta\Big)=0, \;\;\forall \zeta.
\end{array}
\right.
\end{equation}
The local forms of the boundary conditions are
\[
\Big( \frac{\partial \bar{\ell}}{\partial \gamma  }\!\therefore\!  \widehat{ \gamma   }\Big)^\mu_\nu {\rm d} ^nx_\mu=0 \quad \text{and} \quad \frac{\partial \bar \ell}{\partial F_{ \mu \nu  }}{\rm d} ^nx_ \mu =0 , \quad \text{on} \quad  \partial_{\rm cont} \mathcal{N}.
\]
\end{corollary}

\begin{remark}[Isotropy subgroup]\label{Diff_KW}\rm
The results stated above still hold when material covariance is satisfied only with respect to the isotropy subgroup $ \operatorname{Diff}_{W,K}( \mathcal{D} )=\{\varphi\in \operatorname{Diff}(\mathcal{D})\mid \varphi^*W=W,\;\varphi^*K=K\}$ of the material tensor fields. Spacetime covariance with respect to the whole group $ \operatorname{Diff}( \mathcal{M} )$ is however needed.
\end{remark}

\begin{remark}[Variation conversion]\label{strangeness}\rm In equations \eqref{spacetime_EL_sc_gamma}, one observes that the partial derivatives with respect to the non-dynamic field $ \gamma $ appear, while the ones with respect to the actual variables $w$ and $\kappa$, which are to be varied in the action principle, do not. This apparent strangeness can be explained by the  property \eqref{double_covariance} of $\ell$, which implies that the variations of $w$ and $\kappa$ can be ``converted" into variations of the otherwise fixed $ \gamma $. Explicitly, this is seen from the following chain of equalities, in which we write $ \Phi _ \varepsilon = \psi _ \varepsilon \circ \Phi $ for some spacetime diffeomorphism $ \psi _ \varepsilon \in \operatorname{Diff}( \mathcal{M} )$:
\begin{equation}\label{strangeness}
\begin{aligned} 
&\left. \frac{d}{d\varepsilon}\right|_{\varepsilon=0}\int_{ \Phi _ \varepsilon  ( \mathcal{D} )} \ell\big( (\Phi_ \varepsilon  )_* \mathcal{A} _ \varepsilon , (\Phi_ \varepsilon  )_* {\rm d} \mathcal{A} _ \varepsilon ,(\Phi_ \varepsilon  )_*W,(\Phi_ \varepsilon  )_*K, \gamma \big)\\
&=\left. \frac{d}{d\varepsilon}\right|_{\varepsilon=0}\int_{ \psi  _ \varepsilon  ( \mathcal{N} )} \ell\big( (\psi _ \varepsilon  )_* A_ \varepsilon , (\psi_ \varepsilon  )_* {\rm d} A_ \varepsilon ,(\psi_ \varepsilon  )_*w,(\psi_ \varepsilon  )_*\kappa , \gamma \big)\\
&=\left. \frac{d}{d\varepsilon}\right|_{\varepsilon=0}\int_{ \psi  _ \varepsilon  ( \mathcal{N} )} (\psi _ \varepsilon  )_* \big[\ell\big( A_ \varepsilon ,   {\rm d} A_ \varepsilon , w, \kappa , (\psi _ \varepsilon  )^*\gamma \big)\big]\\
&=\left. \frac{d}{d\varepsilon}\right|_{\varepsilon=0}\int_{ \mathcal{N}} \ell\big( A_ \varepsilon ,   {\rm d} A_ \varepsilon , w, \kappa , (\psi _ \varepsilon  )^*\gamma \big),
\end{aligned}
\end{equation}
where \eqref{double_covariance} was used in the second equality. The last expression shows that besides $A$, only the non-dynamic tensor $ \gamma $ is varied, as $ \delta \gamma =\pounds _ \zeta \gamma $. A direct application of the principle as it arises in this last term of \eqref{strangeness}  gives \eqref{spacetime_EL_sc_gamma}.
\end{remark}

\paragraph{The case of a Lorentzian metric on spacetime.} Evidently, the most relevant case is when $\gamma $ is a Lorentzian metric $\mathsf{g}$, see \S\ref{EMFS}. In this case the word-tube satisfies
$\mathsf{g}( \partial _ \lambda \Phi , \partial _ \lambda \Phi )<0$ and $ \partial _{\rm cont} \mathcal{N} = \Phi ([a,b] \times \partial \mathcal{B} )$ is a timelike region.
We give below the reduced Euler-Lagrange equations on spacetime \eqref{spacetime_EL} and \eqref{spacetime_EL_sc_gamma} in this specific situation.

\begin{corollary}
Consider the special case when $\gamma$ is a Lorentzian metric $\mathsf{g}$. Then the reduced Euler-Lagrange equations on spacetime \eqref{spacetime_EL} become
\begin{equation}\label{spacetime_EL_g} 
\left\{
\begin{array}{l}
\displaystyle\vspace{0.2cm}\operatorname{div}^ \nabla \!\Big(  \ell \delta  + w \otimes \frac{\partial \ell}{\partial w} -  \textcolor{black}{  \frac{\partial \ell}{\partial \kappa }\!\therefore\!  \widehat{ \kappa  }} -  \frac{\partial \ell}{\partial A } \otimes  A  -  \frac{\partial \ell}{\partial F }\stackrel{\rm tr}{ \otimes }F\Big)  =  \frac{\partial ^\nabla\!\ell}{\partial x}\\
\displaystyle\vspace{0.2cm} \frac{\partialnew \ell}{\partialnew A} + {\rm d}  \frac{\partialnew \ell}{\partialnew F}=0, \qquad \pounds _w \kappa =0, \qquad i^*_{ \partial _{\rm cont}\mathcal{N} } \frac{\partialnew \ell}{\partialnew F}=0,\\
\displaystyle \Big(  \ell \delta  + w \otimes \frac{\partial \ell}{\partial w} -   \textcolor{black}{ \frac{\partial \ell}{\partial \kappa }\!\therefore\!  \widehat{ \kappa  } }-  \frac{\partial \ell}{\partial A }\otimes A - \frac{\partial \ell}{\partial F }\stackrel{\rm tr}{ \otimes }F\Big) (\cdot, n^\flat)=0\;\;\text{on}\;\;\partial_{\rm cont}\mathcal{N},
\end{array}
\right.
\end{equation} 
where $ \nabla $ is the Levi-Civita covariant derivative and $n$ is a unit normal vector field to $ \partial_{\rm cont} \mathcal{N} $, all associated to $\mathsf{g}$.

If, in addition, the material Lagrangian density is also spacetime covariant, then these equations are given by \eqref{spacetime_EL_sc_gamma} which, for $\gamma=\mathsf{g}$, become
\begin{equation}\label{spacetime_EL_sc_gamma_g} \color{black} 
\left\{
\begin{array}{l}
\displaystyle\vspace{0.2cm}\operatorname{div}^ \nabla \frac{\partial \ell}{\partial \mathsf{g} }=0,\qquad \frac{\partialnew \ell}{\partialnew A} + {\rm d}  \frac{\partialnew \ell}{\partialnew F}=0, \qquad \pounds _w \kappa =0\\
\displaystyle\vspace{0.2cm}i^*_{ \partial _{\rm cont}\mathcal{N} } \frac{\partialnew \ell}{\partialnew \mathsf{F}}=0, \qquad  \frac{\partial \ell}{\partial \mathsf{g}}( \cdot , n ^\flat )=0 \quad \text{on} \quad  \partial \mathcal{N}.
\end{array}
\right.
\end{equation}
\end{corollary}


\subsection{Reduced convective Euler-Lagrange formulation}\label{convective_reduction}

In parallel with the reduced Euler-Lagrange equation on spacetime, which are associated to material covariance, one can also consider the less common \textit{reduced convective Euler-Lagrange equations}. They are associated to the covariance assumption \eqref{material_cov} of the material Lagrangian density $\mathscr{L}$, which is an axiom of the continuum theory, while material covariance is related to specific properties of the continuum. Their potential use in relativistic continuum mechanics doesn't seem to be known. In Newtonian continuum mechanics, they form an indispensable tool in the study of rods and shells, for example, regarded as one- and two-dimensional continua respectively. The general reduced convected Euler-Lagrange equations for Newtonian continuum mechanics were derived in \cite{GBMaRa2012}.

\begin{theorem}[Covariant convective reduction for electromagnetic continua]\label{convective_reduced_EL} Let  $ \mathscr{L} : J^1( \mathcal{D}  \times \mathcal{M} ) \times \wedge ^1 T ^*\mathcal{D} \times \wedge ^2T^* \mathcal{D} \times T \mathcal{D}  \times T^p_q\mathcal{D}  \times T^r_s \mathcal{M}  \rightarrow \wedge ^{n+1} T^*\mathcal{D} $ be a Lagrangian density with the spacetime covariance property \eqref{spacetime_cov} and consider the associated reduced convective Lagrangian density $\mathcal{L} : \wedge ^1 T^* \mathcal{D} \times \wedge ^2 T^* \mathcal{D} \times  T \mathcal{D}  \times  T ^p _q  \mathcal{D}  \times  T^r _s  \mathcal{D} \rightarrow \wedge  ^{n+1}T^* \mathcal{D} $, defined in \eqref{def_calL}.
Fix
\begin{itemize}
\item[-] the reference tensor fields $W \in \mathfrak{X} ( \mathcal{D} )$, $K \in \mathcal{T} _q^p( \mathcal{D} )$;
\item[-]  the spacetime tensor field $ \gamma \in \mathcal{T} ^r_s( \mathcal{M} )$.
\end{itemize}
For each world-tube $ \Phi : \mathcal{D} \rightarrow \mathcal{M} $ define
\begin{itemize}
\item[-] $\Gamma = \Phi ^* \gamma$.
\end{itemize}
 Then, the following statements are equivalent:
\begin{itemize}
\item[\bf (i)] $\Phi: \mathcal{D} \rightarrow \mathcal{M} $ and $\mathcal{A}  \in \Omega ^1 ( \mathcal{D} )$ are critical points of the {\bfi Hamilton principle}
\[
\left. \frac{d}{d\varepsilon}\right|_{\varepsilon=0} \int_\mathcal{D}  \mathscr{L} ( j^1\Phi_ \varepsilon , \mathcal{A} _ \varepsilon , {\rm d} \mathcal{A} _ \varepsilon , W, K , \gamma \circ \Phi _ \varepsilon  )=0
\]
for arbitrary variations $\Phi _ \varepsilon  $ and $\mathcal{A} _ \varepsilon $ \textcolor{black}{with fixed endpoints at $ \lambda = \lambda _0, \lambda _1$}.

\item[\bf (ii)] $\Phi: \mathcal{D} \rightarrow \mathcal{M} $ and $\mathcal{A} \in \Omega ^1 ( \mathcal{D} )$ are solutions of the {\bfi Euler-Lagrange equations} \eqref{EL_material}. 
\item[\bf (iii)] $\mathcal{A}  \in \Omega ^1 ( \mathcal{D} )$ and $\Gamma\in \mathcal{T} ^r_s( \mathcal{D} )$ are critical points of the {\bfi convective variational principle}
\[
\left. \frac{d}{d\varepsilon}\right|_{\varepsilon=0}\int_{ \mathcal{D} } \mathcal{L}\big(\mathcal{A} , {\rm d} \mathcal{A} , W,K, \textcolor{black}{ \Gamma }  \big)=0 \quad \textcolor{black}{ \text{for variations}\quad  \delta \mathcal{A}  \quad\text{and}\quad \delta \Gamma = \pounds _Z\Gamma},
\]
where $ \delta \mathcal{A} $ and $Z$ are arbitrary one-forms and vector field on $ \mathcal{D} $  vanishing at $ \lambda =\lambda _0, \lambda _1$.
\item[\bf (iv)] $\Gamma\in \mathcal{T} ^r_s( \mathcal{D} )$ is a solution of the {\bfi reduced convective Euler-Lagrange equations}\color{black} 
\begin{equation}\label{convected_EL} 
\left\{
\begin{array}{l}
\displaystyle\vspace{0.2cm}\operatorname{div}^ \nabla \!\Big( \mathcal{L}  \delta  -   \frac{\partial   \mathcal{L} }{\partial \Gamma }\!\therefore\!  \widehat{ \Gamma  }\Big)  =  \frac{\partial ^ \nabla \! \mathcal{L} }{\partial X}+\frac{\partial   \mathcal{L} }{\partial \mathcal{A}  } \cdot \nabla  \mathcal{A}+\frac{\partial   \mathcal{L} }{\partial \mathcal{F}  } \cdot \nabla  \mathcal{F} +\frac{\partial   \mathcal{L} }{\partial W } \cdot \nabla  W+\frac{\partial   \mathcal{L} }{\partial K } \cdot \nabla  K\\
\displaystyle\vspace{0.2cm} \frac{\partialnew \mathcal{L} }{\partialnew \mathcal{A} } + {\rm d} \frac{\partialnew \mathcal{L} }{\partialnew \mathcal{F} } =0, \qquad i_{ \partial \mathcal{D} } ^* \frac{\partialnew \mathcal{L} }{\partialnew \mathcal{F} } =0 \quad \text{on} \quad  \textcolor{black}{[\lambda _0, \lambda _1] \times   \partial \mathcal{B}}\\
\displaystyle i_{ \partial \mathcal{D} } ^*\Big( \operatorname{tr}\Big(\frac{\partial   \mathcal{L} }{\partial \Gamma }\!\therefore\!  \widehat{ \Gamma  }\Big)\cdot Z\Big)=0,\;\; \forall Z \quad \text{on} \quad  \textcolor{black}{[\lambda _0, \lambda _1] \times   \partial \mathcal{B}},
\end{array}
\right.
\end{equation}
written with the help of a torsion free covariant derivative $ \nabla $ on $ \mathcal{D} $.  In local coordinates, denoting $ \mathcal{L} =\bar{ \mathcal{L} } d^{n+1}X$, the boundary conditions read
\[
\Big(\frac{\partial   \bar{\mathcal{L}} }{\partial \Gamma }\!\therefore\!  \widehat{ \Gamma  }\Big) ^a_b {\rm d} ^nX_a=0 \quad\text{and}\quad \frac{\partial  \bar{\mathcal{L}} }{\partial \mathcal{F} _{ab}}  {\rm d} ^nX_ a=0\quad \text{on} \quad  [\lambda _0, \lambda _1] \times \partial \mathcal{B}.
\]
\end{itemize}
\end{theorem}
\begin{proof} The proof is left to the reader as it is similar to that of Theorem \ref{spacetime_reduced_EL}.
\end{proof}

\begin{remark}[Isotropy subgroup]\label{Diff_gamma}\rm
The result stated in the above theorem still holds when spacetime covariance is satisfied only with respect to the isotropy subgroup $ \operatorname{Diff}_\gamma ( \mathcal{M} )=\{\psi \in \operatorname{Diff} ( \mathcal{M} )\mid \psi^*\gamma=\gamma\}$ of the spacetime tensor field.
\end{remark}

\begin{remark}[Equivalent writings]\rm There are several equivalent way to express the first equation of \eqref{convected_EL}. For instance, noting that its right hand side can be written as $\nabla[\mathcal{L}(\mathcal{A},\mathcal{F},W, K, \Gamma)]-  \frac{\partial   \mathcal{L} }{\partial \Gamma }: \nabla \Gamma$, it can be simply expressed as
\begin{equation}\label{alternative1}
\operatorname{div}^ \nabla \!\Big(  \frac{\partial   \mathcal{L} }{\partial \Gamma }\!\therefore\!  \widehat{ \Gamma  }\Big)= \frac{\partial   \mathcal{L} }{\partial \Gamma }: \nabla \Gamma.
\end{equation}
Alternatively, using Lemma \ref{crucial_lemma} appearing later, it can be written as
\begin{equation}\label{alternative2}
\operatorname{div}^ \nabla \!\Big( \mathcal{L}  \delta  -   \frac{\partial   \mathcal{L} }{\partial \Gamma }\!\therefore\!  \widehat{ \Gamma  } - \frac{\partial \mathcal{L} }{\partial \mathcal{A} } \otimes \mathcal{A}  -  \frac{\partial \mathcal{L} }{\partial  \mathcal{F} } \stackrel{\rm tr}{\otimes} \mathcal{F} \Big)  =  \frac{\partial ^ \nabla \! \mathcal{L} }{\partial X} +\frac{\partial   \mathcal{L} }{\partial W } \cdot \nabla  W+\frac{\partial   \mathcal{L} }{\partial K } \cdot \nabla  K,
\end{equation}
which should be compared with its  \eqref{spacetime_EL}. The interpretation of these expressions needs further considerations of the stress-energy-momentum tensors, that will be explored in subsequent parts.
\end{remark}

\paragraph{Material covariance.} We now examine the case in which the material Lagrangian density $\mathscr{L}$ in \eqref{Lagrangian_general} is also \textcolor{black}{materially covariant} with respect to $ \operatorname{Diff}( \mathcal{D} )$, see \eqref{material_cov}, in addition to the spacetime covariance \eqref{spacetime_cov}. The proofs are left to the reader.

\begin{lemma} Assume that the material Lagrangian density $ \mathscr{L}$ in \eqref{material_Lagrangian} is spacetime covariant and consider the associated convective Lagrangian $ \mathcal{L}  $, see \eqref{def_calL}. Then if $ \mathscr{L} $ is also materially covariant, we have
\[
\varphi ^* \left[ \mathcal{L} (\mathcal{A} ,\mathcal{F} , W,K,\Gamma) \right] =  \mathcal{L} (\varphi ^* \mathcal{A} , \varphi ^* \mathcal{F} , \varphi ^* W,\varphi ^* K,\varphi ^* \Gamma), \quad \forall \varphi  \in \operatorname{Diff}(\mathcal{D} ).
\]
Therefore, \textcolor{black}{$\mathcal{L}$ satisfies the following equalities}:
\[
\frac{\partial  ^ \nabla \! \mathcal{L} }{\partial X}=0 \quad\text{and}\quad  \frac{\partial \mathcal{L} }{\partial \mathcal{A} } \otimes \mathcal{A}  +  \frac{\partial \mathcal{L} }{\partial \mathcal{F} } \stackrel{\rm tr}{\otimes}  \mathcal{F} - W \otimes \frac{\partial  \mathcal{L} }{\partial W} +   \textcolor{black}{ \frac{\partial  \mathcal{L} }{\partial K }\!\therefore\!   \widehat{ K  } +   \frac{\partial  \mathcal{L} }{\partial \Gamma  }\!\therefore\!   \widehat{ \Gamma   } =  \mathcal{L} \delta}.
\]
\end{lemma}

\begin{corollary} Assume that the material Lagrangian density \eqref{material_Lagrangian}  is material and spacetime covariant. Then the reduced convective Euler-Lagrange equations \eqref{convected_EL} are equivalently written as
\begin{equation}\label{convected_EL_cs_mc} \color{black} 
\!\!\!\left\{
\begin{array}{l}
\displaystyle\vspace{0.2cm}\operatorname{div}^\nabla\Big( - W \otimes \frac{\partial  \mathcal{L} }{\partial W} +   \frac{\partial  \mathcal{L} }{\partial K }\!\therefore\!   \widehat{ K  } \Big) = \frac{\partial  \mathcal{L} }{\partial W } \cdot \nabla  W+\frac{\partial \mathcal{L} }{\partial K } \cdot \nabla  K \\
\displaystyle\vspace{0.2cm} \frac{\partialnew \mathcal{L} }{\partialnew \mathcal{A} } + {\rm d} \frac{\partialnew \mathcal{L} }{\partialnew \mathcal{F} } =0\\
\displaystyle\vspace{0.2cm}i_{ \partial \mathcal{D} } ^* \frac{\partialnew \mathcal{L} }{\partialnew \mathcal{F} } =0 \quad \text{on} \quad  \textcolor{black}{[ \lambda _0 , \lambda _1 ] \times   \partial \mathcal{B}}\\
\displaystyle\vspace{0.2cm} i_{ \partial \mathcal{D} } ^* \Big(\operatorname{tr}\Big(\mathcal{L}  \delta  - \frac{\partial \mathcal{L} }{\partial \mathcal{A} } \otimes  \mathcal{A} -  \frac{\partial \mathcal{L} }{\partial \mathcal{F} } \stackrel{\rm tr}{\otimes} \mathcal{F}   + W \otimes \frac{\partial  \mathcal{L} }{\partial W} -   \frac{\partial  \mathcal{L} }{\partial K }\!\therefore\!   \widehat{ K  }\Big) \cdot Z \Big)=0, \forall Z\\
\displaystyle\text{on} \quad     \textcolor{black}{[\lambda _0 , \lambda _1 ] \times   \partial \mathcal{B}}.
\end{array}
\right.
\end{equation}
The local form of the boundary conditions is
\[
\textcolor{black}{ \Big(\bar{\mathcal{L}}  \delta   - \frac{\partial \mathcal{L} }{\partial \mathcal{A} } \otimes  \mathcal{A} -  \frac{\partial \mathcal{L} }{\partial \mathcal{F} } \stackrel{\rm tr}{\otimes} \mathcal{F}+ W \otimes \frac{\partial \bar{\mathcal{L}}}{\partial W} -   \frac{\partial  \bar{\mathcal{L}}}{\partial K }\!\therefore\!  \widehat{ K  }\Big)^a_b d^nX_a=0.}
\]
\end{corollary}

The first equation is obtained from \eqref{alternative2}. Note that, in a similar way with the case of Corollary \ref{corol_spacetime}, the partial derivatives with respect to the fields $W$ and $K$ appear, which are not dynamic variables in the convective description, but given fixed parameters, see also Remark \ref{strangeness}.

As earlier on the spacetime side, the results stated above still hold when the spacetime covariance is satisfied only with respect to the isotropy subgroup $ \operatorname{Diff}_{\gamma }( \mathcal{M} )$ of the spacetime tensor fields. Material covariance with respect to the whole group $ \operatorname{Diff}( \mathcal{D} )$ is however needed.

\subsection{Material versus spacetime electromagnetic potential}\label{material_spacetime_A}

It is possible to also regard the spacetime electromagnetic potential $A$, rather than $ \mathcal{A} $, as the primitive object in the Hamilton principle. In this case, instead of \eqref{material_Lagrangian}, we can start with the material Lagrangian density given as a bundle map
\begin{equation}\label{material_Lagrangian_2} 
\mathscr{L}' : J^1( \mathcal{D}  \times \mathcal{M} ) \times \wedge  ^1T^* \mathcal{M} \times \wedge ^2 T^* \mathcal{M}  \times T \mathcal{D} \times T^p_q\mathcal{D}  \times T^r_s \mathcal{M}  \rightarrow \wedge ^{n+1}T^* \mathcal{D},\end{equation} 
which is evaluated on the fields $ \Phi $, $A$, $ W$, $K$, $ \gamma $ as
\begin{equation}\label{Lagrangian_general_2} 
\mathscr{L}'(j^1\Phi, A \circ \Phi  , {\rm d}A \circ \Phi , W, K, \gamma \circ \Phi  ),
\end{equation}
to be compared with \eqref{Lagrangian_general}. The Hamilton principle reads
\[
\left. \frac{d}{d\varepsilon}\right|_{\varepsilon=0} \int_ \mathcal{D} \mathscr{L}'(j^1 \Phi _ \varepsilon ,  A_ \varepsilon\circ \Phi _ \varepsilon  , {\rm d} A _ \varepsilon \circ \Phi _ \varepsilon , W, K , \gamma  \circ \Phi _ \varepsilon )=0,
\]
for variations $ \delta \Phi $ and $ \delta A$ with $ \delta \Phi ( \lambda _0, \mathcal{B} )= \delta \Phi ( \lambda _1, \mathcal{B} )=0$ and $ \delta  A|_{ \Phi ( \lambda _0, \mathcal{B} )}= \delta  A |_{ \Phi ( \lambda _1, \mathcal{B} )}=0$.

Evidently, this Hamilton principle is equivalent to the previous one if the two Lagrangian densities $\mathscr{L}$ and $\mathscr{L}'$ become equal under the relation $A= \Phi_* \mathcal{A} $, i.e, if they are related as
\[
\bar{ \mathscr{L} }'\big(  X^a, x^\mu, v^\mu_a, A_ \mu , F _{ \mu \nu },  ... \big) {\rm d} ^{n+1} X =\bar{ \mathscr{L} }\big(  X^a, x^\mu, v^\mu_a, v^ \mu _aA_ \mu , v^ \mu _av^ \nu _bF_ {\mu \nu },  ... \big) {\rm d} ^{n+1} X,
\]
as bundle maps.
In terms of $ \mathscr{L}'$, the spacetime and material covariance properties \eqref{spacetime_cov} and  \eqref{material_cov} are given as
\begin{equation}\label{spacetime_cov_prime}
\begin{aligned} 
&\mathscr{L}'(j^1( \psi \circ  \Phi ), \psi _* A \circ \psi \circ \Phi , \psi _* {\rm d} A \circ \psi \circ \Phi,  W, K,\psi _* \gamma  \circ \psi \circ \Phi )\\
& \qquad    = \mathscr{L}'(j^1 \Phi ,A \circ \Phi  , {\rm d} A \circ \Phi  ,W, K, \gamma  \circ \Phi ), \quad \forall \psi \in \operatorname{Diff}( \mathcal{M} )
\end{aligned}
\end{equation}  
and
\begin{equation}\label{material_cov}
\begin{aligned} 
&\mathscr{L}'(j^1( \Phi   \circ  \varphi  ),A \circ \Phi \circ \varphi  , {\rm d} A \circ \Phi \circ \varphi  ,  \varphi ^*W, \varphi ^* K,\gamma   \circ \Phi \circ \varphi )\\
&\qquad   = \varphi ^* \big[\mathscr{L}'(j^1 \Phi , \mathcal{A} , {\rm d} \mathcal{A},W, K, \gamma  \circ \Phi )\big], \quad \forall \varphi  \in \operatorname{Diff}( \mathcal{D} )
\end{aligned}
\end{equation} 
respectively. These covariance properties are consistent with the fact that the electromagnetic potential is now primarily seen as a spatial tensor, like $ \gamma $, so that $ \operatorname{Diff}( \mathcal{M} )$ acts on it, and not $ \operatorname{Diff}( \mathcal{D} )$ as earlier. Under these hypotheses, the reduced Lagrangians $\ell$ and $ \mathcal{L} $ are the same as before, but the induced variations of the electromagnetic potential are different: the resulting spacetime variational formulation in \eqref{Eulerian_VP} differs in that the variations of the electromagnetic potential become now free instead of  $\delta A= - \pounds _ \zeta A +\deltabar A$. With this change, a similar computation shows that the first and last equations in \eqref{spacetime_EL} now read
\begin{equation}\label{option2_interior} 
\operatorname{div}^ \nabla \!\Big(  \ell \delta  + w \otimes \frac{\partial \ell}{\partial w} -  \textcolor{black}{  \frac{\partial \ell}{\partial \kappa }\!\therefore\!  \widehat{ \kappa  }} \Big)  =  \frac{\partial ^\nabla\!\ell}{\partial x}+\frac{\partial \ell}{\partial \gamma } : \nabla  \gamma +   \frac{\partial \ell}{\partial A } : \nabla   A  +  \frac{\partial \ell}{\partial F }: \nabla F
\end{equation} 
and
\begin{equation}\label{option2_boundary} 
i_{ \partial _{\rm cont} \mathcal{N} }^*\Big(\operatorname{tr}\Big(  \ell \delta  + w \otimes \frac{\partial \ell}{\partial w} -   \textcolor{black}{ \frac{\partial \ell}{\partial \kappa }\!\therefore\!  \widehat{ \kappa  } }\Big) \cdot \zeta\Big)=0,\;\;\forall \zeta,
\end{equation} 
while the others are kept unchanged. To prove that these equations are equivalent to the one previously derived in \eqref{spacetime_EL}, we need the following technical result.

\begin{lemma}\label{crucial_lemma} For an arbitrary Lagrangian density $\ell(A, {\rm d} A, w, \kappa , \gamma )$ and torsion free covariant derivative, we have the formula
\begin{equation}\label{crucial_formula}
\begin{aligned}
\operatorname{div}^ \nabla \left(\frac{\partial  \ell}{\partial A}\otimes\widehat{A} +\frac{\partial \ell}{\partial F}\stackrel{\rm tr}{ \otimes }\widehat{F}  \right) &= \frac{\partial \ell}{\partial A} \cdot \nabla A + \frac{\partial \ell}{\partial F} : \nabla F \\
& \quad -  \mathbf{i} _{\,\_\,} F \wedge \left( {\rm d} \frac{\partialnew  \ell}{\partialnew  F} + \frac{\partialnew  \ell}{\partialnew  A} \right) + \mathbf{i} _{\,\_\,}A \, {\rm d}  \frac{\partialnew  \ell}{\partialnew A}  
\end{aligned}
\end{equation} 
in $ \wedge ^1 T^* \mathcal{M} \otimes \wedge ^{n+1}T^* \mathcal{M} $. In particular, for a Lagrangian density of the form $\ell(A, {\rm d} A, \mathsf{g})$ and $ \nabla $ the Levi-Civita covariant derivative, \eqref{crucial_formula} reduces to
\begin{align*}
\operatorname{div}^ \nabla \left(\ell \delta - \frac{\partial  \ell}{\partial A}\otimes\widehat{A} -\frac{\partial \ell}{\partial F}\stackrel{\rm tr}{ \otimes }\widehat{F}  \right) &= \frac{\partial ^ \nabla \ell}{\partial x}+ \mathbf{i} _{\,\_\,} F \wedge \left( {\rm d} \frac{\partialnew  \ell}{\partialnew  F} + \frac{\partialnew  \ell}{\partialnew  A} \right) - \mathbf{i} _{\,\_\,}A \, {\rm d}  \frac{\partialnew  \ell}{\partialnew A}.
\end{align*} 
\end{lemma} 

By using \eqref{crucial_formula} in \eqref{option2_interior} we observe that, when  the second equation in \eqref{spacetime_EL} holds, then the first equation in \eqref{spacetime_EL} is equivalent to \eqref{option2_interior}. It remains to analyse when the boundary condition in \eqref{option2_boundary} is equivalent to the last one in  \eqref{spacetime_EL}. Clearly, this holds if and only if
\[
\Big(\frac{\partial \bar\ell}{\partial A _ \mu }  A_ \nu   -  \frac{\partial \bar\ell}{\partial F _{ \mu \gamma }}F_{\nu\gamma}\Big) {\rm d} ^nx_\mu=0, \quad \text{on} \quad   \partial_{\rm cont} \mathcal{N}.
\]
From the boundary condition already imposed on $\frac{\partialnew \ell}{\partialnew F}$, the second term vanishes, so that these boundary conditions are equivalent if and only if
\begin{equation}\label{BC_on_A}
\frac{\partial \ell}{\partial A_ \mu }{\rm d} ^n x_ \mu =0  , \quad \text{on} \quad   \partial_{\rm cont} \mathcal{N}.
\end{equation} 
Remarkably, it turns out that this condition holds for the class of Lagrangian densities relevant for electromagnetic media, as we shall see in \S\ref{EMFS}. For this class of Lagrangian that are linear in $A$, the derivative $ \frac{\partial \ell}{\partial A} $ is proportional to the world-velocity $u$, whose boundary condition $\mathsf{g}(u, n)=0$ on $ \partial _{\rm cont} \mathcal{N} $ implies \eqref{BC_on_A}. We summarize these results as follows.

\begin{proposition} Given a Lagrangian density $\ell(A, {\rm d} A, w, \kappa , \gamma )$ the variational principle given in \eqref{Eulerian_VP} and the same variational principle but with the variations $ \delta A$ free (instead of $ \delta A = - \pounds _ \zeta A + 
\deltabar A$) give the same equations and boundary conditions, provided \eqref{BC_on_A} is satisfied. This is the case for the Lagrangian densities of electromagnetic media.
\end{proposition} 

\begin{remark}[Gauge-invariance using $A$]\label{remark_GI}\rm
Gauge invariance can also be formulated in this setting, using the Abelian group $C^\infty(\mathcal{M})$ acting on the bundle $ \wedge ^1T^* \mathcal{M} \rightarrow \mathcal{M} $ of spacetime electromagnetic potentials as $ A_x \mapsto A_x + {\rm d} k(x)$, where $ k \in C^\infty( \mathcal{M} )$. It is expressed analogously to \eqref{gauge_invariance}, now in terms of $A$:
\begin{equation}\label{GI_A}
\mathscr{S}'( \Phi ,A+ {\rm d} k)- \mathscr{S}'( \Phi ,A)=C',
\end{equation}
with
\[
\mathscr{S}'(\Phi,A)=\int_ \mathcal{D} \mathscr{L}'(j^1 \Phi   ,  A \circ \Phi    , {\rm d} A   \circ \Phi   , W, K , \gamma  \circ \Phi  ).
\]
When combined with the covariance properties under material and spacetime diffeomorphisms, the invariance group of the theory becomes $ \operatorname{Diff}( \mathcal{D} )\times \left(\operatorname{Diff}( \mathcal{M} ) \,\circledS\, C^\infty( \mathcal{M} ) \right)$, with the semidirect product acting on spacetime  (see  \S\ref{Geom_setting} for comparison).
\end{remark}


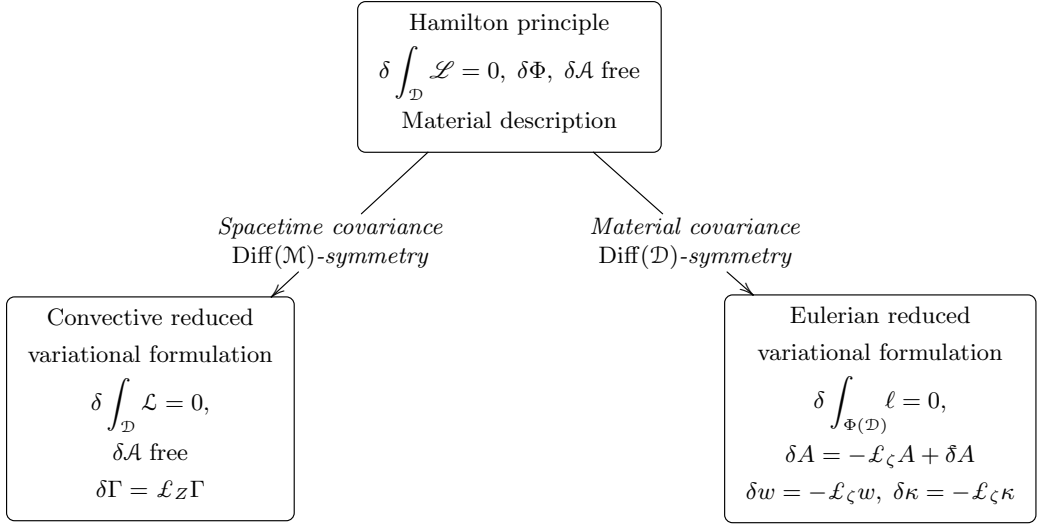
\begin{figure}[h!]
{\noindent
\footnotesize
\begin{center}
\hspace{.3cm}
\begin{xy}
\xymatrix{
& *+[F-:<3pt>]{
\begin{array}{c}
\vspace{0.1cm}\text{Hamilton principle}\\
\vspace{0.1cm}\displaystyle \delta \int_ \mathcal{D}  \mathscr{L}=0,\; \delta \Phi ,\;\delta\mathcal{A} \;\text{free}\\
\vspace{0.1cm}\text{Material description}
\end{array}
}
\ar[ddl]|{\begin{array}{c}\textit{Spacetime covariance} \\
\textit{$ \operatorname{Diff}( \mathcal{M}  )$-symmetry}\\
\end{array}}
\ar[ddr]|{\begin{array}{c}\textit{Material covariance} \\
\textit{$ \operatorname{Diff}( \mathcal{D} )$-symmetry}\\
\end{array}} & & \\
& & & \\
*+[F-:<3pt>]{
\begin{array}{c}
\vspace{0.1cm}\text{Convective reduced}\\
\vspace{0.1cm}\text{variational formulation}\\
\vspace{0.1cm}\displaystyle\delta \int_ \mathcal{D}  \mathcal{L}=0,\\
\vspace{0.1cm}\delta \mathcal{A} \;\text{free}\\
\vspace{0.1cm}\delta \Gamma = \pounds _ Z  \Gamma 
\\
\end{array}
}
& &  *+[F-:<3pt>]{
\begin{array}{c}
\vspace{0.1cm}\text{Eulerian reduced}\\
\vspace{0.1cm}\text{variational formulation}\\
\vspace{0.1cm}\displaystyle\delta \int_ { \Phi  (\mathcal{D})}\!\!\ell =0,\\
\vspace{0.1cm}\delta A = - \pounds _ \zeta A + \deltabar A\\
\vspace{0.1cm}\delta w= - \pounds _ \zeta w, \; \delta \kappa  = - \pounds _ \zeta   \kappa 
\\
\end{array}
}\\
}
\end{xy}
\end{center}
}
\caption{Illustration of the variational principles in the three representations of a relativistic electromagnetic continuum. The primary dynamic fields are the world-tube $\Phi$ and the electromagnetic potential $\mathcal{A}$ (or $A$ see \S\ref{material_spacetime_A}). Given the material and spacetime tensor fields $K$, $W$, and $ \gamma $, one defines the associated dynamic fields $\kappa= \Phi _*K$, $w= \Phi _*W$, and $ \Gamma = \Phi ^* \gamma $.}
\label{figure_2}
\end{figure}




\section{Coupling with the Einstein equations and junction conditions}\label{coupling}

In this section, we shall show how the variational formulation for electromagnetic continua developed in \S\ref{sec_2} can be coupled to the gravitation theory and with the Einstein-Maxwell equations outside the continuum. In particular, the variational formulation that we develop in this section is able to produce:
\begin{itemize}
\item[\bf (1)] the field equations for the gravitational and electromagnetic fields created by the relativistic continuum, both at the interior and outside the continuum;
\item[\bf (2)] the equations of motion of the continuum;
\item[\bf (3)] the junction conditions between the solution at the interior of the relativistic continuum and the solution describing the gravity and electromagnetic fields produced outside from it.
\end{itemize}

It was shown in \cite{GB2024} that the key step to naturally obtain the junction conditions directly from the variational principle is to augment the action functional with the Gibbons-Hawking-York (GHY) boundary terms, \cite{Yo1972}, \cite{GiHa1977}. We shall show here how this setting can be extended to yield the electromagnetic junction conditions.

\paragraph{Geometric setting and preliminary junction conditions.} We shall denote by $ \mathcal{N} ^-= \Phi ( \mathcal{D} )\subset \mathcal{M} $ the portion of spacetime occupied by the continuum, denoted $\mathcal{N} $ earlier, and we shall write $ \mathcal{N} ^+= \mathcal{M} - \operatorname{int}(\mathcal{N}^-)$. For simplicity we assume that the spacetime $ \mathcal{M} $ has no boundary and we have $ \partial \mathcal{N} ^+= \partial \mathcal{N} ^-$ assumed to be piecewise smooth. While $ \partial \mathcal{N}^+$ and $\partial \mathcal{N} ^-$ are equals as manifolds, they have opposite orientation as boundaries of $ \mathcal{N} ^+$ and $ \mathcal{N} ^-$, the latter having the orientation induced from that of $ \mathcal{M} $.

We denote by $\mathsf{g}^-$ and $\mathsf{g}^+$ the Lorentzian metrics on $ \mathcal{N}^-$ and $ \mathcal{N} ^+$, assumed to be smooth. It is assumed that $\mathsf{g}^+$ and $\mathsf{g}^-$ induce the same metric on the boundary, i.e.,
\begin{equation}\label{junction_first_g} 
i_{ \partial \mathcal{N}^-} ^* \mathsf{g}^-= i_{ \partial \mathcal{N} ^+} ^* \mathsf{g}^ +=:h
\end{equation} 
with $i_{ \partial \mathcal{N}^\pm}: \partial \mathcal{N}^\pm \rightarrow \mathcal{N}^\pm $ the inclusions.
We also assume that $h$ is nondegenerate, i.e., either Lorentzian or Riemannian, on each piece of the boundary. Condition \eqref{junction_first_g}, sometimes referred to as the preliminary junction condition, ensures that the boundary has a well-defined geometry.

We denote by $A^-$ and $A^+$ the electromagnetic potentials on $ \mathcal{N} ^-$ and $ \mathcal{N} ^+$, assumed to be smooth. In a similar way with the metric, we assume that $A^-$ and $A^+$ induce the same potential on the boundary, i.e.,
\begin{equation}\label{junction_first_A} 
i_{ \partial \mathcal{N}^-} ^* A^-= i_{ \partial \mathcal{N} ^+} ^* A^ +.
\end{equation} 
Note that this condition does not depend on the metrics. However, by choosing one of the metric, say $\mathsf{g}^+$, it can be written in local coordinates as
\[
[A _  \alpha  ] (\mathsf{g}^+)^{ \alpha \beta  } \mu ^+_{ \beta   \gamma  \nu _1... \nu _{n-1}} (n^+)^ \gamma  =0,
\]
with $[A_ \alpha ]= A_ \alpha ^+- A_ \alpha ^-$, similarly if $ \mathsf{g}^-$ is preferred.

\paragraph{Lagrangian densities and action functional.} In this section we choose to regard the spacetime electromagnetic potential $A$ as the primitive object, rather than $ \mathcal{A} $, see \S\ref{material_spacetime_A}. Hence we are working with $\mathscr{L}'$ instead of $\mathscr{L}$.
We specify the material Lagrangian to the case where the spacetime tensor $ \gamma $ is the Lorentzian metric $\mathsf{g}^-$, hence obtaining a bundle map of the form
\begin{equation}\label{material_Lagrangian_g} 
\mathscr{L} ': J^1( \mathcal{D}  \times \mathcal{M} ) \times \wedge  ^1T^* \mathcal{M} \times \wedge ^2 T^* \mathcal{M}  \times T \mathcal{D} \times T^p_q\mathcal{D}  \times S^2_LT^*\mathcal{M}  \rightarrow \wedge ^{n+1}T^* \mathcal{D},
\end{equation}
where  $ S^2_LT^*\mathcal{M} \rightarrow \mathcal{M} $ is the bundle of Lorentzian metrics. This Lagrangian density is written as $\mathscr{L}'(j^1\Phi, A^- \circ \Phi  , {\rm d}A^-\circ \Phi  , W, K, \mathsf{g}^-\circ \Phi  )$ when evaluated on fields. We assume that $\mathscr{L}'$ is materially covariant with respect to the isotropy subgroup of $(W,K)$ so that there is the associated spacetime Lagrangian density 
\begin{equation}\label{spacetime_Lagran_grav} 
\ell: \wedge  ^1T^* \mathcal{M} \times \wedge ^2 T^* \mathcal{M}  \times T \mathcal{M}  \times  T ^p _q  \mathcal{M}  \times  S^2_LT^*\mathcal{M} \rightarrow \wedge  ^{n+1} T^*\mathcal{M}
\end{equation} 
written as $\ell(A^-, {\rm d} A^-,w, \kappa , \mathsf{g}^-)$  when evaluated on fields.

The total action functional is constructed by adding:
\begin{itemize}
\item[\bf (1)] the action functional of the electromagnetic continuum with Lagrangian density $\ell$ on $ \mathcal{N}^- = \Phi ( \mathcal{D} )$;
\item[\bf (2)] the Einstein-Hilbert action functionals on $ \mathcal{N} ^-$ and $ \mathcal{N} ^ +$;
\item[\bf (3)] the action functional associated to the Maxwell Lagrangian $\ell_{\rm M}$ on $ \mathcal{N} ^+$;
\item[\bf (4)] the GHY boundary terms given by the integral of the trace of the extrinsic curvature (or second fundamental form) $k(\mathsf{g}^\pm)$ on the boundaries.
\end{itemize}
This yields
\begin{equation}\label{total_action}
\begin{aligned} 
&\int_{ \mathcal{N}^- }\! \ell( A^-, {\rm d} A^-,w , \kappa  , \mathsf{g}^- ) + \frac{1}{2 \chi }\int_{ \mathcal{N}^-}  \!R(\mathsf{g}^-) \mu (\mathsf{g}^-) & &\qquad \text{(interior terms)}\\
& +\int_{ \mathcal{N} ^+}\ell_{\rm M}(A^+, {\rm d} A^+, \mathsf{g}^+) + \frac{1}{2 \chi }\int_{\mathcal{N} ^+} \! R(\mathsf{g}^+) \mu (\mathsf{g}^+)& &\qquad \text{(exterior terms)}\\
& + \frac{1}{ \chi }\int_{ \partial \mathcal{N}^-} \!\epsilon\, k(\mathsf{g}^-)   \mu ^-(h)+ \frac{1}{ \chi }\int_{ \partial  \mathcal{N} ^+} \!\epsilon\, k(\mathsf{g}^+)   \mu ^+(h),& &\qquad \text{(boundary terms)}
\end{aligned}
\end{equation}
with the Maxwell Lagrangian on the outer spacetime given as
\begin{equation}\label{Maxwell_Lagrangian} 
\ell_{\rm M}(A^+, {\rm d} A^+, \mathsf{g}^+)= -\frac{1}{2} F^+ \wedge \star F^+.
\end{equation} 
Note that the Maxwell Lagrangian on $ \mathcal{N} ^-$ does not appear explicitly as it is included in the Lagrangian density $\ell$, see \S\ref{EMFS}.
We recall that $R(\mathsf{g})$ is the scalar curvature of $\mathsf{g}$, defined by $R=\mathsf{g}^{ \alpha \beta } Ric_{ \alpha \beta }= \mathsf{g}^{ \alpha \beta } R^ \lambda _{ \alpha \lambda \beta }$ with $R(u,v)= \nabla _u \nabla _v - \nabla _v \nabla _u - \nabla _{[u,v]}$ the Riemann curvature tensor of $\mathsf{g}$, and that $k(\mathsf{g}^\pm)$ is defined on $ \partial \mathcal{N}^\pm$ by $k= \operatorname{Tr}K= h^{ab}K_{ab}$, where $ h_{ab}dx^adx^b$ is the local expression of $h$ on $ \partial \mathcal{N}^\pm $ and $K(\mathsf{g}^\pm)$, $K(u,v)= \mathsf{g}( u, \nabla _v n^\pm)$ is the extrinsic curvature of $ \partial \mathcal{N} ^\pm$. In the last two terms of \eqref{total_action}, $ \epsilon $ is $1$, resp., $-1$, on the timelike, resp., spacelike piece of the boundary, and $ \mu ^\pm(h)$ is the volume form on $ \partial \mathcal{N} ^\pm$ associated to $h= i_{ \partial \mathcal{N} ^-}^*\mathsf{g}^-= i_{ \partial \mathcal{N} ^+}^*\mathsf{g}^+$ and to the boundary orientation of $ \partial \mathcal{N} ^\pm$.

\paragraph{Variational formulation for electromagnetic continuum coupled to gravitation.} We here state and prove our main result concerning the variational formulation for relativistic electromagnetic continuum coupled to gravitation, based on the action functional given in \eqref{total_action}.
As explained earlier, we assume that the Lorentizian metrics and electromagnetic potentials $ \mathsf{g}^+$, $\mathsf{g}^-$, $A^+$, and $A^-$ on which the action functional is evaluated satisfy the preliminary junction conditions \eqref{junction_first_g} and \eqref{junction_first_A} and we derive the gravitational (Lanczos-Israel) and electromagnetic junction conditions as critical point conditions. 
It is also possible to obtain \eqref{junction_first_g} and \eqref{junction_first_A} from the variational formulation without assuming it a priori, by including the corresponding Lagrange multiplier terms. We shall use the notations
\begin{equation}\label{jump_notation} 
[\mathsf{g}]=\mathsf{g}^+- \mathsf{g}^-, \; [K]= K(\mathsf{g}^+)-K(\mathsf{g}^-), \; [F]=F^+-F^-, \; \Big[\frac{\partialnew \ell}{\partialnew F}\Big]=\Big[ \frac{\partialnew \ell_{\rm M}}{\partialnew F^+}-  \frac{\partialnew \ell}{\partialnew F_+}\Big]
\end{equation} 
for the various jumps.

\begin{theorem}\label{main} Consider an electromagnetic continuum with material Lagrangian density given as in \eqref{material_Lagrangian_g} and assumed to be material and spacetime covariant. Let $\ell: \wedge ^1 T ^*\mathcal{M} \times \wedge ^2 T^* \mathcal{M} \times T \mathcal{M}  \times  T ^p _q  \mathcal{M}  \times  S^2_LT^* \mathcal{M} \rightarrow \wedge  ^{n+1} T^*\mathcal{M} $ be the associated spacetime Lagrangian density.

Fix the reference tensor fields $W \in \mathfrak{X} ( \mathcal{D} )$ and $K \in \mathcal{T} _q^p( \mathcal{D} )$.
For each world-tube $ \Phi : \mathcal{D} \rightarrow \mathcal{M} $, define $\mathcal{N}^-= \Phi ( \mathcal{D})$, $ \mathcal{N} ^+= \mathcal{M} - \operatorname{int}( \mathcal{N} ^-)$, $w= \Phi _* W$, and $ \kappa = \Phi _ * K$. Consider smooth Lorentzian metrics $ \mathsf{g}^\pm \in \mathcal{S} ^2_L( \mathcal{N} ^\pm)$ and electromagnetic potentials $A^\pm \in \Omega ^1( \mathcal{N} ^\pm)$ such that $i_{\partial  \mathcal{N} ^-}^* \mathsf{g}^-= i_{ \partial \mathcal{N} ^+}^* \mathsf{g}^+$ and $ i_{ \partial \mathcal{N} ^-}^*A^-=  i_{ \partial \mathcal{N} ^+}^*A^+$.

Then the following statements are equivalent:
\begin{itemize}
\item[\bf (i)] $A^\pm \in \Omega ^1( \mathcal{N} ^\pm)$, $w \in \mathfrak{X} (\mathcal{N}^- )$, $ \kappa \in \mathcal{T} ^p_q( \mathcal{N}^- )$, $\textcolor{black}{\mathsf{g}^\pm }\in \mathcal{S} ^2_L( \mathcal{N} ^\pm)$ are critical points of the {\bfi Eulerian variational principle}
\begin{equation}\label{total_action_variation}
\begin{aligned} 
&\hspace{-0.7cm}\left. \frac{d}{d\varepsilon}\right|_{\varepsilon=0} \left[\int_{ \mathcal{N}^- _\varepsilon} \ell(A_ \varepsilon ^-, {\rm d} A_ \varepsilon ^-, w_\varepsilon , \kappa _\varepsilon ,  \mathsf{g}^- _\varepsilon) + \frac{1}{2 \chi }\int_{ \mathcal{N}^-_\varepsilon}  R( \mathsf{g}^-_\varepsilon) \mu ( \mathsf{g}^-_\varepsilon) \right. \\
& \quad   \quad \left. +\int_{ \mathcal{N} ^+_ \varepsilon }\ell_{\rm M}(A_ \varepsilon ^+, {\rm d} A^+_ \varepsilon , \mathsf{g}^+_ \varepsilon )+ \frac{1}{2 \chi }\int_{\mathcal{N} ^+_\varepsilon}  R( \mathsf{g}^+_\varepsilon) \mu ( \mathsf{g}^+_\varepsilon)\right. \\
&\quad \quad  \qquad  \qquad  \left. + \frac{1}{ \chi }\int_{ \partial \mathcal{N}^-_\varepsilon} \!\epsilon \,k( \mathsf{g}^-_\varepsilon)   \mu ^-(h_\varepsilon)+ \frac{1}{ \chi }\int_{ \partial  \mathcal{N} ^+_\varepsilon} \!\epsilon\, k( \mathsf{g}^+_\varepsilon)   \mu ^+(h_\varepsilon) \right] =0,
\end{aligned}
\end{equation}
for variations
\begin{align*}
\delta \mathcal{N}^- &= \zeta |_{ \partial \mathcal{N} ^-} \big/T \partial \mathcal{N}^-,   & &\delta A^\pm, & &  \delta \mathsf{g}^\pm ,\\
\delta w &= -\pounds _ \zeta w,    & & \delta \kappa = - \pounds _ \zeta \kappa, & &
\end{align*}
where $ \zeta $ is an arbitrary vector field on $ \mathcal{N}^- $ such that $ \zeta |_{ \Phi (\lambda _0, \mathcal{B} )}= \zeta |_{ \Phi ( \lambda _1, \mathcal{B} )}=0$, and where we consider  variations $ \delta \mathsf{g}^\pm$ and $ \delta A^\pm$ connected to $ \zeta $ on $ \partial \mathcal{N} $ via the infinitesimal version of $i_{\partial  \mathcal{N} ^-_ \varepsilon } ^*\mathsf{g}^-_ \varepsilon = i_{ \partial \mathcal{N} ^+_ \varepsilon } ^*\mathsf{g}^+_ \varepsilon $ and $i_{\partial  \mathcal{N} ^-_ \varepsilon } ^*A^-_ \varepsilon = i_{ \partial \mathcal{N} ^+_ \varepsilon } ^*A^+_ \varepsilon$, with $ \delta  \mathsf{g}^\pm|_{ \Phi ( \lambda _0, \mathcal{B} )}=\delta  \mathsf{g}^\pm|_{ \Phi ( \lambda _1, \mathcal{B} )}=0$ and $ \delta A^\pm|_{ \Phi ( \lambda _0, \mathcal{B} )}=\delta  A^\pm|_{ \Phi ( \lambda _1, \mathcal{B} )}=0$.
\item[\bf (ii)] $A^\pm \in \Omega ^1( \mathcal{N} ^\pm)$, $w \in \mathfrak{X} (\mathcal{N}^- )$, $ \kappa \in \mathcal{T} ^p_q( \mathcal{N}^- )$, $\mathsf{g}^\pm \in \mathcal{S} ^2_L( \mathcal{N} ^\pm)$ satisfy the equations and junction conditions
\begin{equation}\label{full_system}
\hspace{-0.7cm}\left\{
\begin{array}{l}
\displaystyle\vspace{0.2cm}\operatorname{div}^ \nabla \! \Big( \frac{\partial \ell}{\partial \mathsf{g}^-} \Big) =0,\quad \pounds _w \kappa =0\quad  \text{on} \quad    \mathcal{N}^-\\
\displaystyle\vspace{0.2cm} Ein (\mathsf{g}^-) \mu (\mathsf{g}^-) = 2 \chi  \frac{\partial \ell}{\partial \mathsf{g}^-}\quad \text{on} \quad    \mathcal{N}^- \qquad \;  Ein(\mathsf{g}^+) =2\chi  \frac{\partial \ell_{\rm M}}{\partial \mathsf{g}^+} \quad   \text{on} \quad    \mathcal{N}^+\\
\displaystyle\vspace{0.4cm} \frac{\partialnew \ell\;\;}{\partialnew A^-} + {\rm d} \frac{\partialnew \ell\;\;}{\partialnew F^-}=0\quad   \text{on} \quad    \mathcal{N}^-\qquad \qquad  \quad  \frac{\partialnew \ell_{\rm M}}{\partialnew A^+} + {\rm d} \frac{\partialnew \ell_{\rm M}}{\partialnew F^+}=0 \quad   \text{on} \quad    \mathcal{N}^+\\
\displaystyle \vspace{0.3cm}i^*_{\partial_{\rm cont} \mathcal{N}} [\mathsf{g}]=0, \quad [K]=0 \quad   \text{on} \quad  \partial_{\rm cont} \mathcal{N}],\quad\text{\rm (Israel-Darmois junctions)}\\
\displaystyle  i^*_{ \partial_{\rm cont} \mathcal{N}}[F]=0,  \quad i^*_{\partial_{\rm cont} \mathcal{N}} \Big[  \frac{\partialnew\ell}{\partialnew F }  \Big] =0 \quad   \text{on} \quad  \partial_{\rm cont} \mathcal{N},\quad\text{\rm (EM junctions)}
\end{array}
\right.
\end{equation}
where $2 \mathsf{g}^- \!\cdot\! \frac{\partial \ell}{\partial \mathsf{g}^-}=  \ell \delta  + w \otimes \frac{\partial \ell}{\partial w} -  \textcolor{black}{  \frac{\partial \ell}{\partial \kappa }\!\therefore\!  \widehat{ \kappa  }} -  \frac{\partial \ell}{\partial A } \otimes  A  -  \frac{\partial \ell}{\partial F }\stackrel{\rm tr}{ \otimes }F$ and $ \nabla $ is the Levi-Civita covariant derivative associated to $ \mathsf{g}^-$.
\end{itemize}
\end{theorem}
\noindent\textbf{Proof.}  We consider arbitrary variations $\Phi _ \epsilon $, $A^\pm_ \varepsilon$, and  $ \mathsf{g}_ \varepsilon ^\pm$ with the fixed endpoint conditions and which satisfy $i_{\partial  \mathcal{N} ^-_ \varepsilon } ^*\mathsf{g}^-_ \varepsilon = i_{ \partial \mathcal{N} ^+_ \varepsilon } ^*\mathsf{g}^+_ \varepsilon $ and $i_{\partial  \mathcal{N} ^-_ \varepsilon } ^*A^-_ \varepsilon = i_{ \partial \mathcal{N} ^+_ \varepsilon } ^*A^+_ \varepsilon $, where $ \mathcal{N} _ \varepsilon ^-= \Phi _ \varepsilon  ( \mathcal{D} )$. As earlier, the constrained variations in {\bf (i)} are induced by the variations $\Phi _ \varepsilon  $ of the world-tube by using the relations $ \mathcal{N} _ \varepsilon = \Phi _ \varepsilon  ( \mathcal{D} )$, $w_ \varepsilon = (\Phi_ \varepsilon  )_*W$, $ \kappa _ \varepsilon = (\Phi_ \varepsilon  )_*K$, and defining $ \zeta= \delta \Phi \circ \Phi ^{-1}$.

\noindent(1) \textit{Metric variation}: We first fix the world-tube and take the variation with respect to the Lorentzian metrics $ \mathsf{g}^\pm$. This has been done in \cite{GB2024} and yields the Einstein equations in $ \mathcal{N} ^+$ and $ \mathcal{N} ^-$ as well as the junction condition on $ \partial _{\rm cont} \mathcal{N} $:
\begin{align} 
\frac{\partial \ell}{\partial  \mathsf{g}^-}&= \frac{1}{2\chi}  \textcolor{black}{Ein} ( \mathsf{g}^-) \mu ( \mathsf{g}^-) \quad \text{on $ \mathcal{N} ^-$}\label{cond_1}\\
\frac{\partial \ell_{\rm M}}{\partial  \mathsf{g}^+}&= \frac{1}{2\chi}  \textcolor{black}{Ein} ( \mathsf{g}^+) \mu ( \mathsf{g}^+) \quad \text{on $ \mathcal{N} ^+$}\label{cond_2}\\
0&= [K]  \quad \text{on $\textcolor{black}{\partial _{\rm cont}\mathcal{N}}$}\label{cond_3}. \textcolor{white}{\frac{1}{2} }
\end{align}

\noindent(2) \textit{Electromagnetic potential variation}: Taking the variations with respect to $ \delta A^\pm$ gives
\begin{align*} 
&\int_ {\mathcal{N}^-} \delta A^- \wedge \Big(\frac{\partialnew \ell}{\partialnew A^-} + {\rm d} \frac{\partialnew \ell}{\partialnew F^-}\Big) + \int_{ \partial \mathcal{N} ^-} \delta A^- \wedge \frac{\partialnew \ell}{\partialnew F^-} \\
& \qquad + \int_ {\mathcal{N}^+} \delta A^+ \wedge \Big(\frac{\partialnew \ell_{\rm M}}{\partialnew A^+} + {\rm d} \frac{\partialnew \ell_{\rm M}}{\partialnew F^+}\Big) + \int_{ \partial \mathcal{N} ^+} \delta A^+ \wedge \frac{\partialnew \ell}{\partialnew F^+}=0.
\end{align*} 
From the relation $i^*_{ \partial \mathcal{N} ^-} \delta A^-= i^*_{ \partial \mathcal{N} ^+} \delta A^+$ and from the arbitrariness of the variations away from the boundary except vanishing at $ \Phi ( \lambda _0, \mathcal{B} )$ and $ \Phi ( \lambda _1, \mathcal{B} )$, we get the Euler-Lagrange equations on $ \mathcal{N} ^+$ and $ \mathcal{N}  ^-$ as well as the electromagnetic junction condition on $ \partial _{\rm cont} \mathcal{N} $:
\begin{align} 
&\frac{\partialnew \ell\;\;}{\partialnew A^-} + {\rm d} \frac{\partialnew \ell\;\;}{\partialnew F^-}=0,  \quad \text{on $ \mathcal{N} ^-$}\label{cond_4}\\
&\frac{\partialnew \ell_{\rm em}}{\partialnew A^+} + {\rm d} \frac{\partialnew \ell_{\rm em}}{\partialnew F^+}=0 \quad \text{on $ \mathcal{N} ^-$}\label{cond_5}\\
&i^*_{ \partial \mathcal{N} _{\rm cont}} \Big[  \frac{\partialnew\ell}{\partialnew F }  \Big] =0 \quad \text{on $\textcolor{black}{\partial _{\rm cont}\mathcal{N}}$}\label{cond_6},
\end{align} 
see \eqref{jump_notation}.
Recalling the local expression $\frac{\partialnew \ell}{\partialnew F} = \frac{1}{2} \frac{\partial \ell}{\partial F_{ \mu \nu }}d^{n-1}x_{\mu \nu }$, the last junction condition can be equivalently written as
\begin{equation}\label{Junction_EM} 
\Big[\frac{\partial \ell}{\partial F_{ \mu \nu }}\Big] n_ \mu^+ =0 \quad\text{on}\quad \partial \mathcal{N} _{\rm cont}.
\end{equation} 
Also, from the preliminary junction condition $i^*_{ \partial \mathcal{N} _{\rm cont}} [A]=0$ on the electromagnetic potential, we get $i^*_{ \partial \mathcal{N} _{\rm cont}} [F]=0$, which is locally written as
\begin{equation}\label{Junction_F} 
[F_{ \mu \nu }] \zeta ^ \mu \eta ^\nu =0
\end{equation} 
for all vector fields $ \zeta $ and $ \eta $ parallel to $\partial \mathcal{N} _{\rm cont}$, or, equivalently
\[
[F _ { \mu \nu }  ] (\mathsf{g}^+)^{ \mu  \gamma   } (\mathsf{g}^+)^{ \nu  \delta  }\mu ^+_{ \gamma    \delta \sigma   \nu _1... \nu _{n-2}} (n^+)^ \sigma   =0.
\]

\noindent(3) \textit{World-tube variation}: Because of the conditions $i_{ \partial \mathcal{N} ^+}^* \mathsf{g}^+= i_{ \partial \mathcal{N} ^-}^* \mathsf{g}^-$ and $i_{ \partial \mathcal{N} ^+}^* A^+= i_{ \partial \mathcal{N} ^-}^*A^-$, one cannot a priori vary $ \Phi $ without varying $\mathsf{g}^\pm$ and $A^\pm$. Nevertheless, it was shown in \cite{GB2024} that, at the critical point already derived, it is possible to vary $ \Phi $ independently of $ \mathsf{g}^\pm$. However this is not the case for $ A^\pm$ which we have to include again in our variations. Using the stationary conditions already derived, it turns out that its corresponding contributions vanish.

The variation of the first term in \eqref{total_action_variation} with respect to $ \delta \Phi $ and $ \delta A^-$ gives
\begin{equation}\label{intermediate} 
\begin{aligned} &\int_ {\mathcal{N}^-}   \Big[-  \operatorname{div}^ \nabla \!\Big( \ell \delta   + w \otimes \frac{\partial \ell}{\partial w}-   \frac{\partial \ell}{\partial \kappa }\!\therefore\!  \widehat{ \kappa  }\Big) +   \frac{\partial \ell}{\partial A } : \nabla   A  +  \frac{\partial \ell}{\partial F }: \nabla F \Big] \cdot \zeta \\
& \qquad \qquad  +  \int_{\mathcal{N} ^-} \deltabar A^- \wedge \left(  \frac{  \partialnew\ell}{\partialnew A^-} + {\rm d}  \frac{  \partialnew\ell}{\partialnew F^-}\right) \\
& \qquad \qquad +\int_{ \partial \mathcal{N} ^-}\operatorname{tr}  \Big( \ell \delta + w \otimes \frac{\partial \ell}{\partial w}-  \frac{\partial \ell}{\partial \kappa }\!\therefore\! \widehat{ \kappa  } \Big) \cdot \zeta  +\int_{ \partial \mathcal{N}^- }  \deltabar A^- \wedge  \frac{  \partialnew\ell}{\partialnew A^-}.
\end{aligned}
\end{equation}

The variation of the third term in \eqref{total_action_variation} with respect to $ \delta \Phi $ and $ \delta A^+$ gives
\begin{equation}\label{termB} 
\int_{ \partial \mathcal{N} ^+} \mathbf{i} _ \zeta \ell_{\rm M} + \int_ {\mathcal{N}^+} \delta A^+ \wedge \Big(\frac{\partialnew \ell_{\rm M}}{\partialnew A^+} + {\rm d} \frac{\partialnew \ell_{\rm M}}{\partialnew F^+}\Big) + \int_{ \partial \mathcal{N} ^+} \delta A^+ \wedge \frac{\partialnew \ell_{\rm M}}{\partialnew F^+}.
\end{equation} 
From the stationarity conditions already derived, all the contributions of $  \delta A^\pm$ vanish. Collecting the terms proportional to $ \zeta $ at the interior and using Lemma \ref{crucial_formula}, we get
\begin{equation}\label{fluid_equa} 
\operatorname{div}^ \nabla \!\Big(  \ell \delta  + w \otimes \frac{\partial \ell}{\partial w} -  \textcolor{black}{  \frac{\partial \ell}{\partial \kappa }\!\therefore\!  \widehat{ \kappa  }} -  \frac{\partial \ell}{\partial A } \otimes  A  -  \frac{\partial \ell}{\partial F }\stackrel{\rm tr}{ \otimes }F\Big)  =0.
\end{equation} 

The variation of the second and of the fourth to sixth terms in \eqref{total_action_variation} with respect to the world-tube have been computed in \cite[Theorem 5.6]{GB2024}. They yield boundary terms only. Collecting these boundary terms together with the ones in \eqref{intermediate} and \eqref{termB}, and using the boundary condition $[K]=0$ already obtained, we get
\begin{equation}\label{intermediate_step1}
\begin{aligned}
&\epsilon \,\mathbf{i} _{n^-} \Big(\ell \delta  + w \otimes \frac{\partial \ell}{\partial w} -  \frac{\partial \ell}{\partial \kappa }\!\therefore\!  \widehat{ \kappa }\Big)( \zeta ^\flat , n^-) +\epsilon \,\mathbf{i} _{n^+} \big(\ell _{\rm M} \delta\big)( \zeta ^\flat , n^+)\\
&   + \frac{1}{ \chi } \left(\textcolor{black}{Ein} ( \mathsf{g}^+)(n^+,n^+) -  Ein ( \mathsf{g}^-)(n^-,n^-) \right) \epsilon \,\zeta _-^\perp \mu  ^-(h)=0 ,\;\forall\; \zeta
\end{aligned}
\end{equation}
on $\partial_{\rm cont}\mathcal{N}$. In order to proceed further with \eqref{intermediate_step1}, let us now show that from the boundary conditions already derived via the variations $ \delta A^\pm$, we have the equality
\begin{equation}\label{important}
\begin{aligned} 
&\mathbf{i} _{n^-} \Big(-  \frac{\partial \ell}{\partial A^- } \otimes  A ^- -  \frac{\partial \ell}{\partial F ^-}\stackrel{\rm tr}{ \otimes }F^-\Big)( \zeta , (n^-)^\flat ) \\
& \qquad \qquad + \mathbf{i} _{n^+} \Big(-  \frac{\partial \ell_{\rm M}}{\partial A^+ } \otimes  A^+  -  \frac{\partial \ell_{\rm M}}{\partial F ^+}\stackrel{\rm tr}{ \otimes }F^+\Big)( \zeta  , (n^+)^\flat)=0.
\end{aligned}
\end{equation} 
First we observe that the terms involving $A^-$ and $A^+$ vanish since $ \frac{\partial \ell_{\rm M}}{\partial A^+} =0$ and since in local coordinates, we have
\[
\Big(\frac{\partial \ell}{\partial A^- } \otimes  A ^-\Big)( \zeta , (n^-)^\flat )= \frac{\partial \ell}{\partial A^- _ \mu } A ^-_ \nu n^-_ \mu \zeta ^ \nu =0
\]
when the boundary condition $ \frac{\partial \ell}{\partial A^- _ \mu } n^-_ \mu=0$ holds, which is the case for the class of Lagrangian densities relevant for electromagnetic continua, see \S\ref{EMFS}. To show \eqref{important} it remains to show the equality
\[
\frac{\partial \ell}{\partial F^-_{ \mu \gamma }} F^-_{ \nu \gamma } n^-_ \mu \zeta ^ \nu   + \frac{\partial \ell_{\rm M}}{\partial F^+_{ \mu \gamma }} F^+_{ \nu \gamma } n^+_ \mu \zeta ^ \nu =0.
\]
This follows in a rather subtle way as
\begin{align*}
&\left( \frac{\partial \ell}{\partial F^-_{ \mu \gamma }} F^-_{ \nu \gamma } -  \frac{\partial \ell_{\rm M}}{\partial F^+_{ \mu \gamma }} F^+_{ \nu \gamma }\right)  n^-_ \mu \zeta ^ \gamma\\
&=  \frac{\partial \ell}{\partial F^-}(( n^-) ^\flat , dx^ \gamma  ) F^-( \zeta , \partial _ \gamma )- \frac{\partial \ell_{\rm M} }{\partial F^+}(( n^-) ^\flat , dx^ \gamma  ) F^+( \zeta , \partial _ \gamma )  \\
&=\sum_{a=1}^n \frac{\partial \ell}{\partial F^-}(( n^-) ^\flat ,  \eta _ a^ \flat   ) F^-( \zeta , \eta _a )+ \frac{\partial \ell}{\partial F^-}(( n^-) ^\flat ,  (n^-)^ \flat   ) F^-( \zeta ,n_- )\\
& \qquad -\sum_{a=1}^n \frac{\partial \ell_{\rm M} }{\partial F^+}(( n^-) ^\flat , \eta _ a^ \flat  ) F^+( \zeta ,\eta _a )- \frac{\partial \ell_{\rm M} }{\partial F^+}(( n^-) ^\flat , ( n^-) ^\flat ) F^+( \zeta ,n_-)\\
&=\sum_{a=1}^n \frac{\partial \ell}{\partial F^-}(( n^-) ^\flat ,  \eta _ a^ \flat   ) F^-( \zeta , \eta _a )-\sum_{a=1}^n \frac{\partial \ell_{\rm M} }{\partial F^+}(( n^-) ^\flat , \eta _ a^ \flat  ) F^+( \zeta ,\eta _a )\\
&=\sum_{a=1}^n \frac{\partial \ell}{\partial F^-}(( n^-) ^\flat ,  \eta _ a^ \flat   ) \left( F^-( \zeta , \eta _a ) -  F^+( \zeta ,\eta _a ) \right) =0,
\end{align*}
where we selected an orthonormal basis $(n^-, \eta _a)$ with $ \eta _a$ parallel to $ \partial \mathcal{N} _{\rm cont}$, $a=1,...,n$.
In the fourth equality, we used \eqref{Junction_EM} and in the last equality we used \eqref{Junction_F}, which requires both $ \zeta $ and $ \eta _a$ to be parallel to $ \partial \mathcal{N} $, which is the case. This is the reason why we had to reformulate the trace operation, originally computed by summing along $dx^ \gamma $ and $ \partial _ \gamma $, via a summation along an orthonormal basis, so that \eqref{Junction_EM} can be exploited.

Now, thanks to \eqref{important}, equation \eqref{intermediate_step1} can be equivalently written as
\begin{equation}\label{intermediate_step1_bis}
\begin{aligned}
&\epsilon \,\mathbf{i} _{n^-} \Big(\ell \delta  + w \otimes \frac{\partial \ell}{\partial w} -  \frac{\partial \ell}{\partial \kappa }\!\therefore\!  \widehat{ \kappa } -  \frac{\partial \ell}{\partial A^- } \otimes  A ^- -  \frac{\partial \ell}{\partial F ^-}\stackrel{\rm tr}{ \otimes }F^-\Big)( \zeta ^\flat , n^-) \\
&+\epsilon \,\mathbf{i} _{n^+} \Big(\ell _{\rm M} \delta -  \frac{\partial \ell_{\rm M}}{\partial A^+ } \otimes  A^+  -  \frac{\partial \ell_{\rm M}}{\partial F ^+}\stackrel{\rm tr}{ \otimes }F^+ \Big)( \zeta ^\flat , n^+)\\
&   + \frac{1}{ \chi } \left(\textcolor{black}{Ein} ( \mathsf{g}^+)(n^+,n^+) -  Ein ( \mathsf{g}^-)(n^-,n^-) \right) \epsilon \,\zeta _-^\perp \mu  ^-(h)=0 ,\;\forall\; \zeta.
\end{aligned}
\end{equation}
From the spacetime covariance of $\ell$ and $\ell_{\rm em}$ this is just
\begin{equation}\label{intermediate_step1_bis}
\begin{aligned}
&\epsilon \,\mathbf{i} _{n^-} 2  \mathsf{g}^- \!\cdot\! \frac{\partial \ell}{\partial  \mathsf{g}^-}( \zeta ^\flat , n^-)+\epsilon \,\mathbf{i} _{n^+} 2  \mathsf{g}^+ \!\cdot\! \frac{\partial \ell_{\rm M}}{\partial  \mathsf{g}^+}( \zeta ^\flat , n^+)\\
&   + \frac{1}{ \chi } \left(\textcolor{black}{Ein} ( \mathsf{g}^+)(n^+,n^+) -  Ein ( \mathsf{g}^-)(n^-,n^-) \right) \epsilon \,\zeta _-^\perp \mu  ^-(h)=0 ,\;\forall\; \zeta.
\end{aligned}
\end{equation}
Hence, using the Einstein equations \eqref{cond_1} and \eqref{cond_2} already derived, this can be written as
\begin{align*} 
&Ein ( \mathsf{g}^-)( \zeta  , n^-)  \mu ^-(h) -  Ein ( \mathsf{g}^-)(n^-,n^-) \zeta _-^\perp \mu  ^-(h)\\
&=Ein ( \mathsf{g}^+)( \zeta  , n^+)  \mu ^+(h) -  Ein ( \mathsf{g}^+)(n^+,n^+) \zeta _+^\perp \mu  ^+(h) ,\;\forall\; \zeta \;\; \text{on}\;\; \textcolor{black}{\partial _{\rm cont}\mathcal{N}}.
\end{align*} 
This is equivalent to the O’Brien-Synge condition
\[
i_{ \partial \mathcal{N} } ^* \left( \mathbf{i} _{n^-}[\textcolor{black}{Ein} ] \right) =0.
\]
It turns out that this condition is a consequence of the junction condition already derived, as it is easily seen by using the Gauss-Codazzi equation.

Collecting the stationarity conditions obtained in \eqref{cond_1}--\eqref{cond_6} and \eqref{fluid_equa} as well as the preliminary junction conditions, we get  \eqref{full_system}. $\qquad\blacksquare$

\section{Electromagnetic fluids and solids}\label{EMFS}

In this section, we apply the variational framework developed above to electromagnetic fluids and solids. A key insight of our approach is the \textit{systematic and unambiguous derivation of the stress-energy-momentum tensor} for these electromagnetic systems. 
Furthermore, we derive the boundary and junction conditions for these systems in a straightforward manner by directly applying the general results established earlier. Consistent with most of the literature on this subject (e.g., \cite{ErMa1990}), we adopt the Heaviside-Lorentz units in our analysis.

\subsection{Electromagnetic fields}\label{subsec_EMfields}

\paragraph{Magnetic and electric components.} We recall that the decomposition of the Faraday $2$-form $F \in \Omega ^2 ( \mathcal{M} )$ with respect to an observer with world-velocity $u$, $\mathsf{g}(u,u)=- c ^2 $, into its electric and magnetic components is given by:
\begin{equation}\label{decomposition_F} 
F= \frac{1}{c}    u ^\flat \wedge E - \frac{1}{c} \star ( u ^\flat \wedge B)  , 
\end{equation} 
where the electric and magnetic fields are defined as:
\begin{equation}\label{E_B}
E= - \frac{1}{c} \mathbf{i} _u F \in \Omega ^1 ( \mathcal{M} ), \qquad  B= -\frac{1}{c}  \mathbf{i} _ u ( \star F) \in \Omega ^{n-2}( \mathcal{M} ).
\end{equation}
In coordinates, equations \eqref{decomposition_F} and \eqref{E_B} take the form:
\[
F_{ \alpha \beta }=  \frac{1}{c}( u_ \alpha E_ \beta - u_ \beta E_ \alpha)  -   \frac{1}{c}\frac{1}{(n-2)!} u ^ \nu  B^{ \nu _1... \nu _{n-2}} \mu(\mathsf{g})_{ \nu \nu _1... \nu _{n-2} \alpha \beta}
\]
\[
E_ \nu =  \frac{1}{c}F_{ \nu \mu } u^ \mu , \qquad B_ { \nu _1... \nu _{n-2}}=  -\frac{1}{c} \frac{1}{2} F ^ {\nu \mu } u^ { \alpha  }\mu ( \mathsf{g}) _{ \nu\mu  \alpha    \nu _1...  \nu _{n-2}} \in \Omega ^{n-2}( \mathcal{M} ).
\]

Similarly, we consider the decomposition of the partial derivative $ \frac{\partialnew \ell}{\partialnew F} \in \Omega ^{n-1}( \mathcal{M} )$ given by
\begin{equation}\label{dellldF1} 
\frac{\partialnew \ell}{\partialnew F} = - \frac{1}{c}\star ( u ^\flat \wedge D) - \frac{1}{c}u ^\flat \wedge H,
\end{equation} 
where the electric displacement field $D$ and magnetic intensity field $H$ are defined as:
\begin{equation}\label{D_H} 
D:=-\frac{1}{c}\mathbf{i} _u   \Big(\!\!\star  \frac{\partialnew \ell}{\partialnew F} \Big) \in \Omega ^1( \mathcal{M} )\quad\text{and}\quad H:= \frac{1}{c}\mathbf{i} _u  \frac{\partialnew \ell}{\partialnew F} \in \Omega ^{n-2} ( \mathcal{M} ).
\end{equation}
In coordinates, \eqref{dellldF1} and \eqref{D_H} read
\[
\Big( \frac{\partialnew \ell}{\partialnew F} \Big)  _{ \alpha _1... \alpha _{n-1}}  =  - \frac{1}{c}u^ \alpha D^ \beta \mu (\mathsf{g})_{ \alpha \beta \alpha _1... \alpha _{n-1}} - \frac{1}{c}\sum_{k=1}^{n-1}(-1)^{k-1} u_{ \alpha _k} H_{ \alpha _1 ...\check{ \alpha _k} ... \alpha _{n-1}}
\]
\[
D_ \alpha = -\frac{1}{c}\frac{1}{(n-1)!} u^ \nu   \Big( \frac{\partialnew \ell}{\partialnew F} \Big) ^{ \gamma_1... \gamma _{n-1}} \hspace{-0.2cm}\mu (\mathsf{g})_{ \gamma _1... \gamma _{n-1} \nu \alpha }, \qquad H_{ \alpha _1... \alpha _{n-2}}= \frac{1}{c}u^ \gamma \Big( \frac{\partialnew \ell}{\partialnew F} \Big) _{ \gamma  \alpha _1... \alpha _{n-2}} .
\]
It will also be useful to consider the decomposition of the other incarnation of the partial derivative of the Lagrangian density with respect to $F$, namely, the one given by the 2-multivector field density $ \frac{\partial \ell}{\partial F}  \in  \mathfrak{X} ^2_d( \mathcal{M} )$, see \eqref{first_incarnation} versus \eqref{secondincarnation}. In this case, we find
\begin{equation}\label{dellldF2} 
\frac{\partial \ell}{\partial F} = - \frac{1}{c} (u \wedge D^\sharp) \mu (\mathsf{g}) + \frac{1}{c}\star ( u \wedge H ^\sharp  ) \mu (\mathsf{g}),
\end{equation}
where the Hodge star operator is naturally induced on $k$-vector fields.

Some identities that are useful to establish these formulas are:
\[
\mathbf{i} _u(\star  \omega ) =(-1)^k \star ( u ^\flat \wedge \omega ), \quad 
\star \mathbf{i} _u \omega = (-1)^{k-1} u ^\flat \wedge \star \omega  , \quad \star\star \omega = - (-1)^{k(n+1-k)} \omega ,
\]
for $ \omega \in \Omega ^k ( \mathcal{M} )$.


\paragraph{Maxwell Lagrangian.} To illustrate these decompositions and our conventions in the simplest case, we consider the Maxwell Lagrangian density:
\begin{equation}\label{derivative_Maxwell} 
\ell_{\rm M}( A, F, \mathsf{g})= - \frac{1}{2} F \wedge \star F=\frac{1}{2}  E \wedge \star E -\frac{1}{2} B \wedge \star B,
\end{equation} 
where the second equality follows from \eqref{decomposition_F}.
The partial derivative with respect to $F$ is found as:
\begin{equation}\label{partial_Maxwell} 
\begin{aligned} 
\frac{\partialnew \ell_{\rm M}}{\partialnew F} &= - \star F = -\frac{1}{ c } \star(  u ^\flat \wedge E)   -  \frac{1}{ c }  u ^\flat \wedge B\\
\frac{\partial \ell_{\rm M}}{\partial F} &= - F^\sharp \mu (\mathsf{g})= - \frac{1}{c} (u \wedge E^\sharp) \mu (\mathsf{g}) + \frac{1}{c}\star ( u \wedge B ^\sharp  ) \mu (\mathsf{g}),
\end{aligned} 
\end{equation} 
so that \eqref{dellldF1} recovers the usual vacuum relations $D=E$, $H=B$. In terms of the Lagrangian, following our earlier derived formulas, the stress-energy-momentum tensor takes the general form:
\begin{equation}\label{general_Maxwell_SEM} 
\mathfrak{T}_{\rm M} = \ell_{\rm M}\delta -\frac{\partial \ell_{\rm M}}{\partial F} \stackrel{\rm tr}{\otimes} F.
\end{equation} 
It gives the well-known expression of the Maxwell  stress-energy-momentum tensor given by $\mathfrak{T}_{\rm M}= \mathsf{t}_{\rm M} \mu (\mathsf{g})$ with
\[
\mathsf{t}_{\rm M} = - \frac{1}{2} |F| ^2  \delta  + F ^\sharp  \stackrel{\rm tr}{\otimes} F , \qquad  \mathsf{t}_{\rm M}{}^ \mu _ \nu = - \frac{1}{4} F_{ \alpha \beta }F^{ \alpha \beta } \delta ^ \mu _ \nu + F^{ \mu \alpha  } F_{ \nu \alpha  },
\]
where $|F|^2= \mathsf{g}(F,F)$.
Using equations \eqref{decomposition_F} and \eqref{partial_Maxwell}, the tensor can also be expressed in terms of $E$ and $B$ as:
\begin{equation}\label{Maxwell_T}
\begin{aligned} 
\mathsf{t}_{\rm M}&=   \frac{1}{2} \left(  |E| ^2 +  |B| ^2 \right) \frac{1}{c^2}u \otimes u ^\flat +\frac{1}{c} u \otimes S_{\rm M}+ \frac{1}{c}  S_{\rm M}^\sharp  \otimes  u ^\flat \\
& \qquad -  E ^\sharp  \otimes E- B ^\sharp  \stackrel{\rm tr}{\otimes} B +  \frac{1}{2} \left(|E| ^2 +   |B| ^2 \right) \mathsf{P} ,
\end{aligned}
\end{equation} 
where we define the Maxwell Poynting one-form and the projection tensor:
\begin{equation}\label{S_P} 
S_{\rm M}= (-1)^n\frac{1}{c}  \mathbf{i} _{E^\sharp} \mathbf{i} _u (\star B) \in \Omega ^1 ( \mathcal{M} ), \qquad \mathsf{P}= \delta + \frac{1}{c^2} u \otimes u ^\flat \in \mathcal{T} ^1_1( \mathcal{M} ).
\end{equation} 


\subsection{General relativistic electromagnetic fluids}\label{GR_em_fluid}

In this section, we develop the variational framework from \S\ref{sec_2} and \S\ref{coupling} for relativistic electromagnetic fluids. As outlined in \cite{GB2024}, see also \S\ref{Geom_setting}, the material tensor fields describing a relativistic fluid include the vector field $\partial _ \lambda \in \mathfrak{X} ( \mathcal{D} )$ and the volume forms $R, S \in \Omega ^{n+1}( \mathcal{D} )$, which satisfy the condition \eqref{condition_R_S}, as expressed by \eqref{RR_0SS_0}. The corresponding spacetime Eulerian quantities are the generalized velocity $w$, generalized mass density $ \varrho $, and generalized entropy density $ \varsigma $, defined via the pull-back:
\begin{equation}\label{relation_fluid_0} 
w= \Phi _* W, \qquad \varrho = \Phi _* R, \qquad \varsigma= \Phi_*S.
\end{equation}
The only spacetime tensor involved is the Lorentzian metric $\mathsf{g}$.
From these, the world-velocity, the proper mass and entropy densities, and the proper specific entropy are defined as
\begin{equation}\label{relation_fluid_1}
u = \frac{cw}{\sqrt{-
\mathsf{g}(w,w)}} , \quad \rho  = \frac{\sqrt{-\mathsf{g}(w,w)}\;\varrho}{c \;\mu (\mathsf{g})} , \quad s  = \frac{\sqrt{-\mathsf{g}(w,w)}\; \varsigma}{c \;\mu (\mathsf{g})} , \quad \eta = \frac{s}{\rho    }=  \frac{\varsigma }{\varrho }.
\end{equation}
Besides these reference tensor fields, the theory also involves the electromagnetic potential $ \mathcal{A} \in \Omega ^1 ( \mathcal{D} )$ in the material frame.

\subsubsection{Lagrangian variational formulation}\label{Lagrangian_setps}

\paragraph{Lagrangian densities for electromagnetic fluids.} In the setting described above, the Lagrangian density in the material description depends on the following variables: 
\[
\mathscr{L}( j^1 \Phi , \mathcal{A} , {\rm d} \mathcal{A} ,\partial _ \lambda , R,S,\mathsf{g} \circ \Phi ).
\]
Its concrete expression for electromagnetic fluids is most easily obtained by following these steps:
\begin{itemize}
\item[\bf(1)] Writing the continuum analog to the Lagrangian density
\[
\mathscr{L}_{\rm part}(x, \dot  x)=  - mc \sqrt{- \mathsf{g}( \dot  x, \dot  x)} \,{\rm d}  \lambda - e A (x) \cdot  \dot  x \,{\rm d}  \lambda 
\]
for a relativistic point particle of mass $m$ and charge $e$;
\item[\bf(2)] Including the internal ``energy" contribution $\mathscr{E}$ of the matter;
\item[\bf(3)] Adding the Maxwell Lagrangian density \eqref{derivative_Maxwell}.
\end{itemize}
Recalling that the word-tube describes a continuum collection of world-lines and using the mass form $R_0$ on $ \mathcal{B} $, one obtains from steps {\bf(1)}-{\bf(2)}-{\bf(3)} the following expression for the Lagrangian density $\mathscr{L}$:
\begin{equation}\label{general_expression_L} 
\left[  - \frac{1}{c}\sqrt{ -\mathsf{g}  \big(\dot \Phi , \dot \Phi  \big) }\; \big(c^2 + \mathscr{E}\big)- q A( \Phi ) \cdot \dot  \Phi  \right] \; {\rm d}  \lambda \wedge R_0 - \frac{1}{2  } \Phi ^* (  F \wedge \star F),
\end{equation}
with $q$ the specific charge.
Then, the explicit formula for $\mathscr{L}( j^1 \Phi , \mathcal{A} , {\rm d} \mathcal{A} ,\partial _ \lambda , R,S,\mathsf{g} \circ \Phi )$ is obtained from the choice of an ``energy" function $\mathscr{E}$\footnote{As we will see, in the presence of electromagnetic effects this is not to be interpreted as the internal energy, but as a Legendre transformed version of it with respect to the electric field, see Remark \ref{energy_interpretation}.} which must take into account of the presence of the matter, here the compressible fluid, as well as of the interactions between this fluid and the electromagnetic fields, from which follows the notion of polarization and magnetization.

Recall that for ordinary non-charged fluids, this function $\mathscr{E}$ is constructed from the state equation $e=e( \rho  , \eta )$ of the fluid in which the mass density $ \rho  $ and specific entropy $ \eta $ are expressed in terms of the material variables $R$, $S$, $ \Phi $ by using the relations \eqref{relation_fluid_0} and \eqref{relation_fluid_1}. This gives the following fluid Lagrangian density
\begin{equation}\label{L_fluid}
\begin{aligned} 
&\mathscr{L}_{\rm fluid}( j^1 \Phi ,  \partial _ \lambda , R,S,\mathsf{g} \circ \Phi )\\
&= - \frac{1}{c}\sqrt{ -\mathsf{g}  \big(\dot \Phi , \dot \Phi  \big) }\; \Big(c^2 +e \Big( \frac{1}{c} \sqrt{-  \mathsf{g}( \dot \Phi   ,  \dot \Phi  )}\frac{R }{\Phi ^* [\mu (\mathsf{g})]} ,\frac{ S }{R} \Big) \Big){\rm d}  \lambda \wedge R_0
\end{aligned}
\end{equation}
derived in \cite{GB2024}.

For electromagnetic fluids, this function $\mathscr{E}$ has to be extended to include dependence on the electric and magnetic fields. As we shall see, the construction of such a function naturally also involves the relativistic right Cauchy-Green tensor $C$ defined by
\begin{equation}\label{def_F_C} 
C_{AB}= \mathsf{g}_{ \mu  \nu  } F^ \mu  _A F^ \nu  _B, \qquad F^ \mu  _A =\mathsf{P}^ \mu _ \lambda {\Phi _{,A} }^ \lambda,
\end{equation} 
see \cite{GrEr1966a,ErMa1990}. Its intrinsic geometric definition is
\begin{equation}\label{intrinsic_C} 
C= \Phi ^* \mathsf{p}, \qquad \mathsf{p}:=\textcolor{black}{ \mathsf{g}} + \frac{1}{c ^2 } u ^\flat \otimes u ^\flat ,
\end{equation} 
see \cite{GB2024}, where $\mathsf{p}$ is referred to as the projection tensor or Landau-Lifshitz radar metric, related to $\mathsf{P}$ in \eqref{S_P} by lowering the upper index. In the most general case, the internal energy function for electromagnetic fluids thus takes the form $\mathscr{E}( \rho  , \eta ,  \mathcal{E} ,  \mathcal{B} , C )$, where
\begin{equation}\label{E_B_material} 
\mathcal{E} :=\Phi ^*E= -\frac{c}{\sqrt{-\mathsf{g}( \dot \Phi , \dot  \Phi )}} \left( \mathbf{i} _{ \partial _ \lambda } \mathcal{F} \right) , \quad \mathcal{B} :=\Phi ^*B= -\frac{c}{\sqrt{-\mathsf{g}( \dot \Phi , \dot  \Phi )}} \left( \mathbf{i} _{ \partial _ \lambda } \star\mathcal{F} \right),
\end{equation} 
are the electric and magnetic parts of the Faraday $2$-form $ \mathcal{F} = \Phi ^* F$ in the material frame. The Hodge star operator is associated with $\Phi^*\mathsf{g}$. In this case, the material Lagrangian density \eqref{general_expression_L} takes the explicit form
\begin{equation}\label{L_EMF} 
\begin{aligned} 
&\mathscr{L}( j^1 \Phi , \mathcal{A} , {\rm d} \mathcal{A} ,\partial _ \lambda , R,S,\mathsf{g} \circ \Phi )\\
&=\Big[  - \frac{1}{c}\sqrt{ -\mathsf{g}  \big(\dot \Phi , \dot \Phi  \big) }\; \big(c^2 +\mathscr{E} \big( \cdot , \cdot  , \mathcal{E} ,  \mathcal{B} , C\big) \big)- q  \mathcal{A} \cdot \partial_ \lambda  \Big] {\rm d}  \lambda \wedge R_0- \frac{1}{2} \mathcal{F}  \wedge \star \mathcal{F} ,
\end{aligned}
\end{equation} 
where the first two slots of $\mathscr{E}$ are the same as those of $e$ in \eqref{L_fluid}. 
One easily checks that, as it should, this Lagrangian density is spacetime covariant. The expression of the associated convective Lagrangian density $ \mathcal{L} (  \mathcal{A} , {\rm d} \mathcal{A}, \partial _ \lambda, R,S,\Gamma )$ is easily derived. Furthermore, \eqref{L_EMF} is materially covariant with respect to $ \operatorname{Diff}_{\partial _ \lambda }( \mathcal{D} )$ if and only if the function $ \mathscr{E} $ satisfies the condition
\begin{equation}\label{isotropy_E} 
\mathscr{E} ( \rho  \circ \psi  , \eta \circ\psi , \psi ^* \mathcal{E} , \psi ^* \mathcal{B} , \psi ^* C)= \mathscr{E} ( \rho  \circ \psi  , \eta , \mathcal{E} , \mathcal{B} , C) \circ \psi , \quad  \text{for all $ \psi \in \operatorname{Diff}( \mathcal{B} )$}\footnote{No confusion can arise while we use the same letter $\mathcal{B}$ for the magnetic field $\mathcal{B}=\Phi^*B$ in the material frame and the reference spatial domain $\mathcal{B}$ in $\mathcal{D}=[\lambda_0,\lambda_1]\times\mathcal{B}$},
\end{equation} 
which is an isotropy condition on the electromagnetic dependence which is usually assumed for electromagnetic fluids, see \S\ref{anisotropic} for anisotropic electromagnetic continua. The fact that it is enough to verify \eqref{isotropy_E} for $\operatorname{Diff}( \mathcal{B} )$ instead of the larger group $\operatorname{Diff}_{\partial_\lambda}( \mathcal{B} )$ follows from the special form of $\mathcal{E}$, $\mathcal{B}$, and $C$, namely, they vanish along $\partial_\lambda$. Recall from \S\ref{spacetime_reduction} that material covariance with respect to $ \operatorname{Diff}_{\partial _ \lambda }( \mathcal{D} )$ is enough to define the corresponding spacetime Lagrangian density $\ell$, see Remark \ref{Diff_KW}. It is obtained via the general formula \eqref{def_ell}, given here by 
\begin{equation}\label{fluid_L_ell}
\mathscr{L}(j^1 \Phi , \mathcal{A} , {\rm d} \mathcal{A} , \partial _ \lambda  , R, S, \mathsf{g} \circ \Phi )= \Phi ^* [\ell(A, {\rm d} A, w, \varrho , \varsigma , \mathsf{g})].
\end{equation}
We hence get
\begin{equation}\label{ell_electromagnetic_fluid} 
\ell(A, F, w, \varrho , \varsigma , \mathsf{g})=- \rho  \left(  c^2 + e( \rho , \eta , E,B, \mathsf{p} )  + q A \!\cdot\! u \right)  \mu (\mathsf{g})  -\frac{1}{2} F \wedge \star F,
\end{equation} 
where the spacetime ``energy" function $e$ is
\begin{equation}\label{def_e} 
e (\rho , \eta , E,B, \mathsf{p})= \mathscr{E}( \rho \circ \Phi  , \eta \circ \Phi ,  \Phi ^*E, \Phi ^*B, \Phi ^*  \mathsf{p} ) \circ \Phi ^{-1} 
\end{equation}
for any world-tube $ \Phi : \mathcal{D} \rightarrow \mathcal{N} $ with $ \Phi _* \partial _ \lambda = w$. From the spacetime covariance of the material Lagrangian \eqref{L_EMF}, the spacetime Lagrangian density satisfies \eqref{double_covariance} and, in particular, the function $e$ has the property
\begin{equation}\label{cov_e} 
e (\rho \circ \psi , \eta \circ \psi, \psi ^*E,\psi ^*B, \psi ^*\mathsf{p})= e( \rho  , \eta , E, B, \mathsf{p}) \circ \psi , \quad \forall\psi \in \operatorname{Diff}(\mathcal{M} ). 
\end{equation}
As we shall see later (Remark \ref{no_p}), the dependence of $e$ on $\mathsf{p}$ can be replaced by a dependence on $\mathsf{g}$ due to the properties $ \mathbf{i} _uE=0$, $ \mathbf{i} _uB=0$.
A similar conclusion holds for solids when the Cauchy deformation tensor $\mathsf{c}$ is included. Hence, from now on we express the energy in terms of $\mathsf{g}$.

\paragraph{Formulation in terms of densities.} The choice of describing the Lagrangian density in terms of a function $e$ with the dependences $e=e( \rho  , \eta, E, B, \mathsf{p})$ is made in accordance with most of the literature\footnote{Especially concerning the use of $ \rho  $, $ \eta $, $E$, and $B$, while the dependence on $\mathsf{p}$ is usually not explicitly mentioned.}. Also, with this choice, the passing from the material variables $ \mathcal{E} , \mathcal{B} , C$ to the spacetime variables $E, B, \mathsf{p}$ is quite transparent, see \eqref{def_e}. However, in our derivation of the stress-energy-momentum tensor from the variational approach, it will be more convenient to rewrite \eqref{ell_electromagnetic_fluid} in the equivalent form
\begin{equation}\label{Lagrangian_epsilon} 
\ell(A, {\rm d} A,w, \varrho , \varsigma , \mathsf{g} )= - \epsilon ( \rho  , s, E, B, \mathsf{g}) \mu(\mathsf{g})- q \varrho A\! \cdot \! w,
\end{equation} 
where we defined the density $\epsilon $ as
\begin{equation}\label{decomposition_e} 
\epsilon ( \rho  , s, E, B, \mathsf{g})= \epsilon _{\rm m} ( \rho  , s, E, B, \mathsf{g})+\epsilon _{\rm M}( E,B, \mathsf{g} ),
\end{equation} 
with $ \epsilon _{\rm m}$ and $ \epsilon _{\rm M}$ referring to the matter and Maxwell parts of $ \epsilon $:
\begin{equation}\label{e_m_e_M} 
\begin{aligned}
\epsilon _{\rm m}( \rho  , s, E, B, \mathsf{g})\mu ( \mathsf{g} )&= \rho  ( c ^2 + e( \rho  , \eta , E, B, \mathsf{g} )) \mu ( \mathsf{g} )\\
\epsilon _{\rm M}( E,B, \mathsf{g} )\mu ( \mathsf{g} )&= - \frac{1}{2}  E \wedge \star E +\frac{1}{2} B \wedge \star B.
\end{aligned} 
\end{equation} 
Note that we use the entropy density $s$, rather than the specific entropy $ \eta $, and that the density $ \epsilon _{\rm m}$ is related to our previous function $e$ via the first relation in \eqref{e_m_e_M}.


At some occasion, we shall also consider simpler, but often less practical, expression of $ \epsilon $ given directly in terms of the Faraday 2-form $F$ rather than its electric and magnetic component with respect to the fluid velocity $u$. Such an expression is written as
\begin{equation}\label{simpler_F} 
\mathfrak{e}(\rho  , \eta , u, F, \mathsf{g}),
\end{equation} 
where we note that it is now needed to include the world-velocity $u$ as one of the variables, as well as the Lorentzian metric $ \mathsf{g} $.

\paragraph{Spacetime reduced Euler-Lagrange equations.} The results stated in Theorem \ref{spacetime_reduced_EL} and \ref{convective_reduced_EL} are directly applicable to relativistic electromagnetic fluids. In particular, we get the following Eulerian variational principle and spacetime reduced Euler-Lagrange equations by direct application of \eqref{Eulerian_VP}. 

\begin{proposition}[Variational formulation for electromagnetic fluids]\label{prop_Eulerian_VP_fluid} The Eulerian variational formulation for the relativistic electromagnetic fluid takes the form\color{black} 
\begin{equation}\label{Eulerian_VP_fluid}
\begin{aligned} 
&\!\!\!\!\!\!\left. \frac{d}{d\varepsilon}\right|_{\varepsilon=0}\int_{ \mathcal{N}_ \varepsilon } \ell\big( A_ \varepsilon , {\rm d} A_ \varepsilon , w_ \varepsilon , \varrho  _ \varepsilon , \varsigma _ \varepsilon , \mathsf{g} \big)=0 \quad \text{for variations}\\
&\!\!\!\! \delta \mathcal{N} = \zeta |_{ \partial \mathcal{N} } \big/T \partial \mathcal{N} , \;\; \delta A = - \pounds _ \zeta A + 
\deltabar A,\, \delta w = - \pounds _ \zeta w, \,   \delta \varrho = - \pounds _ \zeta \varrho, \,   \delta \varsigma = - \pounds _ \zeta \varsigma,\phantom{\int_A^B}\hspace{-0.7cm}
\end{aligned} 
\end{equation}\color{black} 
where $ \zeta$ is an arbitrary vector field on $ \mathcal{N} $ such that $ \zeta |_{ \Phi ( \lambda _0, \mathcal{B} )}= \zeta |_{ \Phi ( \lambda _1, \mathcal{B} )}=0$ and $\deltabar A$ is an arbitrary one-form on $ \mathcal{N} $ such that $\deltabar A|_{ \Phi ( \lambda _0, \mathcal{B} )}= \deltabar A |_{ \Phi ( \lambda _1, \mathcal{B} )}=0$.

The critical conditions associated to \eqref{Eulerian_VP_fluid} are
\begin{equation}\label{spacetime_EL_fluid} 
\!\!\!\!\!\!\!\left\{
\begin{array}{l}
\displaystyle\vspace{0.2cm}\!\!\operatorname{div}^ \nabla \!\Big( \Big(  \ell - \varrho \frac{\partial \ell}{\partial \varrho }- \varsigma  \frac{\partial \ell}{\partial \varsigma  } \Big)  \delta  + w \otimes \frac{\partial \ell}{\partial w}-  \frac{\partial \ell}{\partial A } \otimes  A  -  \frac{\partial \ell}{\partial F }\stackrel{\rm tr}{ \otimes }F\Big)=0\\
\displaystyle\vspace{0.2cm}\!\! \frac{\partialnew \ell}{\partialnew A} + {\rm d}  \frac{\partialnew \ell}{\partialnew F}=0, \qquad  \qquad i^*_{ \partial _{\rm cont}\mathcal{N} } \frac{\partialnew \ell}{\partialnew F}=0\\
\displaystyle \!\!\Big( \!\Big(  \ell - \varrho \frac{\partial \ell}{\partial \varrho }- \varsigma  \frac{\partial \ell}{\partial \varsigma  } \Big)  \delta  + w \otimes \frac{\partial \ell}{\partial w}-  \frac{\partial \ell}{\partial A } \otimes  A  -  \frac{\partial \ell}{\partial F }\stackrel{\rm tr}{ \otimes }F\Big) ( \cdot ,n ^\flat )=0\;\;\text{on $\textcolor{black}{\partial_{\rm cont} \mathcal{N}}$}\hspace{-0.5cm}
\end{array}
\right.
\end{equation}
and the variables $w$, $ \varrho $, $ \varsigma $ satisfy
\begin{equation}\label{advections_fluid} 
\pounds _w \varrho =0 \quad\text{and}\quad \pounds _w \varsigma  =0.
\end{equation} 
\end{proposition}

\medskip 

We note that equations \eqref{advections_fluid} are equivalently written in terms of the world-velocity, the proper mass and entropy density as
\[
\pounds _u (\rho  \mu(g))  =0 \quad\text{and}\quad \pounds _u (s  \mu(g))   =0.
\]

As usual, the variational formulation only provides half of Maxwell's equations, specifically $\frac{\partialnew \ell}{\partialnew A} + {\rm d}  \frac{\partialnew \ell}{\partialnew F}=0$, corresponding to the relativistic Gauss and Amp\`ere laws. The other half, ${\rm d}F=0$, corresponding to the relativistic magnetic Gauss law and Faraday's law, is automatically satisfied by the definition $F= {\rm d}A$. Besides this, \eqref{spacetime_EL_fluid} also gives the energy and momentum equations for the continuum, the boundary conditions for the continuum, and the boundary condition for $\frac{\partial\ell}{\partial F}$, which determines the boundary conditions for the fields $D$ and $H$. Boundary conditions for $F$, i.e. for $E$ and $B$, will be addressed when considering the physically relevant case where the fluid is coupled to gravitational theory and the external electromagnetic field via junction conditions.

In the next three paragraphs, we study in details the form of the stress-energy-momentum tensor appearing in \eqref{spacetime_EL_fluid}. This is useful both to connect with the previous literature and to derive a concrete formulation for the system \eqref{spacetime_EL_fluid}.

\subsubsection{Stress-energy-momentum tensor for electromagnetic fluids}\label{SEM_fluid}

To facilitate comparison with the literature, we present three forms of the stress-energy-momentum tensor, depending on whether we express it in terms of the original ($\Phi$-transported) variables $(A, F, w, \varrho, \varsigma)$ appearing in the Lagrangian,  in terms of the proper variables explicitly involving the electric and magnetic fields $(\rho,s,E,B)$, or in terms of the Faraday 2-form $(\rho, s, u,F)$.

\paragraph{The stress-energy-momentum tensor I.} The first step is to compute the partial derivatives of the Lagrangian density obtained in \eqref{Lagrangian_epsilon}. We obtain
%
%
%
\begin{equation}\label{computation_derivatives}
\begin{aligned}
\frac{\partial \ell}{\partial w }&=\frac{1}{c} u ^\flat  \frac{1}{\sqrt{-\mathsf{g}(w,w)}} \Big(  \frac{\partial \epsilon }{\partial \rho  } \rho+   \frac{\partial \epsilon }{\partial s  } s-  \frac{\partial \epsilon }{\partial E  } \cdot E- \frac{\partial \epsilon }{\partial B } :B \Big)\mu ( \mathsf{g} )\\
& \qquad \qquad \qquad + \Big(\frac{\partial \epsilon}{\partial E} \cdot \mathbf{i} _{\_\,} F+ \frac{\partial \epsilon}{\partial B} \cdot \mathbf{i} _{\_\,} (\star F)  \Big)  \mu ( \mathsf{g} )- q A \varrho  \\
\frac{\partial \ell}{\partial \varrho }&=  - \frac{1}{c} \sqrt{-\mathsf{g}(w,w)} \frac{\partial \epsilon }{\partial \rho  }  - q A \cdot w, \qquad \frac{\partial \ell}{\partial \varsigma  }= - \frac{1}{c}\sqrt{-\mathsf{g}(w,w)}  \frac{\partial \epsilon }{\partial s} \\
\frac{\partial \ell}{\partial A  }&=- q \varrho w, \qquad \qquad  \frac{\partialnew \ell}{\partialnew A  }=- q  \mathbf{i} _w \varrho \\
\frac{\partial \ell}{\partial F  }& =  \frac{1}{c}  u \wedge  \frac{\partial \epsilon }{\partial E } \mu ( \mathsf{g} )+\frac{1}{c}\star\Big( u \wedge\frac{\partial \epsilon }{\partial B }\Big) \mu ( \mathsf{g} )\\
\frac{\partialnew \ell}{\partialnew F  }&= \frac{1}{c} \star \Big( u ^\flat \wedge   \Big(\frac{\partial \epsilon }{\partial E } \Big)^\flat \Big) \mu ( \mathsf{g} )- \frac{1}{c}  u ^\flat \wedge   \Big(\frac{\partial \epsilon }{\partial B } \Big)^\flat\mu ( \mathsf{g} ).
\end{aligned}
\end{equation} 
Note that from the expression $\frac{\partialnew \ell}{\partialnew F  }$ and the definitions in \eqref{dellldF1} and \eqref{D_H}, the electric displacement field $D$ and magnetic intensity field $H$ are given by
\begin{equation}\label{D_H_epsilon} 
\frac{\partial \epsilon }{\partial E} = - D ^\sharp , \qquad 
\frac{\partial \epsilon }{\partial B} = H ^\sharp  .
\end{equation} 
From \eqref{decomposition_e}, we have the usual relations $D=P+E$ and $H=-M+B$ with the polarization $P \in \Omega ^1 ( \mathcal{M} )$, the magnetization $M \in \Omega ^{n-2}( \mathcal{M} )$ and 
\begin{equation}\label{P_M_epsilon}
\frac{\partial\epsilon_{\rm m} }{\partial E} = - P ^\sharp ,\quad  \frac{\partial\epsilon_{\rm m} }{\partial B} = - M ^\sharp  , \quad \frac{\partial\epsilon_{\rm M} }{\partial E} = - E ^\sharp   ,\quad  \frac{\partial\epsilon_{\rm M} }{\partial B} = B^\sharp  .
\end{equation}
Some useful formulas to proceed further are:
\begin{align*} 
\frac{\partial \epsilon }{\partial E}  \cdot \mathbf{i} _{\_\,}F &= \frac{1}{c} \Big( \frac{\partial \epsilon }{\partial E} \cdot E \Big) u ^\flat+ \frac{1}{c} (-1)^n\mathbf{i} _{ \frac{\partial \epsilon }{\partial E}  } \mathbf{i} _u (\star B) \\
\frac{\partial \epsilon }{\partial B}    : \mathbf{i} _{\_\,}(\star F) &= \frac{1}{c} \Big( \frac{\partial \epsilon }{\partial E} : B \Big) u ^\flat  + \frac{1}{c} (-1)^n \mathbf{i} _{E^\sharp  } \mathbf{i} _u\Big(\!\star \frac{\partial \epsilon }{\partial B} \Big) 
\end{align*} 
as well as
\begin{align*} 
w \otimes \frac{\partial \ell}{\partial w} &=   \Big[ \Big(\frac{\partial \epsilon }{\partial \rho  } \rho  + \frac{\partial \epsilon }{\partial s}s \Big)  \frac{1}{c^2} u \otimes u ^\flat\\
& \qquad  +(-1)^n \frac{1}{c^2} u \otimes \Big( \mathbf{i} _{ \frac{\partial \epsilon }{\partial E}} \mathbf{i} _u (\star B) +\mathbf{i} _{E^\sharp} \mathbf{i} _u \Big(\star \frac{\partial \epsilon }{\partial B} \Big)^\flat \Big) \\
& \qquad  - \rho  u \otimes  qA\Big] \mu (\mathsf{g})\\
\frac{\partial \ell}{\partial F} \stackrel{\rm tr}{ \otimes }F&= \Big[ \frac{\partial \epsilon }{\partial E} \cdot E \frac{1}{c^2}  u \otimes u ^\flat  - \frac{\partial \epsilon }{\partial E}  \otimes E + B \stackrel{\rm tr}{ \otimes } \frac{\partial \epsilon }{\partial B}   -\Big(\frac{\partial \epsilon }{\partial B}:  B \Big)\mathsf{P}\\
& \qquad  +(-1)^n \frac{1}{c^2} u \otimes \mathbf{i} _{ \frac{\partial \epsilon  }{\partial E}  } \mathbf{i} _u (\star B) -(-1)^n \frac{1}{c^2}\mathbf{i} _{E   } \mathbf{i} _{u^\flat}\Big(\!\star \frac{\partial \epsilon }{\partial B} \Big) \otimes u ^\flat\Big]\mu ( \mathsf{g} ),
\end{align*}
with $\mathsf{P}$ the projection tensor defined earlier.

From these expressions, the concrete form of the stress energy-momentum tensor is obtained as
\begin{equation}\label{T_EMF}
\begin{aligned} 
\mathfrak{T}&=  \Big(  \ell - \varrho \frac{\partial \ell}{\partial \varrho }- \varsigma  \frac{\partial \ell}{\partial \varsigma  } \Big)  \delta  + w \otimes \frac{\partial \ell}{\partial w}-  \frac{\partial \ell}{\partial A } \otimes  A  -  \frac{\partial \ell}{\partial F }\stackrel{\rm tr}{ \otimes }F\\
&=\Big[\Big(\epsilon   - \frac{\partial \epsilon   }{\partial E} \cdot E \Big) \frac{1}{c^2} u \otimes u ^\flat  + \frac{1}{c} \Big( u \otimes S_ \epsilon   + S^\sharp_ \epsilon    \otimes u ^\flat \Big)\\
& \qquad  \qquad +  \frac{\partial \epsilon    }{\partial E  } \otimes E - B ^\sharp  \stackrel{\rm tr}{ \otimes }  \frac{\partial \epsilon    }{\partial B  }^\flat    + \Big(  \frac{\partial \epsilon   }{\partial \rho  } \rho  + \frac{\partial \epsilon   }{\partial s} s + \frac{\partial \epsilon    }{\partial B  } : B - \epsilon\Big) \mathsf{P}\Big]\mu ( \mathsf{g} ),
\end{aligned} 
\end{equation}  
with $S_ \epsilon  \in \Omega ^1 ( \mathcal{M} )$ the Poynting one-form given by
\[
S_ \epsilon = (-1)^n \frac{1}{c} \mathbf{i} _{E ^\sharp  } \mathbf{i} _u \star \Big(\frac{\partial \epsilon }{\partial B}\Big) ^\flat\in \Omega ^1 ( \mathcal{M} ).
\]

\begin{remark}[On the symmetry of $ \mathfrak{T}$]\rm
We know from \eqref{spacetime_mat_ell} with $\gamma=\mathsf{g}$ that the tensor \eqref{T_EMF}, with indices both up or both down, must be symmetric when spacetime covariance holds. This can be concretely seen as follows. First we note that it is enough to check that $D \otimes E + B \stackrel{\rm tr}{\otimes} H$ or, equivalently $P \otimes E - B \stackrel{\rm tr}{\otimes} M$, is symmetric, as the other terms of $ \mathfrak{T} $ clearly are. Here we are using the notations \eqref{D_H_epsilon} and \eqref{P_M_epsilon}. Using $e$, instead of $ \epsilon $, and taking the derivative of the covariance property \eqref{cov_e} with respect to the diffeomorphism $ \psi $ at the identity yields
\begin{equation}\label{true} 
\frac{\partial e}{\partial E_ \mu }  E_ \nu  + \frac{1}{(n-3)!} \frac{\partial e}{\partial B_{\mu \alpha  _1... \alpha _{n-3}}} B_{ \nu  \alpha_1... \alpha_{n-3}} + 2\frac{\partial e}{\partial \mathsf{p}_{ \alpha \mu }} \mathsf{p}_{ \alpha \nu }=0.
\end{equation}
Multiplying by $ \mathsf{p}_{ \gamma \mu }$ and using $ \mathsf{p}_{ \gamma \mu }\frac{\partial e}{\partial E_ \mu }=\frac{\partial e}{\partial E_ \gamma  }$ and $ \mathsf{p}_{ \gamma \mu }\frac{\partial e}{\partial B_{\mu \alpha  _1... \alpha _{n-3}}}=\frac{\partial e}{\partial B_{\gamma  \alpha  _1... \alpha _{n-3}}}$, which follow from \eqref{D_H}, \eqref{D_H_epsilon} and \eqref{P_M_epsilon}, yields  
\[
(P \otimes E + M \stackrel{\rm tr}{\otimes} B)_{ \mu \nu }= \rho  2 \mathsf{p}_{ \gamma \mu }\frac{\partial e}{\partial \mathsf{p}_{ \alpha \mu }} \mathsf{p}_{ \alpha \nu }
\]
so that $P \otimes E + M \stackrel{\rm tr}{\otimes} B$ is symmetric. This proves that $P \otimes E - B \stackrel{\rm tr}{\otimes} M$ is, up to a symmetric tensor, equal to $-M \stackrel{\rm tr}{\otimes} B - B \stackrel{\rm tr}{\otimes} M$, which is symmetric. In conclusion, $\mathfrak{T}$ in \eqref{T_EMF} is symmetric, as it should.
\end{remark}

\begin{remark}[On the dependence of $e$ on $\mathsf{p}$ and $ \mathsf{g} $]\label{no_p}\rm If the dependence on $\mathsf{p}$, instead of $\mathsf{g}$, would be considered, then in the expression of the partial derivative $ \frac{\partial \ell}{\partial w}$ in \eqref{computation_derivatives} there should be in principle a contribution of $ \frac{\partial e}{\partial \mathsf{p}}$ since $\mathsf{p}$ depends on the variable $w$. This contribution however vanishes. The derivative with respect to $w$ yields the expression
\begin{align*} 
- \rho  \frac{\partial e}{\partial \mathsf{p} _{ \mu \nu } }\frac{\partial \mathsf{p} _{ \mu \nu }}{\partial w^ \lambda } &=  - \rho  \frac{\partial e}{\partial \mathsf{p} _{ \mu \nu } }\Big(\frac{2w_ \mu w_ \nu }{\textcolor{black}{ \mathsf{g}}(w,w) ^2 }   \textcolor{black}{ \mathsf{g}}_{ \alpha \lambda  }   w ^ \alpha - \frac{1}{\textcolor{black}{ \mathsf{g}}_{ \alpha \beta } w ^ \alpha w^ \beta } (\textcolor{black}{ \mathsf{g}}_{ \mu \lambda } w_ \nu   +  \textcolor{black}{ \mathsf{g}}_{ \nu \lambda }w_ \mu  )\Big)\\
&=\frac{2}{\textcolor{black}{ \mathsf{g}}(w,w)} \rho  \frac{\partial e}{\partial \mathsf{p} _{ \mu \nu }} \mathsf{p} _{ \mu \lambda } w_ \nu.
\end{align*} 
Then, using \eqref{true} again,  this expression is zero, from $\frac{\partial e}{\partial E_ \mu }  E_ \nu w_ \mu =0$ and $\frac{\partial e}{\partial B_{\mu \alpha  _1... \alpha _{n-3}}} w_ \mu =0$.
This is consistent with the fact that we can use $ \mathsf{g} $ instead of $\mathsf{p}$ in the function $e$. We will see concretely an instance of this fact in one example later.
\end{remark}

\paragraph{The stress-energy-momentum tensor II.} We now give the expression of $ \mathfrak{T} $ in terms of the energy function expressed with the Faraday 2-form, as in \eqref{simpler_F}. We thus rewrite the Lagrangian density \eqref{Lagrangian_epsilon} in the following equivalent way
\begin{equation}\label{Lagrangian_e_F} 
\ell(A, {\rm d} A,w, \varrho , \varsigma , \mathsf{g} )= - \mathfrak{e} ( \rho  , s, u, F, \mathsf{g}) \mu ( \mathsf{g} )- q \varrho A\! \cdot \! w,
\end{equation} 
with $ \mathfrak{e}$ subsequently split as
\begin{equation}\label{decomposition_frake} 
\mathfrak{e}( \rho  , s, u, F, \mathsf{g} )= \mathfrak{e}_{\rm m}( \rho  , s, u, F, \mathsf{g} )+ \mathfrak{e}_{\rm M}( F, \mathsf{g} )
\end{equation}
as in \eqref{decomposition_e}. Note the dependence of $ \mathfrak{e}$ on the world-velocity $u$, which was not present in  $ \epsilon $. It is needed to reproduce the expression $ \epsilon ( \rho  , s, E, B, \mathsf{g} )$ with general, although still covariant, dependence on $E$ and $B$.
The expressions of the partial derivatives are found as
\begin{equation}\label{computation_derivatives_e}
\begin{aligned}
\frac{\partial \ell}{\partial \varrho }&=- \frac{\partial \mathfrak{e}  }{\partial \rho  } \frac{\sqrt{- \mathsf{g} (w,w)}}{c} - qA \cdot w, \qquad \frac{\partial \ell}{\partial \varsigma  }=- \frac{\partial \mathfrak{e}  }{\partial s  } \frac{\sqrt{- \mathsf{g} (w,w)}}{c } \\
\frac{\partial \ell}{\partial w} &= \Big(\rho  \frac{\partial \mathfrak{e}  }{\partial \rho  } +  s  \frac{\partial \mathfrak{e}  }{\partial s  }\Big) \frac{ u^\flat }{c} \frac{1}{\sqrt{- \mathsf{g} (w,w)}} \mu ( \mathsf{g} )- \frac{\partial \mathfrak{e}  }{\partial u}\frac{c}{\sqrt{- \mathsf{g} (w,w)}} \mathsf{P}\mu ( \mathsf{g} ) \\ 
\frac{\partial \ell}{\partial A }&=- \rho  q u \otimes \mu (\mathsf{g} ), \qquad \frac{\partialnew \ell}{\partialnew A }=- \rho  q \mathbf{i} _u  \mu (\mathsf{g} )\\
\frac{\partial \ell}{\partial F }&= - \frac{\partial \mathfrak{e}  }{\partial F} \mu ( \mathsf{g} ), \qquad \frac{\partialnew \ell}{\partial F }= - \star\left(\frac{\partial \mathfrak{e}  }{\partial F} \right)^\flat.
\end{aligned}
\end{equation} 
From these expressions, the stress energy-momentum tensor is obtained as
\begin{equation}\label{T_EMF_F}
\begin{aligned} 
\!\!\!\mathfrak{T}&=\Big(  \ell - \varrho \frac{\partial \ell}{\partial \varrho }- \varsigma  \frac{\partial \ell}{\partial \varsigma  } \Big)  \delta  + w \otimes \frac{\partial \ell}{\partial w}-  \frac{\partial \ell}{\partial A } \otimes  A  -  \frac{\partial \ell}{\partial F }\stackrel{\rm tr}{ \otimes }F\\
&=\Big[\Big(\mathfrak{e}   -  \frac{\partial \mathfrak{e}  }{\partial u}  \cdot u \Big)\frac{1}{c^2}u \otimes u ^\flat  - u \otimes \frac{\partial \mathfrak{e}  }{\partial u} + \frac{\partial \mathfrak{e}  }{\partial F}   \stackrel{\rm tr}{\otimes} F +   \Big( \rho  \frac{\partial \mathfrak{e}  }{\partial \rho  }+ s \frac{\partial \mathfrak{e}  }{\partial s} - \mathfrak{e}    \Big) \mathsf{P}\Big]\mu ( \mathsf{g} ).
\end{aligned}
\end{equation}  
While this expression is much simpler than the one obtained earlier, it is of less practical use since $- u \otimes \frac{\partial \mathfrak{e}  }{\partial u} + \frac{\partial \mathfrak{e}  }{\partial F}   \stackrel{\rm tr}{\otimes} F$ is not decomposed into its projections along and orthogonal to the fluid velocity. Implementing these projections would basically consist in rederiving \eqref{T_EMF}. Related to this, it is important to observe that the  factors in front of $\frac{1}{c^2}u \otimes u ^\flat $ and $\mathsf{P}$ in \eqref{T_EMF_F} do not correspond to those factors in \eqref{T_EMF}. The alternative formulation in \eqref{T_EMF_F} is useful to get insight on the final expression of the equations, as we will see in \S\ref{subsubsec_equ}.

The results obtained in the three preceding paragraphs are summarized in the following proposition.

\begin{proposition}[Stress-energy-momentum tensor for electromagnetic fluids]\label{end_of_controversy} We have the following three equivalent formulations of the stress-energy-momentum tensor of electromagnetic fluids.
\begin{itemize}
\item[\bf (i)] Formulation in terms of the $ \Phi $-transported variables. For a Lagrangian density $\ell(A, {\rm d} A, w, \varrho , \varsigma , \mathsf{g} )$, the stress-energy-momentum tensor takes the form
\begin{equation}\label{T_1} 
\mathfrak{T} = \Big(  \ell - \varrho \frac{\partial \ell}{\partial \varrho }- \varsigma  \frac{\partial \ell}{\partial \varsigma  } \Big)  \delta  + w \otimes \frac{\partial \ell}{\partial w}-  \frac{\partial \ell}{\partial A } \otimes  A  -  \frac{\partial \ell}{\partial F }\stackrel{\rm tr}{ \otimes }F.
\end{equation} 
\item[\bf (ii)] Formulation in terms of the electric and magnetic components. If the Lagrangian density is expressed as
\[
\ell(A, {\rm d} A,w, \varrho , \varsigma , \mathsf{g} )= - \epsilon ( \rho  , s, E, B, \mathsf{g}) \mu ( \mathsf{g} ) - q \rho   A\! \cdot \!u\, \mu ( \mathsf{g} ),
\]
the stress-energy-momentum tensor takes the form
\begin{equation}\label{T_2}
\begin{aligned} 
\mathfrak{T}&= \Big[\Big(\epsilon   - \frac{\partial \epsilon   }{\partial E} \cdot E \Big) \frac{1}{c^2} u \otimes u ^\flat  + \frac{1}{c} \Big( u \otimes S_ \epsilon   + S^\sharp_ \epsilon    \otimes u ^\flat \Big)\\
& \qquad   +  \frac{\partial \epsilon    }{\partial E  } \otimes E - B ^\sharp  \stackrel{\rm tr}{ \otimes }  \frac{\partial \epsilon    }{\partial B  }^\flat    + \Big(   \frac{\partial \epsilon   }{\partial \rho  } \rho  + \frac{\partial \epsilon   }{\partial s} s + \frac{\partial \epsilon    }{\partial B  } : B- \epsilon \Big) \mathsf{P}\Big]\mu ( \mathsf{g} ),
\end{aligned} 
\end{equation}
with $S_ \epsilon $ the Poynting one-form
\[
S_ \epsilon = (-1)^n \frac{1}{c}\mathbf{i} _{E^\sharp} \mathbf{i} _u \star \Big(\frac{\partial \epsilon }{\partial B} \Big)^\flat .
\]
\item[\bf (iii)] Formulation in terms of the Faraday 2-form. If the Lagrangian density is expressed as
\[
\ell(A, {\rm d} A,w, \varrho , \varsigma , \mathsf{g} )= - \mathfrak{e} ( \rho  , s, u,F, \mathsf{g}) \mu ( \mathsf{g} )- q \rho   A\! \cdot \!u \,\mu ( \mathsf{g} ) 
\]
the stress-energy-momentum tensor takes the form
\begin{equation}\label{T_3}
\mathfrak{T}=\Big[\Big(\mathfrak{e}   -  \frac{\partial \mathfrak{e}  }{\partial u}  \cdot u \Big)\frac{1}{c^2}u \otimes u ^\flat  - u \otimes \frac{\partial \mathfrak{e}  }{\partial u} + \frac{\partial \mathfrak{e}  }{\partial F}   \stackrel{\rm tr}{\otimes} F +   \Big( \rho  \frac{\partial \mathfrak{e}  }{\partial \rho  }+ s \frac{\partial \mathfrak{e}  }{\partial s} - \mathfrak{e}    \Big) \mathsf{P}\Big]\mu ( \mathsf{g} ).
\end{equation}
\end{itemize}
\end{proposition}

\begin{remark}[Role of current term]\rm
We observe that the current term in \eqref{Lagrangian_epsilon} does not contribute to the expression of the stress-energy-momentum tensor $ \mathfrak{T} $. However, it determines the source term of the Maxwell equations. As we shall see in the examples later, when the latter equations are used when writing the energy and momentum fluid equations $\operatorname{div}\mathfrak{T}=0$, this source term generates the Lorentz force, despite not appearing explicitly in the expression of $\mathfrak{T}$.
\end{remark}



\begin{remark}[Matter+Maxwell decomposition]\rm The writing \eqref{T_2} allows a very clear separation into the charged matter and pure electromagnetic contributions to the stress-energy-momentum tensor. If $ \epsilon  $ decomposes as in \eqref{decomposition_e}, we can write $ \mathfrak{T} = \mathfrak{T} _{\rm m}+ \mathfrak{T} _{\rm M}$, with both terms symmetric and with $ \mathfrak{T} _{\rm M}$ the Maxwell stress-energy-momentum tensor, see \eqref{Maxwell_T}. For further use and discussion later, we write them explicitly as follows
\begin{equation}\label{our_decomposition}
\begin{aligned} 
\mathfrak{T}_{\rm m}&= \Big[\Big(\epsilon _{\rm m}  - \frac{\partial \epsilon  _{\rm m} }{\partial E} \cdot E \Big) \frac{1}{c^2} u \otimes u ^\flat  + \frac{1}{c} \Big( u \otimes S_ {\epsilon _{\rm m}}  + S^\sharp_ {\epsilon _{\rm m}}  \otimes u ^\flat \Big)\\
& \quad   +  \frac{\partial \epsilon_{\rm m}    }{\partial E  } \otimes E - B ^\sharp  \stackrel{\rm tr}{ \otimes }  \frac{\partial \epsilon _{\rm m}   }{\partial B  }  ^\flat   + \Big(  \frac{\partial \epsilon   _{\rm m}  }{\partial \rho  } \rho  + \frac{\partial \epsilon  _{\rm m}   }{\partial s} s + \frac{\partial \epsilon   _{\rm m}   }{\partial B  } : B -\epsilon  _{\rm m} \Big) \mathsf{P}\Big]\mu ( \mathsf{g} )\\
\mathfrak{T} _{\rm M}&=\Big[\Big( \epsilon _{\rm M}  - \frac{\partial \epsilon  _{\rm M} }{\partial E} \cdot E  \Big) \frac{1}{c^2} u \otimes u ^\flat  + \frac{1}{c} \Big( u \otimes S_ {\epsilon _{\rm M}}  + S^\sharp_{\epsilon _{\rm M}}  \otimes u ^\flat \Big)\\
& \quad    +  \frac{\partial \epsilon  _{\rm M}  }{\partial E  } \otimes E - B ^\sharp  \stackrel{\rm tr}{ \otimes }  \frac{\partial \epsilon_{\rm M}    }{\partial B  }  ^\flat  +\Big( \frac{\partial \epsilon  _{\rm M} }{\partial B} \cdot B  - \epsilon _{\rm M} \Big) \mathsf{P}\Big]\mu ( \mathsf{g} ).
\end{aligned} 
\end{equation}
Here one can use \eqref{P_M_epsilon} to replace the partial derivatives by their corresponding expression. Due to the specific form of $ \epsilon _{\rm M}$, in the expression of $\mathfrak{T} _{\rm M}$ we have
\[
\epsilon _{\rm M}  - \frac{\partial \epsilon  _{\rm M} }{\partial E} \cdot E= - \epsilon _{\rm M}  + \frac{\partial \epsilon  _{\rm M} }{\partial B} \cdot B=\frac{1}{2} |E| ^2 + \frac{1}{2} |B| ^2.
\]

\end{remark}

\begin{remark}[Comparison with  \cite{ErMa1990}]\label{ErMa}\rm 
To ease the comparison with this earlier major work, we first note that the two free energy functions that they consider, denoted $ \psi $ and $\widetilde \psi $  (see \S15 in \cite{ErMa1990}), are related to our functions $ \epsilon _{\rm m}$ and $e$ via the relations
\begin{align*}
\epsilon _{\rm m} &=\rho ( c ^2 + e) \mu ( \mathsf{g} )= \rho \big(c^2+  \widetilde\psi + \eta T\big) \\
\epsilon _{\rm m}- \frac{\partial \epsilon _{\rm m}}{\partial E} \cdot E &= \rho  \big(c^2 + e - \frac{\partial e}{\partial E} \cdot E\big)\mu ( \mathsf{g} )= \rho( c^2 +  \psi + \eta T) .
\end{align*}
Also, the decomposition of the stress-energy-momentum tensor they consider, denoted $ \mathfrak{T} ={}_{\rm m} {\rm T} +{}_{\rm M}{\rm T}$, is not the same as our $ \mathfrak{T} = \mathfrak{T} _{\rm m}+ \mathfrak{T} _{\rm M}$ which corresponds to the decomposition $ \epsilon = \epsilon _{\rm m}+ \epsilon _{\rm M}$. In our notations, their splitting reads
\begin{equation}\label{T_M}
\begin{aligned} 
{}_{\rm m}{\rm T}&= \Big[\Big(\epsilon _{\rm m}  - \frac{\partial \epsilon  _{\rm m} }{\partial E} \cdot E \Big) \frac{1}{c^2} u \otimes u ^\flat  + \Big(  \frac{\partial \epsilon   _{\rm m} }{\partial  \rho    } \rho  + \frac{\partial \epsilon   _{\rm m} }{\partial  s    } s - \epsilon  _{\rm m}  \Big) \mathsf{P}\Big]\mu ( \mathsf{g} )\\
{}_{\rm M}{\rm T}&= \Big[\Big(\epsilon _{\rm M}  - \frac{\partial \epsilon  _{\rm M} }{\partial E} \cdot E \Big) \frac{1}{c^2} u \otimes u ^\flat  + \frac{1}{c} \Big( u \otimes S_ \epsilon   + S^\sharp_ \epsilon    \otimes u ^\flat \Big)\\
& \qquad  \qquad +  \frac{\partial \epsilon    }{\partial E  } \otimes E - B ^\sharp  \stackrel{\rm tr}{ \otimes }  \frac{\partial \epsilon    }{\partial B  }  ^\flat  + \Big(  \frac{\partial \epsilon   _{\rm M} }{\partial B  } : B +   \frac{\partial \epsilon   _{\rm m} }{\partial B  } : B - \epsilon  _{\rm M}  \Big) \mathsf{P}\Big]\mu ( \mathsf{g} ),
\end{aligned} 
\end{equation}
to be compared with \eqref{our_decomposition}. Note in particular the occurrence of the term $-   \frac{\partial \epsilon   _{\rm m} }{\partial B  } : B$ in ${}_{\rm M}{\rm T}$, which appears in $\mathfrak{T}_{\rm m}$ in our case. This form of the splitting originates from the way ${}_{\rm M}{\rm T}$ was  derived, namely from a microscopic model of interactions, see \cite{dGSu1972}. It is quite remarkable that although the two decompositions of the stress-energy-momentum tensors are different and originate from totally different approaches, the final expressions coincide with the one obtained from our variational approach.
\end{remark}

\begin{remark}[Interpretation of $\epsilon$]\label{energy_interpretation}\rm The first term in the expression \eqref{T_2} clarifies the physical meaning of the ``energy" density $\epsilon$ appearing in the Lagrangian density  $\ell$, see \eqref{Lagrangian_epsilon}, namely, $\epsilon$ is such that the Legendre transform with respect to $E$, i.e., $\epsilon_{\rm tot}=\epsilon   - \frac{\partial \epsilon   }{\partial E} \cdot E$, represents the total energy density of the continuum in the rest frame. This clarifies also the physical meaning of $\epsilon_m$, $\epsilon_M$, $e$, and $\mathscr{E}$, via the relations
\eqref{decomposition_e}, \eqref{e_m_e_M}, and \eqref{def_e}, see also \eqref{our_decomposition}.
\end{remark}

We now give the expression of the stress-energy-momentum tensor in some important examples.

\paragraph{Example 1: Relativistic Euler-Maxwell.} This is the special case in which we have simply
\begin{equation}\label{Euler_Maxwell} 
\epsilon _{\rm m}( \rho  , s, E,B, \mathsf{g} )= \mathfrak{e}_{\rm m}( \rho , \eta , u, F, \mathsf{g} )= \epsilon _0( \rho  , s),
\end{equation} 
with $\epsilon _0( \rho  , s)= \rho  (c^2+  e( \rho  , \eta ))$ given by a fluid state equation. The associated Lagrangian density is
\[
\ell( A, {\rm d} A, w, \varrho , \varsigma , \mathsf{g} ) = - \epsilon _0( \rho  , s)\mu ( \mathsf{g} )- q \rho  A\! \cdot  \! u \mu ( \mathsf{g} ) - \mathfrak{e}_{\rm M}(F, \mathsf{g} )\mu ( \mathsf{g} ),
\]
see \eqref{decomposition_frake}. In this case, the expression of the stress-energy-momentum tensor in \eqref{T_3} reduces to the sum of that of an ordinary fluid and that of Maxwell theory. Indeed, from \eqref{T_3} and \eqref{general_Maxwell_SEM} one gets
\begin{align*} 
\mathfrak{T} &=\big[ (\epsilon _0+ \mathfrak{e}_{\rm M}) \frac{1}{c^2}u \otimes u ^\flat  + \frac{\partial \mathfrak{e}_{\rm M}}{\partial F} \stackrel{\rm tr}{ \otimes }F +\Big(\rho  \frac{\partial \epsilon _0}{\partial \rho  }+s  \frac{\partial \epsilon _0}{\partial s  } - \epsilon _0 - \mathfrak{e}_{\rm M}  \Big)\mathsf{P}\Big]\mu ( \mathsf{g} )\\
&= \Big[\epsilon _0 \frac{1}{c^2}u \otimes u ^\flat + p_0\mathsf{P} + \Big( \frac{\partial \mathfrak{e}_{\rm M}}{\partial F} \stackrel{\rm tr}{ \otimes }F - \mathfrak{e}_{\rm M}  \delta\Big)  \Big]\mu ( \mathsf{g} ) = \mathfrak{T} _{\rm fluid} + \mathfrak{T} _{\rm M}.
\end{align*} 
In this very special case $ \mathfrak{T} _m=\mathfrak{T} _{\rm fluid}$ thus recovers the standard expression of the stress-energy-momentum tensor of a fluid, which is not the case even in simple situations like the following example. Also, in this case the decompositions $ \mathfrak{T} = \mathfrak{T} _{\rm m}+ \mathfrak{T} _{\rm M}$ and $ \mathfrak{T} ={}_{\rm m} {\rm T} +{}_{\rm M}{\rm T}$ (see Remark \ref{ErMa}) coincide.

\paragraph{Example 2: Linear electromagnetic constitutive equations.} Consider the special case of $ \epsilon _{\rm m}$ given by
\begin{equation}\label{linear_epsilon} 
\epsilon _{\rm m} ( \rho  , s , E, B, \mathsf{g} )= \epsilon _0( \rho  , s) - \frac{1}{2} \chi ^E|E| ^2-  \frac{1}{2} \chi ^B|B| ^2,
\end{equation} 
for which \eqref{D_H_epsilon} and \eqref{P_M_epsilon} yield the relations $P= \chi ^EE$, $M= \chi ^BB$, $D= \varepsilon _0 E$, and $H= \mu_0 ^{-1} B$, with $ \varepsilon  _0 = 1+ \chi^ E$, $ \mu _0 ^{-1} = 1- \chi ^B$. Spacetime covariance clearly holds.

Assuming $ \chi ^E$ and $ \chi ^B$ constant, the first expression in \eqref{our_decomposition} becomes 
\begin{equation}\label{T_cont_1}
\begin{aligned} 
\mathfrak{T}_{\rm m}&= \Big[\Big(\epsilon _0  + \frac{1}{2} \chi ^E|E| ^2  -  \frac{1}{2} \chi ^B|B| ^2  \Big) \frac{1}{c^2} u \otimes u ^\flat  + \frac{1}{c} \Big( u \otimes S_ {\epsilon _{\rm m}}  + S^\sharp_ {\epsilon _{\rm m}}  \otimes u ^\flat \Big)\\
& \qquad- \chi ^EE\otimes E + \chi ^B B\stackrel{\rm tr}{ \otimes } B  + \Big( p_0 + \frac{1}{2} \chi ^E|E| ^2 - \frac{1}{2} \chi ^B|B| ^2  \Big) \mathsf{P}\Big] \mu ( \mathsf{g} ),
\end{aligned} 
\end{equation}
while the second one, $ \mathfrak{T} _{\rm M}$, is given in \eqref{Maxwell_T}. In \eqref{T_cont_1}, we have 
\[
S_{ \epsilon _{\rm m}}= -(-1)^n \frac{1}{c}\chi ^B \mathbf{i} _{E ^\sharp  } \mathbf{i} _ u (\star B)\quad\text{and}\quad p_0 = \frac{\partial \epsilon _0}{\partial \rho  } \rho  + \frac{\partial \epsilon  _0  }{\partial s} s  - \epsilon  _0,
\]
with $p_0$ the fluid pressure. The total density $ \epsilon $ in \eqref{decomposition_e} becomes
\[
\epsilon  ( \rho  , s , E, B, \mathsf{g} )= \epsilon _0( \rho  ,s) - \frac{1}{2} \varepsilon _0|E| ^2 +  \frac{1}{2} \mu_0 ^{-1} |B| ^2,
\]
so that the total stress-energy-momentum tensor for this model is found from \eqref{T_2} as 
\begin{equation}\label{T_example2}
\begin{aligned} 
\mathfrak{T}&=  \Big[\Big(\epsilon _0  + \frac{1}{2} \varepsilon _0 |E| ^2  +  \frac{1}{2} \mu _0 ^{-1} |B| ^2 \Big) \frac{1}{c^2} u \otimes u ^\flat  + \frac{1}{c} \Big( u \otimes S_ { \epsilon }  + S^\sharp_{ \epsilon }  \otimes u ^\flat \Big)\\
& \qquad \quad - \varepsilon _0 E\otimes E - \mu _0 ^{-1} B\stackrel{\rm tr}{ \otimes } B  +\Big( p_0 + \frac{1}{2} \varepsilon _0|E| ^2 + \frac{1}{2} \mu_0 ^{-1} |B| ^2\Big) \mathsf{P}\Big] \mu ( \mathsf{g} )
\end{aligned} 
\end{equation}
with Poynting one-form
\[
S_{ \epsilon}= (-1)^n \mu _0 ^{-1} c ^{-1}  \mathbf{i} _{E ^\sharp  } \mathbf{i} _ u (\star B).
\]

In the case $ \chi ^E$ and $ \chi ^B$ are functions of $ \rho  $ and $s$, only the last term in \eqref{T_cont_1} is changed by replacing $ \chi ^E  \rightarrow \tilde \chi ^E$ and  $ \chi ^B  \rightarrow \tilde \chi ^B$ with 
\[
\tilde \chi ^E= \chi ^E - \rho  \partial _ \rho  \chi ^E- s  \partial _ s  \chi ^E \quad\text{and}\quad  \tilde \chi ^B= \chi ^B + \rho  \partial _ \rho  \chi ^E+ s  \partial _ s  \chi ^E,
\]
with a similar change $\varepsilon _0  \rightarrow \tilde \varepsilon $ and  $\mu ^{-1}   \rightarrow \tilde \mu ^{-1} $ in the last term of \eqref{T_example2}.

Regarding our earlier comment after \eqref{cov_e} that one can use either $\mathsf{g} $ or $ \mathsf{p}$ in $e$ to construct invariants from $E$ and $B$ (see also Remark \ref{no_p}), we indeed concretely see that in \eqref{linear_epsilon}, one can indifferently use $ \mathsf{g} $ or $ \mathsf{p}$, as we can write
\begin{equation}\label{from_g_to_p} 
|E| ^2 =  \mathsf{g} (E,E)=\mathsf{g} ^{ \alpha \beta }E_{ \alpha } E_{ \beta }= ({} ^{-1} \mathsf{p})^{ \alpha \beta }E_ \alpha E_ \beta,
\end{equation} 
similarly for $B$.
This is possible because $\mathbf{i} _uE=0$, i.e. $E \in (\operatorname{span}u)^\circ$ and $\mathsf{p}: (\operatorname{span}u)^\perp \times (\operatorname{span}u)^\perp \rightarrow \mathbb{R} $ is nondegenerate, so that we have the isomorphism $ \mathsf{p}_{ \alpha \beta }: (\operatorname{span}u)^\perp \rightarrow (\operatorname{span}u)^\circ$, to which we can associate its inverse, denoted ${}^{-1}\mathsf{p}: (\operatorname{span}u)^\circ \rightarrow (\operatorname{span}u)^\perp $ to be used in \eqref{from_g_to_p}.

\paragraph{Example 3: Nonlinear electromagnetic constitutive equations.} Consider the case where the energy density $ \epsilon _{\rm m}$ is given by
\begin{equation}\label{nonlinear_epsilon} 
\epsilon _{\rm m} ( \rho  , s , E, B, \mathsf{g} )= f( \rho  , s , I_1,I_2, I_3),
\end{equation} 
with the invariants $I_1=|E|^2= \mathsf{g} (E,E)$, $I_2= |B|^2=\mathsf{g} (B, B)$, and $I_3= \mathsf{g} (E,B)$. This expression extends equation \eqref{linear_epsilon} to the nonlinear case. For the general isotropic form of  $\epsilon _{\rm m} $ in terms of invariants, we refer to \cite{GrEr1966b}, and also to \S\ref{sec_GREE} for further discussion in the context of electromagnetic continua. As in the previous cases, spacetime covariance is satisfied. The stress-energy-momentum tensor in this formulation can be determined using the following relations:
\[
P= - 2\frac{\partial f}{\partial I_1} E - \frac{\partial f}{\partial I_3}B \quad\text{and}\quad M= - 2\frac{\partial f}{\partial I_2} B - \frac{\partial f}{\partial I_3}E,
\]
along with $D=P+E$, and $H=-M +B$ to substitute into equation \eqref{T_2}. The details are left to the reader.

\subsubsection{Relativistic electromagnetic fluid equations}\label{subsubsec_equ}

We present two equivalent forms of the relativistic electromagnetic fluid equations, each offering valuable insights into the distinctions from the non-charged fluid case.

\paragraph{Relativistic electromagnetic fluid equations I.} We deduce the system of equations and boundary conditions \eqref{spacetime_EL_fluid} for the Lagrangian densities \eqref{Lagrangian_epsilon} of electromagnetic fluids. 
For the first equation in \eqref{spacetime_EL_fluid}, the standard approach to derive the associated momentum (Euler-type) and energy conservation equations involves projecting the equation 
$ \operatorname{div} \mathfrak{T} =0$, where 
$ \mathfrak{T} $ is given in \eqref{T_2}, along and orthogonal to the fluid velocity.
Since it is already expressed in the usual decomposed form, the following expressions are straightforward to obtain:
\begin{equation}\label{equations_EMF} 
\!\!\!\begin{aligned} 
& \operatorname{div}( \epsilon_{\rm tot} u+ S) + \frac{1}{c^2} g( S, \nabla _uu) = \mathfrak{t}_{\rm em}: \nabla u -p \operatorname{div}u \\
&\frac{1}{c^2} (\epsilon_{\rm tot}+ p) \nabla _u u + \frac{1}{c^2} \mathsf{P} \operatorname{div}( S \otimes u + u \otimes S)+\frac{1}{c^2} u\mathfrak{t}_{\rm em}: \nabla u = \operatorname{div}\mathfrak{t}_{\rm em}  - \mathsf{P} \nabla p,
\end{aligned} 
\end{equation} 
with the total energy density $\epsilon_{\rm tot}$, the electromagnetic stress tensor $\mathfrak{t}_{\rm em}$ and a pressure function $p$ defined as:
\begin{align*}
\epsilon_{\rm tot}= \epsilon   - \frac{\partial \epsilon   }{\partial E} \cdot E, \quad \mathfrak{t}_{\rm em}&= -\frac{\partial \epsilon    }{\partial E  } \otimes E + B ^\sharp  \stackrel{\rm tr}{ \otimes }  \frac{\partial \epsilon    }{\partial B  }^\flat ,\quad  p = \frac{\partial \epsilon    }{\partial \rho  } \rho  + \frac{\partial \epsilon  }{\partial s} s + \frac{\partial \epsilon  }{\partial B  } : B -  \epsilon,\\
&=D ^\sharp  \otimes E + B ^\sharp  \otimes H.
\end{align*}
These equations are complemented by the continuity equations for mass and entropy:
\begin{equation}\label{cont_eq_EMF} 
\pounds _u (\rho  \mu(\mathsf{g}))  =0   \quad\text{and}\quad \pounds _u(s\mu(\mathsf{g}))=0.
\end{equation} 
The second equation in \eqref{spacetime_EL_fluid}, representing Maxwell's equations in matter, yields:
\begin{equation}\label{Maxwell_E_B} 
- \delta  \Big(\frac{\partial \mathfrak{e}  }{\partial F}\Big)^\flat =  \rho  q u  ^\flat
\end{equation} 
by \eqref{computation_derivatives_e}, which can subsequently be decomposed using the formulas in \S\ref{subsec_EMfields}. Here $\delta = (-1)^{(n+1)(k-1)}\star{\rm d}\star $ is the codifferential acting on $k$-forms, $k=2$.
Finally, it remains to derive the various boundary conditions by using the last two conditions in \eqref{spacetime_EL_fluid}. This is done in the next proposition.

\begin{proposition}\label{Prop_EMF} The system of equations for relativistic electromagnetic fluids obtained in \eqref{spacetime_EL_fluid} by the variational principle \eqref{Eulerian_VP_fluid} yield the equations \eqref{equations_EMF}, \eqref{cont_eq_EMF} and \eqref{Maxwell_E_B}, together with the boundary conditions
\begin{equation}\label{BC_EMF_1} 
\mathsf{g} (u , n) =0, \quad \mathfrak{t}_{\rm em}( \cdot , n ^\flat )= pn ^\flat, \quad \mathbf{i} _n D=0, \quad \mathbf{i} _n \mathbf{i} _u (\star H)=0 \quad\text{on}\quad \partial _{\rm cont} \mathcal{N}.
\end{equation} 
\end{proposition}  
\noindent\textbf{Proof.} First, we note that $ \mathsf{g}  (u,n)=0$ follows from the definition of the word-velocity in terms of the word-tube. Then, we note that $ i^*_{ \partial _{\rm cont}\mathcal{N} } \frac{\partialnew \ell}{\partialnew F}=0$ is equivalent to $ \mathbf{i} _{n^\flat } \frac{\partial \ell}{\partial F} =0$ which, by \eqref{dellldF2}, gives $ u ^\flat \mathbf{i} _nD+(-1)^n \mathbf{i} _n \mathbf{i} _u(\star H)=0$. Applying $ \mathbf{i} _u$, we get $\mathbf{i} _nD=0$, from which $\mathbf{i} _n \mathbf{i} _u(\star H)=0$ follows. Finally, using \eqref{T_2}, the condition $ \mathfrak{T} ( \cdot , n ^\flat )=0$ gives $ \frac{1}{c^2} \mathbf{i} _n S u ^\flat - \mathfrak{t}_{\rm em}( \cdot , n ^\flat ) + p n ^\flat =0$. We note however that $\mathbf{i} _n S= (-1)^n\mathbf{i} _n \mathbf{i} _{ E ^\sharp  } \mathbf{i} _u (\star H)$ is zero by the condition $\mathbf{i} _n \mathbf{i} _u(\star H)=0$ already derived, hence we have obtained all the boundary conditions in \eqref{BC_EMF_1}. $\qquad\blacksquare$

\medskip

It should be observed that the variational principle does not provide any boundary conditions for $F$, i.e., for $E$ and $B$. This deficiency is automatically removed when considering the physically relevant case in which the fluid is coupled to the gravitation theory and to the outer electromagnetic field via the junction conditions as will be considered in \S\ref{subsub_coupling_fluid} below.

\paragraph{Relativistic electromagnetic fluid equations II.} 
We now use the expression of $\mathfrak{T} $ given in \eqref{T_3}, rather than \eqref{T_2}, to rewrite the energy and momentum conservation equations \eqref{equations_EMF} in a way that makes the ponderomotive four-force explicitly appear. We decompose the stress-energy-momentum tensor given in \eqref{T_3} as $ \mathfrak{T} = \mathfrak{T} _{\rm f}+ \mathfrak{T} _{\rm F}$ with
\begin{align*}
\mathfrak{T} _{\rm f}&=\Big[\mathfrak{e} _{\rm m} \frac{1}{c^2}u \otimes u ^\flat+   \Big(   \frac{\partial \mathfrak{e} _{\rm m} }{\partial \rho  } \rho+  \frac{\partial \mathfrak{e} _{\rm m} }{\partial s} s- \mathfrak{e}   _{\rm m} \Big) \mathsf{P}\Big]\mu(\mathsf{g})\\
\mathfrak{T} _{\rm F}&=  \Big[ - u \otimes \frac{\partial \mathfrak{e} _{\rm m} }{\partial u} \cdot \mathsf{P} + \frac{\partial \mathfrak{e}  }{\partial F}   \stackrel{\rm tr}{\otimes} F -  \mathfrak{e}_{\rm M}  \delta\Big]\mu(\mathsf{g}).
\end{align*}
There are two key advantages to this formulation. First, $ \mathfrak{T} _{\rm f}$ is formally identical to the stress-energy momentum tensor of a fluid. Second, it facilitates the use of the Maxwell equations in matter when computing the divergence of $ \mathfrak{T} _{\rm F}$, which would be more challenging with the form of $ \mathfrak{T} $ given in \eqref{T_2}. While $ \mathfrak{T} _{\rm F}$ is symmetric, its first two terms may not be symmetric individually. Additionally, note that $ \mathfrak{T} _{\rm f}$ only slightly differs from ${}_{\rm m}{\rm T}$ in \eqref{T_M} by the presence of the additional term $ \frac{\partial \epsilon _{\rm m}}{\partial E} \cdot E$.

We shall thus consider the equations $ \operatorname{div} \mathfrak{T} =0$ written as $\operatorname{div} ^\nabla\mathfrak{T} _{\rm f}= \mathfrak{f}$, with ponderomotive force $ \mathfrak{f}= -\operatorname{div} ^\nabla\mathfrak{T} _{\rm F}$. For this computation, it is useful to write the Maxwell equation in matter, i.e., the second equation in \eqref{spacetime_EL_fluid}, in terms of $\bar{\mathfrak{e}}=\mathfrak{e}\mu(\mathsf{g})$ as:
\begin{equation}\label{Maxwell_m_e} 
{\rm d} \frac{\partialnew \bar{\mathfrak{e}}}{\partialnew F} + q \rho  \mathbf{i} _u \mu ( \mathsf{g} )=0.
\end{equation} 
Using \eqref{crucial_formula} in Lemma \ref{crucial_lemma} as well as \eqref{Maxwell_m_e} we can compute the divergence of $ \mathfrak{T} _{\rm F}$ by noting the equalities
\begin{align*} 
\operatorname{div} ^\nabla\Big(\frac{\partial \bar{\mathfrak{e}} }{\partial F}   \stackrel{\rm tr}{\otimes} F - \bar{\mathfrak{e}}_{\rm M}  \delta\Big) &=  \frac{\partial \bar{\mathfrak{e}}}{\partial F} : \nabla F - \mathbf{i} _{\_\,} F \wedge {\rm d} \frac{\partialnew \bar{\mathfrak{e}}}{\partialnew F} - {\rm d} \bar{\mathfrak{e}}_{\rm M}\\
&=\frac{\partial \bar{\mathfrak{e}}}{\partial F} : \nabla F + \mathbf{i} _{\_\,} F \wedge q \rho  \mathbf{i} _u \mu ( \mathsf{g} ) - \frac{\partial \bar{\mathfrak{e}}_{\rm M}}{\partial F} : \nabla F\\
&= \frac{\partial \bar{\mathfrak{e}}_{\rm m}}{\partial F} : \nabla F - q \rho  \mathbf{i} _u F \mu( \mathsf{g} ).
\end{align*} 
Therefore, one gets the equations in the form
\begin{equation}\label{div_f_form} 
\operatorname{div} \mathfrak{T} _{\rm f}= \mathfrak{f}, \qquad \mathfrak{f}=  \Big[\operatorname{div}( u \otimes \omega _{\rm m}) -\frac{\partial \mathfrak{e}_{\rm m}}{\partial F} : \nabla F + q \rho  \mathbf{i} _u F \Big]\mu ( \mathsf{g} )
\end{equation} 
with $ \omega _{\rm m}$ the one-form defined by
\[
\omega_{\rm m} = \frac{\partial \mathfrak{e} _{\rm m} }{\partial u} \cdot \mathsf{P} = - (-1)^n\frac{1}{c^2}  \Big( \mathbf{i} _{ \frac{\partial \epsilon _{\rm m}}{\partial E}} \mathbf{i} _u (\star B)  + \mathbf{i} _{ E ^\sharp  } \mathbf{i} _u \star \Big( \frac{\partial \epsilon _{\rm m}}{\partial B}\Big) ^\flat \Big).
\]
The second equality follows from writing $ \mathfrak{e}_{\rm m}( \rho  , s, u, F, \mathsf{g} )= \epsilon _{\rm m}(\rho  , s, E, B, \mathsf{g} )$ and computing $\frac{\partial \mathfrak{e} _{\rm m} }{\partial u}$. The writing \eqref{div_f_form} has the advantage to make clearly appear the  the relativistic Lorentz force term $\rho  q\, \mathbf{i} _u F  \mu( \mathsf{g} )$ as well as the force terms due the dependence of the matter energy on the electromagnetic field, all summing up to the ponderomotive force. Finally \eqref{div_f_form} yields the energy and momentum conservation equations \eqref{equations_EMF} in the alternative form
\begin{equation}\label{equations_EMF_2}
\begin{aligned}  
&- \operatorname{div}( \mathfrak{e}_{\rm m} u) - p_{\rm m} \operatorname{div}u  = \mathfrak{f} \cdot u = - \omega ( \nabla _u u) - \frac{\partial \mathfrak{e}_{\rm m}}{\partial F} : \nabla _uF\\
&\frac{1}{c^2} ( \mathfrak{e}_{\rm m} + p_{\rm m} )\nabla _uu + \mathsf{P} \nabla p_{\rm m} = \mathsf{P} \mathfrak{f} =\mathsf{P}\big( \rho  \nabla _u ( \omega / \rho  ) -\frac{\partial \mathfrak{e}_{\rm m}}{\partial F} : \nabla F\big)+ q \rho  \mathbf{i} _u F,
\end{aligned} 
\end{equation}
with $p_{\rm m}=  \frac{\partial \mathfrak{e} _{\rm m} }{\partial \rho  } \rho+  \frac{\partial \mathfrak{e} _{\rm m} }{\partial s} s- \mathfrak{e}   _{\rm m}$. When comparing with \eqref{equations_EMF} one should keep in mind the relations
\begin{align*}
 \epsilon _{\rm tot}&= \mathfrak{e}_{\rm m}+ \mathfrak{e}_{\rm M} + D ^\sharp  \cdot E= \mathfrak{e}_{\rm m}+P ^\sharp \cdot E + \frac{1}{2}|E|^2+\frac{1}{2}|B|^2\\
 p&= p_{\rm m}- \mathfrak{e}_{\rm M} + H^\sharp  :B=p_m - M^\sharp : B+ \frac{1}{2}|E|^2+\frac{1}{2}|B|^2.
\end{align*}

In the Euler-Maxwell case, using \eqref{Euler_Maxwell}, the system \eqref{equations_EMF_2} reduces to
\[
\operatorname{div}(\epsilon _0 u)+ p_0 \operatorname{div}u=0 \quad\text{and}\quad \frac{1}{c^2} (  \epsilon _0 + p_0 )\nabla _uu + \mathsf{P} \nabla p_0 =q \rho  \mathbf{i} _u F,
\]
with $ \epsilon _0= \rho(c^2+ e(\rho, \eta))$ and $p_0= \rho^2\frac{\partial e}{\partial \rho}$ the total energy density and pressure of the fluid part of the system.

\subsubsection{Coupling with the Einstein equations and junction conditions}\label{subsub_coupling_fluid}

This coupling is obtained by applying Theorem \ref{main} with $\ell$ the Lagrangian density of the electromagnetic fluid given in \eqref{Lagrangian_epsilon} (or its equivalent form \eqref{Lagrangian_e_F}) and $\ell_{\rm M}$ the Maxwell Lagrangian. As shown in \eqref{full_system} this principles gives the Einstein and Maxwell equations in $ \mathcal{N} ^-$ and $ \mathcal{N} ^+$, i.e.,
\begin{equation}\label{full_system_fluid}
\left\{
\begin{array}{l}
\displaystyle\vspace{0.2cm} Ein (\mathsf{g}^-) \mu (\mathsf{g}^-) =  \chi \mathfrak{T}^- \quad \text{on} \quad    \mathcal{N}^- \qquad \quad \;\;   Ein(\mathsf{g}^+) = \chi \mathfrak{T} _{\rm M}^+ \quad   \text{on} \quad    \mathcal{N}^+\\
\displaystyle \delta \Big(\frac{\partial \epsilon \;\;}{\partial F^-}\Big)^\flat  +  \rho q  u^\flat=0\quad   \text{on} \quad    \mathcal{N}^-\qquad \quad \;\;\;    \delta\frac{\partial \epsilon_{\rm_M}}{\partial F^+}=0 \quad   \text{on} \quad    \mathcal{N}^+,
\end{array}
\right.
\end{equation}
together with their junction conditions at the boundary of the fluid. Here $ \mathfrak{T} ^-$ is the stress-energy-momentum tensor found in \eqref{T_EMF} and $ \mathfrak{T}_{\rm M}^+$ is given in \eqref{general_Maxwell_SEM}-\eqref{Maxwell_T}   Using the Gauss-Codazzi
equations on $\partial_{\rm cont} \mathcal{N}$ the Israel-Darmois junction conditions in \eqref{full_system} yield the O'Brien-Synge conditions $[ Ein( \cdot , n)]=0$. From the Einstein equations on both sides of $ \partial _{\rm cont} \mathcal{N} $ in \eqref{full_system_fluid} and from the expression of $ \mathfrak{T} ^-$ and $ \mathfrak{T} ^+_{\rm M}$, we get
\begin{equation}\label{interm_BC} 
\frac{1}{c^2} \mathbf{i} _n[S] u^\flat - [ \mathfrak{t}] ( \cdot , n ^\flat )+ [p] n ^\flat =0,
\end{equation}
where $[ \cdot ]= (\cdot) ^+ - (\cdot )^-$, with $S^+= (-1)^n \mathbf{i} _{E^+} \mathbf{i} _u (\star B^+)$, $S^-=(-1)^n \mathbf{i} _{E^-} \mathbf{i} _u (\star H ^-)$, $ \mathfrak{t}^+= E ^+  \otimes E^+ + B ^+   \stackrel{\rm tr}\otimes B^+$, $ \mathfrak{t}^- = \mathsf{t}_{\rm em} = D^-   \otimes E^- + B ^-  \stackrel{\rm tr}{\otimes} H^-$, $p^+= \frac{1}{2} (|E^+| ^2 + |B^+| ^2 )$, $p^-= p$. Applying $ \mathbf{i} _u$ on \eqref{interm_BC} yields $ \mathbf{i} _n[S]=0$, so that we are left with 
\begin{equation}\label{BC_EM_0}
[ \mathfrak{t}] ( \cdot , n ^\flat )+ [p] n ^\flat =0.
\end{equation} 
Condition $ i_{ \partial _{\rm cont} \mathcal{N}} ^* [F]=0$ is readily seen to be equivalent to $ \mathbf{i} _n (\star [F])=0$, thereby yielding
\begin{equation}\label{BC_EM_1} 
\mathbf{i} _n \mathbf{i} _u (\star [E])=0 \quad\text{and}\quad \mathbf{i} _n[B]=0,
\end{equation} 
while $  i_{ \partial _{\rm cont} \mathcal{N}} ^*  [ \frac{\partialnew \ell}{\partialnew F} ]=0$ is equivalent to $ \mathbf{i} _{ n ^\flat } [\frac{\partial \ell}{\partial F}]=0$, yielding 
\begin{equation}\label{BC_EM_2} 
\mathbf{i} _n \mathbf{i} _u (\star [H])=0 \quad\text{and}\quad \mathbf{i} _n[D]=0.
\end{equation} 
From the first condition in \eqref{BC_EM_2}, we have $ \mathbf{i} _n [S]= -(-1)^n \mathbf{i} _{[E]} \mathbf{i} _n \mathbf{i} _u(\star H^-)$. This expression vanishes if the first condition in \eqref{BC_EM_1} holds. Therefore, the boundary condition $ \mathbf{i} _n[S]=0$ obtained earlier doesn't need to be imposed since it follows from \eqref{BC_EM_1} and \eqref{BC_EM_2}.
Also, using \eqref{BC_EM_1} and \eqref{BC_EM_2}, we can further write $[ \mathfrak{t}]( \cdot, n ^\flat )$ in \eqref{BC_EM_0} as $\mathbf{i} _n D ^- [E] + \mathbf{i} _n B^- \!\!:\![\mathbf{i} _{\_\,}H]$.
In conclusion, we get the following result regarding the variational derivation of the equations and junction conditions for electromagnetic fluids coupled to gravity.

\begin{proposition}\label{GREMF} The system of equations for relativistic electromagnetic fluids obtained in \eqref{full_system} by the variational principle \eqref{total_action_variation} yield the equations \eqref{equations_EMF} (equivalently written as \eqref{equations_EMF_2}) and  \eqref{cont_eq_EMF} on $ \mathcal{N} ^-$, as well as \eqref{full_system_fluid}, together with the boundary and junction conditions
\begin{equation}\label{BC_EMF_2} 
\begin{aligned} 
&\mathsf{g} (u , n) =0, \qquad & &  [ \mathfrak{t}] ( \cdot , n ^\flat )+ [p] n ^\flat =0& & \\
&\mathbf{i} _n \mathbf{i} _u (\star [E])=0, \qquad & & \mathbf{i} _n [B]=0 \qquad& & \qquad \text{on}\qquad \partial _{\rm cont} \mathcal{N}.\\
&\mathbf{i} _n [D]=0, \qquad & & \mathbf{i} _n \mathbf{i} _u (\star [H])=0& & 
\end{aligned}
\end{equation} 

\end{proposition}

\subsection{General relativistic electromagnetic elasticity}\label{sec_GREE}

The treatment of relativistic electromagnetic elasticity follows in almost a straightforward way by combining the treatment of electromagnetic fluids made above together with the treatment of relativistic elasticity in \cite{GB2024}. Moreover, we can directly consider a general setting of relativistic continuum mechanics that encompass both fluid and elasticity, see \cite[\S6.3]{GB2024}. We shall therefore only mention the results while leaving the details to the interested reader.
The description of relativistic elasticity requires the introduction of an additional material tensor, a Riemannian metric $G_0$ on $\mathcal{B}$. From it, the $2$-covariant symmetric positive tensor $G=\pi_\mathcal{B}^*G_0$ is defined on $\mathcal{D}= [\lambda_0,\lambda_1]\times\mathcal{B}$. This tensor $G$ allows the definition of the relativistic Cauchy deformation tensor $\mathsf{c}$ defined by
\[
\mathsf{c}=\Phi_*G,
\] 
see \cite{GB2024} for the link with the original definitions in \cite{GrEr1966a,ErMa1990,Ma1978a}.

\subsubsection{Lagrangian variational formulation}

The Lagrangian to be used in Hamilton's principle originates from the same systematic construction as the one for electromagnetic fluids detailed in \S\ref{Lagrangian_setps}. One gets a Lagrangian density with the additional reference variable $G$, i.e., \begin{equation}\label{L_EM_elasticity}
\mathscr{L}( j^1 \Phi , \mathcal{A} , {\rm d} \mathcal{A} ,\partial _ \lambda , R,S,G,\mathsf{g} \circ \Phi ),
\end{equation}
still given by \eqref{L_EMF} in which the ``energy" expression $\mathscr{E} \big( \cdot , \cdot  , \mathcal{E} ,  \mathcal{B} , C\big)$ is replaced by $\mathcal{W} \big( \cdot , \cdot  , \mathcal{E} ,  \mathcal{B} , G_0,C\big)$. This Lagrangian density is spacetime covariant, with associated convected Lagrangian given as $ \mathcal{L} (  \mathcal{A} , {\rm d} \mathcal{A}, \partial _ \lambda, R,S,G,\Gamma )$ whose expression can be easily found. Furthermore, \eqref{L_EM_elasticity} is materially covariant with respect to $ \operatorname{Diff}_{\partial _ \lambda }( \mathcal{D} )$ if and only if the function $ \mathcal{W} $ satisfies the isotropy condition
\begin{equation}\label{isotropy_E} 
\mathcal{W}( \rho  \circ \psi  , \eta \circ\psi , \psi ^* \mathcal{E} , \psi ^* \mathcal{B} , \psi^*G_0,\psi ^* C)= \mathcal{W}( \rho  \circ \psi  , \eta , \mathcal{E} , \mathcal{B} , G_0,C) \circ \psi ,
\end{equation}
for all $ \psi \in \operatorname{Diff}( \mathcal{B} )$. In this case, one gets the spacetime Lagrangian density
\begin{equation}\label{ell_continua}
\begin{aligned}
\!\!\ell(  A, {\rm d}A, w, \varrho , \varsigma , \mathsf{c},\mathsf{g})&= - \rho  \left(  c^2 + \varpi( \rho , \eta  , E,B, \mathsf{c}, \mathsf{p})  + q A \cdot u \right)  \mu (\mathsf{g})- \frac{1}{2} F\wedge\star F\\
&=-\epsilon( \rho , \eta  , E,B, \mathsf{c}, \mathsf{g}) \mu (\mathsf{g}) - q\varrho A\cdot w\\
&=-\mathfrak{e}( \rho , s,u  ,F, \mathsf{c}, \mathsf{g}) \mu (\mathsf{g}) - q\varrho A\cdot w,
\end{aligned}
\end{equation}
where the spacetime ``energy" function $\varpi$ is
\begin{equation}\label{def_varpi} 
\varpi (\rho , \eta , E,B, \mathsf{c}, \mathsf{p})= \mathcal{W}( \rho \circ \Phi  , \eta \circ \Phi ,  \Phi ^*E, \Phi ^*B, \Phi ^*  \mathsf{c} , \Phi ^*  \mathsf{p} ) \circ \Phi ^{-1},
\end{equation}
and we introduced the densities $\varepsilon$ and $\mathfrak{e}$ as earlier for fluids in \eqref{Lagrangian_epsilon} and \eqref{simpler_F}, now with dependence on $\mathsf{c}$.
The dependence of $\varpi$ on $\mathsf{p}$ can be replaced by a dependence on $\mathsf{g}$, because of the properties $ \mathbf{i} _uE=0$, $ \mathbf{i} _uB=0$, and $\mathbf{i}_u\mathsf{c}=0$, see also Remark \ref{GB2024_VS_2025}.

The general results stated in Theorem \ref{spacetime_reduced_EL} and \ref{convective_reduced_EL} are directly applicable to this Lagrangian. In particular, we get the following Eulerian variational principle and spacetime reduced Euler-Lagrange equations by direct application of \eqref{Eulerian_VP}. The result extends Theorem \ref{prop_Eulerian_VP_fluid} to continua as follows:

\begin{proposition}[Variational formulation for electromagnetic continua]\label{prop_Eulerian_VP_continua} The Eulerian variational formulation for the relativistic electromagnetic continua takes the form\color{black} 
\begin{equation}\label{Eulerian_VP_continua}
\begin{aligned} 
&\!\!\left. \frac{d}{d\varepsilon}\right|_{\varepsilon=0}\int_{ \mathcal{N}_ \varepsilon } \ell\big( A_ \varepsilon , {\rm d} A_ \varepsilon , w_ \varepsilon , \varrho  _ \varepsilon , \varsigma _ \varepsilon , \mathsf{c}_\varepsilon,\mathsf{g} \big)=0 \quad \text{for variations}\\
& \delta \mathcal{N} = \zeta |_{ \partial \mathcal{N} } \big/T \partial \mathcal{N} , \;\; \delta A = - \pounds _ \zeta A + 
\deltabar A,\;\; \delta w = - \pounds _ \zeta w\phantom{\int_A}\\
& \delta \varrho = - \pounds _ \zeta \varrho, \;\;   \delta \varsigma = - \pounds _ \zeta \varsigma, \;\;   \delta \mathsf{c} = - \pounds _ \zeta \mathsf{c},
\end{aligned} 
\end{equation}\color{black} 
where $ \zeta$ is an arbitrary vector field on $ \mathcal{N} $ such that $ \zeta |_{ \Phi ( \lambda _0, \mathcal{B} )}= \zeta |_{ \Phi ( \lambda _1, \mathcal{B} )}=0$ and $\deltabar A$ is an arbitrary one-form on $ \mathcal{N} $ such that $\deltabar A|_{ \Phi ( \lambda _0, \mathcal{B} )}= \deltabar A |_{ \Phi ( \lambda _1, \mathcal{B} )}=0$.

The critical conditions associated to \eqref{Eulerian_VP_fluid} are
\begin{equation}\label{spacetime_EL_continua} 
\!\!\!\!\!\!\!\left\{
\begin{array}{l}
\displaystyle\vspace{0.2cm}\!\!\operatorname{div}^ \nabla \!\Big( \Big(  \ell - \varrho \frac{\partial \ell}{\partial \varrho }- \varsigma  \frac{\partial \ell}{\partial \varsigma  } \Big)  \delta  + w \otimes \frac{\partial \ell}{\partial w} - 2  \frac{\partial \ell}{\partial \mathsf{c}}\cdot \mathsf{c}-  \frac{\partial \ell}{\partial A } \otimes  A  -  \frac{\partial \ell}{\partial F }\stackrel{\rm tr}{ \otimes }F\Big)=0\\
\displaystyle\vspace{0.2cm}\!\! \frac{\partialnew \ell}{\partialnew A} + {\rm d}  \frac{\partialnew \ell}{\partialnew F}=0, \qquad  \qquad i^*_{ \partial _{\rm cont}\mathcal{N} } \frac{\partialnew \ell}{\partialnew F}=0\\
\displaystyle\vspace{0.2cm} \!\!\Big( \!\Big(  \ell - \varrho \frac{\partial \ell}{\partial \varrho }- \varsigma  \frac{\partial \ell}{\partial \varsigma  } \Big)  \delta  + w \otimes \frac{\partial \ell}{\partial w} - 2  \frac{\partial \ell}{\partial \mathsf{c}}\cdot \mathsf{c}-  \frac{\partial \ell}{\partial A } \otimes  A  -  \frac{\partial \ell}{\partial F }\stackrel{\rm tr}{ \otimes }F\Big) ( \cdot ,n ^\flat )=0\\
\!\!\quad\text{on $\textcolor{black}{\partial_{\rm cont} \mathcal{N}}$}\hspace{-0.5cm}
\end{array}
\right.
\end{equation}
and the variables $w$, $ \varrho $, $ \varsigma $ satisfy
\begin{equation}\label{advections_cont} 
\pounds _w \varrho =0,\quad  \pounds _w \varsigma  =0 ,\quad\text{and}\quad \pounds _w \mathsf{c}  =0.
\end{equation} 
\end{proposition}

\medskip 

We note that equations \eqref{advections_fluid} are equivalently written in terms of the world-velocity, the proper mass and entropy density as
\begin{equation}\label{cont_eq_EMC}
\pounds _u (\rho  \mu(g))  =0,\quad \pounds _u (s  \mu(g))   =0 ,\quad\text{and}\quad \pounds _u \mathsf{c}  =0.
\end{equation}

\subsubsection{Stress-energy-momentum tensor for electromagnetic continua}\label{SEM_el}

The expression of the stress-energy-momentum tensor is obtained as in Proposition \ref{end_of_controversy}, now taking account of the dependence on the relativistic Cauchy deformation tensor $\mathsf{c}$. We consider the various expressions of the Lagrangian density as given in \eqref{ell_continua}.

\begin{proposition}[Stress-energy-momentum tensor for electromagnetic continua]\label{end_of_controversy_continua} Three equivalent formulations of the stress-energy-momentum tensor for electromagnetic continua are obtained by adding the term $-\mathfrak{t}_{\rm el}\,\mu(\mathsf{g})$, where
\[
\mathfrak{t}_{\rm el}\mu(\mathsf{g})= 2  \frac{\partial \ell}{\partial \mathsf{c}}\!\cdot\!\mathsf{c}, \qquad 
\mathfrak{t}_{\rm el}=-2 \frac{\partial\epsilon}{\partial \mathsf{c}}\!\cdot \! \mathsf{c} ,\qquad\text{and} \qquad \mathfrak{t}_{\rm el}=-2 \frac{\partial\mathfrak{e}}{\partial \mathsf{c}}\!\cdot  \!\mathsf{c} ,
\]
to the expressions given in \eqref{T_1}, \eqref{T_2}, and \eqref{T_3}, respectively.
\end{proposition}

\begin{remark}[]\label{GB2024_VS_2025}\rm In terms of the function $\varpi (\rho , \eta , E,B, \mathsf{c}, \mathsf{g})$, we can write
\[
\mathfrak{t}_{\rm el}=-2\rho \frac{\partial\varpi}{\partial \mathsf{c}}\cdot  \mathsf{c} ,\qquad (\mathfrak{t}_{\rm el})^\mu_\nu = - 2\rho \frac{\partial\varpi}{\partial \mathsf{c}_{\lambda\mu}}\mathsf{c}_{\lambda\nu},
\]
where we recall that either $\mathsf{g}$ or $\mathsf{p}$ can be used in $\varpi $. Using $\mathsf{p}$ introduces an additional dependence on $w$ (since $\mathsf{p}$ depends on both $\mathsf{g}$ and $w$) when deriving the expression of $\mathfrak{T}$ via its general expression in terms of $\ell$, see \eqref{spacetime_EL_continua}. This approach leads to the seemingly different expression
\[
\mathfrak{t}_{\rm el}=-2\rho \mathsf{P}\cdot \frac{\partial\varpi}{\partial \mathsf{c}}\cdot  \mathsf{c} \mu(\mathsf{g}),\qquad (\mathfrak{t}_{\rm el})^\mu_\nu = -2\rho \mathsf{P}^\mu_\alpha\frac{\partial\varpi}{\partial \mathsf{c}_{\lambda\alpha}}\mathsf{c}_{\lambda\nu} \mu(\mathsf{g}),
\]
as obtained in \cite[\S6.5]{GB2024}. However, as expected, these expressions are equivalent, because $\mathsf{P}^\mu_\alpha\frac{\partial\varpi}{\partial \mathsf{c}_{\lambda\alpha}}=\frac{\partial\varpi}{\partial \mathsf{c}_{\lambda\mu}}$. This identity follows directly from the chain rule applied to $\varpi$ under the identity $\mathsf{P}^\alpha_\lambda\mathsf{c}_{\alpha\beta}=\mathsf{c}_{\lambda\beta}$.
\end{remark}

The examples discussed in \S\ref{SEM_fluid} for fluids have direct analogues in the broader framework considered here. The simplest such example is the extension of the Euler-Maxwell system to general continua, described by the Lagrangian density:
\begin{equation}\label{Cauchy_Maxwell_ell}
\ell( A, {\rm d} A, w, \varrho , \varsigma , \mathsf{c},\mathsf{g} ) = - \epsilon _0( \rho  , s,\mathsf{c},\mathsf{g})- q \rho  A\! \cdot  \! u \mu ( \mathsf{g} ) - \mathfrak{e}_{\rm M}(F, \mathsf{g} ).
\end{equation}
Here, $\epsilon _0( \rho  , s,\mathsf{c},\mathsf{g})=\rho(c^2+\varpi(\rho  , \eta,\mathsf{c},\mathsf{g}))$ is the energy density of the matter, which is assumed not to induce polarization or magnetization effects. The associated stress-energy-momentum tensor is:
\[
\mathfrak{T} = \Big[\epsilon _0 \frac{1}{c^2}u \otimes u ^\flat + p_0\mathsf{P} + 2\frac{\partial\epsilon_0}{\partial\mathsf{c}}\cdot\mathsf{c} + \Big( \frac{\partial \mathfrak{e}_{\rm M}}{\partial F} \stackrel{\rm tr}{ \otimes }F - \mathfrak{e}_{\rm M}  \delta\Big) \mathsf{P} \Big]\mu(\mathsf{g}),
\]
which can be decomposed as $ \mathfrak{T}=\mathfrak{T} _{\rm cont} + \mathfrak{T} _{\rm M}$, where $\mathfrak{T} _{\rm cont}$ represents the contribution from the continuum mechanics (\cite[p.52]{GB2024}) with $p_0= \rho^2\frac{\partial\varpi}{\partial\rho}$, and $ \mathfrak{T} _{\rm M}$ is the Maxwell stress-energy-momentum tensor.

\subsubsection{Equations for electromagnetic continua}\label{equ_el}

Using Proposition \ref{end_of_controversy_continua}, we can directly extend the equations \eqref{equations_EMF} and \eqref{equations_EMF_2}, along with Proposition \ref{Prop_EMF}, to encompass electromagnetic continua.

\paragraph{Relativistic electromagnetic continua equations I.} We obtain the energy and momentum equations:
\begin{equation}\label{equations_EMC} 
\!\!\!\begin{aligned} 
& \operatorname{div}( \epsilon_{\rm tot} u+ S) + \frac{1}{c^2} g( S, \nabla _uu) = \mathfrak{t}_{\rm ec}: \nabla u -p \operatorname{div}u \\
&\frac{1}{c^2} (\epsilon_{\rm tot}+ p) \nabla _u u + \frac{1}{c^2} \mathsf{P} \operatorname{div}( S \otimes u + u \otimes S)+\frac{1}{c^2} u\mathfrak{t}_{\rm ec}: \nabla u = \operatorname{div}\mathfrak{t}_{\rm ec}  - \mathsf{P} \nabla p,
\end{aligned} 
\end{equation} 
where 
$\epsilon_{\rm tot}$ and $p$ are as defined earlier. The electromagnetic continuum (``ec\,") stress tensor, $\mathfrak{t}_{\rm ec}$, is composed of its electromagnetic and elastic contributions:
\begin{align*}
\mathfrak{t}_{\rm ec}&=  -\frac{\partial \epsilon    }{\partial E  } \otimes E + B ^\sharp  \stackrel{\rm tr}{ \otimes }  \frac{\partial \epsilon    }{\partial B  }^\flat -2 \frac{\partial\epsilon}{\partial \mathsf{c}}\!\cdot  \!\mathsf{c} \\
&=D ^\sharp  \otimes E + B ^\sharp  \otimes H- 2 \rho \frac{\partial \varpi}{\partial \mathsf{c}} \!\cdot\! \mathsf{c}=\mathfrak{t}_{\rm em}+\mathfrak{t}_{\rm el}.
\end{align*}

\begin{proposition}\label{Prop_EMC} The system of equations for relativistic electromagnetic continua obtained in \eqref{spacetime_EL_continua} by the variational principle \eqref{Eulerian_VP_continua} yield the equations \eqref{equations_EMC}, \eqref{cont_eq_EMC} and \eqref{Maxwell_E_B}, together with the boundary conditions
\begin{equation}\label{BC_EMF_1} 
\mathsf{g} (u , n) =0, \quad \mathfrak{t}_{\rm ec}( \cdot , n ^\flat )= p n ^\flat, \quad \mathbf{i} _n D=0, \quad \mathbf{i} _n \mathbf{i} _u (\star H)=0 \quad\text{on}\quad \partial _{\rm cont} \mathcal{N}.
\end{equation} 
\end{proposition}  

\paragraph{Relativistic electromagnetic fluid equations II.} To highlight the ponderomotive four-force, we express $\mathfrak{T}$ in terms of $\mathfrak{e}$ and decompose it as $\mathfrak{T}= \mathfrak{T} _{\rm c}+ \mathfrak{T} _{\rm F}$, with
\begin{align*}
\mathfrak{T} _{\rm c}&=\Big[\mathfrak{e} _{\rm m} \frac{1}{c^2}u \otimes u ^\flat+   \Big(   \frac{\partial \mathfrak{e} _{\rm m} }{\partial \rho  } \rho+  \frac{\partial \mathfrak{e} _{\rm m} }{\partial s} s- \mathfrak{e}   _{\rm m} \Big) \mathsf{P}+ 2  \frac{\partial \mathfrak{e} _{\rm m} }{\partial \mathsf{c}} \cdot \mathsf{c}\Big]\mu(\mathsf{g})\\
\mathfrak{T} _{\rm F}&=  \Big[ - u \otimes \frac{\partial \mathfrak{e} _{\rm m} }{\partial u} \cdot \mathsf{P} + \frac{\partial \mathfrak{e}  }{\partial F}   \stackrel{\rm tr}{\otimes} F -  \mathfrak{e}_{\rm M}  \delta\Big]\mu(\mathsf{g}).
\end{align*}
Following similar steps as before, the equation of motion is given by:
\begin{equation}\label{div_f_form_cont} 
\operatorname{div} \mathfrak{T} _{\rm c}= \mathfrak{f}, \qquad \mathfrak{f}=  \operatorname{div}( u \otimes \omega _{\rm m}) -\frac{\partial \mathfrak{e}_{\rm m}}{\partial F} : \nabla F + q \rho  \mathbf{i} _u F \mu ( \mathsf{g} ).
\end{equation} 
From this alternative form, the energy and momentum conservation equations \eqref{equations_EMC} can be written as: 
\begin{equation}\label{equations_EMC_2}
\begin{aligned}  
&- \operatorname{div}( \mathfrak{e}_{\rm m} u) - p_{\rm m} \operatorname{div}u  +\mathfrak{t}_{\rm el}:\nabla u\\
&\qquad\qquad= \mathfrak{f} \cdot u = - \omega ( \nabla _u u) - \frac{\partial \mathfrak{e}_{\rm m}}{\partial F} : \nabla _uF\\
&\frac{1}{c^2} (( \mathfrak{e}_{\rm m} + p_{\rm m} )\nabla _uu + u\mathfrak{t}_{\rm el}:\nabla u) + \mathsf{P} \nabla p_{\rm m}  -\operatorname{div}\mathfrak{t}_{\rm el}\\
&\qquad\qquad= \mathsf{P} \mathfrak{f} =\mathsf{P}\Big( \rho  \nabla _u ( \omega / \rho  ) -\frac{\partial \mathfrak{e}_{\rm m}}{\partial F} : \nabla F\Big)+ q \rho  \mathbf{i} _u F,
\end{aligned} 
\end{equation}
Equations \eqref{div_f_form_cont} and \eqref{equations_EMC_2} should be compared with \eqref{div_f_form} and \eqref{equations_EMF_2}, as they illustrate the role of the elastic stress $\mathfrak{t}_{\rm el}$, which arises from the dependence on the relativistic Cauchy deformation tensor $\mathsf{c}$.

In the absence of magnetization or polarization (analogous to Euler-Maxwell), and using equation \eqref{Cauchy_Maxwell_ell}, the system \eqref{equations_EMC_2} simplifies to:
\begin{align*}
&\operatorname{div}(\epsilon _0 u)+ p_0 \operatorname{div}u -\mathfrak{t}_{\rm el}:\nabla u=0\\
&\frac{1}{c^2}( ( \epsilon _0 + p_0 )\nabla _uu + u\mathfrak{t}_{\rm el}:\nabla u) + \mathsf{P} \nabla p_0 -\operatorname{div}\mathfrak{t}_{\rm el} =q \rho  \mathbf{i} _u F,
\end{align*}
with $\epsilon _0= \rho(c^2+ \varpi(\rho, \eta,\mathsf{c},\mathsf{g}))$.

\subsubsection{Coupling with the Einstein equations and junction conditions}

Building on the previous developments, Proposition  \ref{GREMF} can now be directly formulated for continua. The only modification required is to replace the stress tensor
$\mathfrak{t}^-=\mathfrak{t}_{\rm em}$
  with 
$\mathfrak{t}^-=\mathfrak{t}_{\rm em}+\mathfrak{t}_{\rm el}$. This modification affects both the interior and boundary conditions, where 
$[\mathfrak{t}]= \mathfrak{t}^+-\mathfrak{t}^-$.

\begin{proposition}\label{GREMC} The system of equations for relativistic electromagnetic continua obtained in \eqref{full_system} by the variational principle \eqref{total_action_variation} yield the equations \eqref{equations_EMC} (equivalently written as \eqref{equations_EMC_2}) and  \eqref{cont_eq_EMC} on $ \mathcal{N} ^-$, as well as \eqref{full_system_fluid}, together with the boundary and junction conditions
\begin{equation}\label{BC_EMF_2} 
\begin{aligned} 
&\mathsf{g} (u , n) =0, \qquad & &  [ \mathfrak{t}] ( \cdot , n ^\flat )+ [p] n ^\flat =0& & \\
&\mathbf{i} _n \mathbf{i} _u (\star [E])=0, \qquad & & \mathbf{i} _n [B]=0 \qquad& & \qquad \text{on}\qquad \partial _{\rm cont} \mathcal{N}.\\
&\mathbf{i} _n [D]=0, \qquad & & \mathbf{i} _n \mathbf{i} _u (\star [H])=0& & 
\end{aligned}
\end{equation} 
\end{proposition}

\subsubsection{Extensions}\label{anisotropic}

\paragraph{Extension to anisotropic continua.} Appropriate choices of the ``energy" function $\mathcal{W} \big( \rho , \eta  , \mathcal{E} ,  \mathcal{B} , G_0,C\big)$ appearing in the Lagrangian $\mathscr{L}$ of electromagnetic continua allow for the treatment of both isotropic and anisotropic media. As discussed earlier, the \textit{spacetime covariance} of $\mathscr{L}$ is guaranteed for any such expression. However, it is only in the isotropic case, i.e., when \textit{material covariance}, as expressed in \eqref{isotropy_E}, holds, that $\mathcal{W}$ can be described by a corresponding spacetime ``energy" function $\varpi (\rho , \eta , E,B, \mathsf{c}, \mathsf{p})$; see \eqref{def_varpi}. In this scenario, the equations admit a spacetime description and a variational derivation in the form given in \S\ref{sec_GREE}, with the spacetime Lagrangian density $\ell(A, {\rm d}A,\varrho, \varsigma,\mathsf{c},\mathsf{g})$.

As noted in \cite[Remark 6.5]{GB2024}, \textit{this spacetime variational formulation is also applicable to anisotropic continua, provided additional material tensor fields are included}. Specifically, consider an anisotropic electromagnetic material where $\mathcal{W} \big( \rho , \eta  , \mathcal{E} ,  \mathcal{B} , G_0,C\big)$ has a symmetry group that is a proper subgroup of the $G_0(\mathsf{X})$-orthogonal transformations, defined as the isotropy group of the so called structural (or anisotropic) tensors $\{\alpha^K\}_{K=1}^n$ on $\mathcal{B}$. For such materials, a generalized notion of material covariance can be achieved by writing the energy as function $\mathcal{W} \big( \rho , \eta  , \mathcal{E} ,  \mathcal{B} , G_0,\{\alpha^K\},C\big)$, while allowing the diffeomorphism group $\operatorname{Diff}(\mathcal{B})$ to act by pull-back on the structural tensors when verifying \eqref{isotropy_E}. This approach enables the definition of a corresponding spacetime ``energy" function $\varpi$, which incorporates the spacetime counterparts of these tensors. With this formulation, the general variational framework developed in \S\ref{sec_2} directly applies to the anisotropic case by including $\{\alpha^K\}$ among the material tensor fields of the theory, alongside $\partial_\lambda$, $R$, $S$, and $G$ (which are instances of $K$ of the abstract theory in \S\ref{sec_2}).

For more information and applications of structural tensors and anisotropy in relation to material covariance in the non-relativistic case, we refer to \cite{Li1982,Bo1987,SiMaKr1988,ZhSp1993,LuPa2000,MaRa2006,SoYa2020}, and references therein. Earlier treatments of isotropic and anisotropic forms of energy functions in electromagnetic media using invariants can be found in \cite{GrEr1966b,ErMa1990}.

\paragraph{Extension to nonlinear theories: Born-Infeld and Euler-Heisenberg Lagrangians.} The extension of the present variational approach to nonlinear electrodynamic theories is straightforward. One replace the Maxwell Lagrangian density $\ell_{\rm M}$, by more general expressions of the form
\[
\ell_{\rm nl}(A, F, \mathsf{g} )= - \epsilon _{\rm nl} \big( \frac{1}{2} F \wedge \star F, \frac{1}{2} F \wedge F, \mathsf{g} \big),
\]
where here $n=3$. Appropriate choice of the nonlinear density $ \epsilon _{\rm nl} $ in terms of the invariants $F \wedge \star F$ and $ F \wedge F$ yields the Born-Infeld or Euler-Heisenberg electrodynamic theories, \cite{BoIn1934}, \cite{HeEu1936}. We refer to \cite{BaLeSoTo2020,So2022} and references therein for recent development in nonlinear electrodynamics. We leave to the reader the generalization of Propositions \ref{end_of_controversy}, \ref{Prop_EMF},  \ref{GREMF}, \ref{Prop_EMC}, and \ref{GREMC} to this case, which are obtained by replacing everywhere $ \epsilon _{\rm M}$ by $ \epsilon _{\rm nl}$, and $ \mathfrak{T} _{\rm M}$ by $ \mathfrak{T} _{\rm nl}$ found from \eqref{general_Maxwell_SEM} as being given by
\[
\mathfrak{T} _{\rm nl}= \ell_{\rm nl} \delta  - \frac{\partial \ell_{\rm nl}}{\partial F} \stackrel{\rm tr} \otimes F = \frac{\partialnew  \epsilon _{\rm nl} }{\partialnew  \alpha } \mathfrak{T} _{\rm M} + \left(  \frac{\partialnew  \epsilon _{\rm nl} }{\partialnew \alpha }  \alpha    + \frac{\partialnew \epsilon _{\rm nl}}{\partialnew \beta } \beta  - \epsilon _{\rm nl} \right)  \delta
\]
with $ \alpha = \frac{1}{2}  F \wedge \star F$ and $\beta = \frac{1}{2} F \wedge  F$.

\subsection{Gauge invariance of electromagnetic continua}\label{gauge_continua}

To verify the gauge invariance \eqref{gauge_invariance} for electromagnetic continua, we consider the Lagrangian density $\mathscr{L}$ in \eqref{L_EM_elasticity} given by the expression  \eqref{L_EMF}. In this expression, the ``energy" term $\mathscr{E} \big( \cdot , \cdot  , \mathcal{E} ,  \mathcal{B} , C\big)$ for fluids is replaced by $\mathcal{W} \big( \cdot , \cdot  , \mathcal{E} ,  \mathcal{B} , G_0,C\big)$ for general continua. As it is readily verified, the corresponding action functional satisfies:
\[
\mathscr{S}( \Phi ,\mathcal{A}+ {\rm d} f)- \mathscr{S}( \Phi ,\mathcal{A}) = -q\int_\mathcal{B} \big(f(\lambda_1,\mathsf{X}) - f(\lambda_0,\mathsf{X}) \big)R_0,
\]
for all $f\in C^\infty(\mathcal{D})$. Gauge invariance holds since $C$ in \eqref{gauge_invariance} does not depend on quantities that are varied in Hamilton's principle. This result remains valid when $A$ is treated as the primary field (see \S\ref{material_spacetime_A}). In this case, we have
\[
\mathscr{S}'( \Phi ,A+ {\rm d} k)- \mathscr{S}'( \Phi ,A) = -q\int_\mathcal{B} \big(k(\Phi(\lambda_1,\mathsf{X})) - k(\Phi(\lambda_0,\mathsf{X})) \big)R_0,
\]
for all $k\in C^\infty(\mathcal{M})$, demonstrating gauge invariance again, as $\Phi(\lambda_0,\mathsf{X})$ and $\Phi(\lambda_1,\mathsf{X})$ are held fixed in Hamilton's principle (see \eqref{GI_A}).
Gauge invariance in the convective and spacetime descriptions can can also be verified straightforwardly.

Thus, as previously noted, the invariance of the theory of electromagnetic continua is characterized by either of the groups $\left( \operatorname{Diff}( \mathcal{D} ) \,\circledS\, C^\infty( \mathcal{D} ) \right)  \times \operatorname{Diff}( \mathcal{M} )$ or $ \operatorname{Diff}( \mathcal{D} )\times \left(\operatorname{Diff}( \mathcal{M} ) \,\circledS\, C^\infty( \mathcal{M} ) \right)$. This invariance has significant implications, which will be analyzed in subsequent parts of this paper.

\section{Conclusion}\label{conclusion}

In this work, we have developed a unified variational framework for relativistic electromagnetic continua that simultaneously incorporates material, spacetime, and convective descriptions. Building directly on Hamilton’s principle and the continuum analogue of the action of a charged particle, the formulation provides a transparent and systematic approach to deriving the field equations, stress-energy-momentum tensors, and junction conditions relevant for relativistic fluids and solids interacting with electromagnetic and gravitational fields. By avoiding auxiliary variables or Lagrange multipliers and relying exclusively on the freely varied primary fields - the spacetime configuration of the continuum and the electromagnetic potential - the theory maintains both conceptual clarity and strong physical grounding.

A central outcome of our analysis is the general expression for the stress-energy-momentum tensor arising from arbitrary covariant couplings between elastic deformation and electromagnetic fields. This tensor is obtained without resorting to ad-hoc decompositions or constitutive assumptions beyond covariance and the specification of an energy density. In doing so, the framework places various classical forms of the electromagnetic stress-energy-momentum tensor within a single geometrically consistent picture and clarifies the origin of different splittings that have appeared in the literature. Likewise, the resulting covariant Euler-type balance equations equations admit several equivalent formulations, which we systematically compare and relate to traditional presentations in the theory of relativistic electromagnetic media.
The inclusion of Gibbons-Hawking-York terms in the action leads naturally to the Israel-Darmois and electromagnetic junction conditions governing the matching between interior solutions for the continuum and exterior Einstein-Maxwell fields. This provides a coherent variational route to interface conditions that are usually introduced through geometric or distributional arguments. 
Beyond its conceptual appeal, the framework is directly relevant to a number of astrophysical scenarios in which strong electromagnetic fields interact with relativistic matter, including neutron-star crusts, magnetar flares, magnetized accretion flows, and the formation of relativistic jets. Because the formalism accommodates general constitutive laws - including anisotropic and magneto-elastic responses - it provides a foundation for future modeling efforts in these extreme environments.







{\footnotesize

\bibliographystyle{new}

}

\end{document}